\documentclass[twoside,11pt]{article}

\usepackage{blindtext}

\usepackage{jmlr2e}

\usepackage{amssymb}
\usepackage{tikz}
\usepackage{xr}
\usetikzlibrary{arrows, calc, decorations.pathmorphing, backgrounds, positioning, fit, petri, automata}\usepackage{algorithmicx}
\usepackage{algcompatible}
\usepackage{algorithm}
\usepackage[shortlabels]{enumitem}
\usepackage{hyperref}
\hypersetup{ 
  colorlinks,
  citecolor=blue,
  linkcolor=red,
  urlcolor=magenta}

\usepackage{mathtools}
\usepackage{graphicx}
\usepackage{subcaption}
\usepackage{amsmath,bm,bbm}
\usepackage{rotating}
\usepackage{mathrsfs}
\usepackage{chngcntr}
\usepackage{xcolor}         
\usepackage{enumitem}
\usepackage{algorithm}
\usepackage{algpseudocode}
\usepackage{diagbox}

\newcommand{\transpose}{^{\mathrm{T}}}
\newcommand{\inverse}{^{-1}}
\newcommand{\halfpower}{^{1/2}}
\newcommand{\invhalfpower}{^{-1/2}}
\newcommand{\frobenius}{_{\mathrm{F}}}
\newcommand{\expect}{\mathbb{E}}
\newcommand{\prob}{\mathbb{P}}
\newcommand{\indicator}{\mathbbm{1}}
\newcommand{\zero}{{\mathbf{0}}}
\newcommand{\eps}{\varepsilon}
\newcommand{\Optilde}{\widetilde{O}_{\prob}}

\DeclareMathOperator*{\argmin}{arg\,min}
\DeclareMathOperator*{\argmax}{arg\,max}

\newcommand{\balpha}{\bm{\alpha}}
\newcommand{\bbeta}{\bm{\beta}}
\newcommand{\bgamma}{\bm{\gamma}}

\newcommand{\btheta}{\bm{\theta}}

\newcommand{\bDelta}{\bm{\Delta}}
\newcommand{\bSigma}{\bm{\Sigma}}

\newcommand{\bb}{\mathbf{b}}

\newcommand{\be}{\mathbf{e}}
\newcommand{\bg}{\mathbf{g}}

\newcommand{\bq}{\mathbf{q}}
\newcommand{\br}{\mathbf{r}}
\newcommand{\bs}{\mathbf{s}}
\newcommand{\bt}{\mathbf{t}}
\newcommand{\bu}{\mathbf{u}}
\newcommand{\bv}{\mathbf{v}}

\newcommand{\bx}{\mathbf{x}}
\newcommand{\by}{\mathbf{y}}
\newcommand{\bz}{\mathbf{z}}
\newcommand{\bA}{\mathbf{A}}
\newcommand{\bB}{\mathbf{B}}

\newcommand{\bE}{\mathbf{E}}

\newcommand{\bG}{\mathbf{G}}
\newcommand{\bH}{\mathbf{H}}
\newcommand{\bI}{\mathbf{I}}

\newcommand{\bP}{\mathbf{P}}

\newcommand{\bR}{\mathbf{R}}
\newcommand{\bS}{\mathbf{S}}

\newcommand{\bU}{\mathbf{U}}
\newcommand{\bV}{\mathbf{V}}
\newcommand{\bW}{\mathbf{W}}
\newcommand{\bX}{\mathbf{X}}
\newcommand{\bY}{\mathbf{Y}}
\newcommand{\bZ}{\mathbf{Z}}
\newcommand{\calA}{{\mathcal{A}}}

\newcommand{\calE}{{\mathcal{E}}}

\newcommand{\calL}{{\mathcal{L}}}

\newcommand{\calX}{{\mathcal{X}}}

\usepackage{lastpage}
\jmlrheading{23}{2022}{1-\pageref{LastPage}}{1/21; Revised 5/22}{9/22}{21-0000}{Dingbo Wu and Fangzheng Xie}

\ShortHeadings{Random Graphs with Surrogate Likelihood}{Wu and Xie}
\firstpageno{1}

\begin{document}

\title{Statistical Inference of Random Graphs With a Surrogate Likelihood Function}

\author{\name Dingbo Wu \email dw16@iu.edu \\
       \addr Department of Statistics\\
       Indiana University\\
       Bloomington, IN 47405, USA
       \AND
       \name Fangzheng Xie \email fxie@iu.edu \\
       \addr Department of Statistics\\
       Indiana University\\
       Bloomington, IN 47405, USA
       }

\editor{My editor}

\maketitle

\begin{abstract}
Spectral estimators have been broadly applied to statistical network analysis, but they do not incorporate the likelihood information of the network sampling model.
This paper proposes a novel surrogate likelihood function for statistical inference of a class of popular network models referred to as random dot product graphs.
In contrast to the structurally complicated exact likelihood function, the surrogate likelihood function has a separable structure and is log-concave yet approximates the exact likelihood function well. From the frequentist perspective, we study the maximum surrogate likelihood estimator and establish the accompanying theory. We show its existence, uniqueness, large sample properties, and that it improves upon the baseline spectral estimator with a smaller sum of squared errors. Furthermore, we derive the second-order bias of the proposed estimator and gain insight into why it outperforms some of the existing estimators. A computationally convenient stochastic gradient descent algorithm is designed to find the maximum surrogate likelihood estimator in practice. From the Bayesian perspective, we establish the Bernstein--von Mises theorem of the posterior distribution with the surrogate likelihood function and show that the resulting credible sets have the correct frequentist coverage. 
The empirical performance of the proposed surrogate-likelihood-based methods is validated through the analyses of simulation examples and a real-world Wikipedia graph dataset. 
\end{abstract}

\section{Introduction}
\label{section:introduction}
In the contemporary world of data science, network data are pervasive in a broad range of applications such as sociology \citep{lacetera-macis-mele-2016,young-scheinerman-2007}, econometrics \citep{mele-2017,mele2019spectral}, and neuroscience \citep{tang2019connectome}. 
Statistical network analysis is also an interdisciplinary area of research connected with many other fields, including computer science, machine learning, combinatorics, applied mathematics, and physics.
To model and analyze network data, various random graph models have been proposed in the literature, including the Erd\"os-R\'enyi graph \citep{erdHos1960evolution}, the stochastic block model \citep{holland-laskey-leinhardt-1983}, and the latent space model \citep{hoff-raftery-handcock-2002}. 

In this paper, we focus on random dot product graphs \citep{young-scheinerman-2007}, a class of random graph models
that are popular due to its simple architecture and flexibility.
On one hand, the edge probability matrix of a random dot product graph has a low-rank structure, which motivates, among others, the use of spectral methods in statistical network analysis.
On the other hand, the random dot product graph model is quite flexible because it not only encompasses the popularly used stochastic block models \citep{holland-laskey-leinhardt-1983,abbe-bandeira-hall-2016} and their offspring \citep{airoldi-blei-fienberg-xing-2008, binkiewicz-vogelstein-rohe-2017, lyzinski2017HSBM, sengupta-chen-2018}, but can also approximate general latent position graphs when the rank of the edge probability matrix grows with the number of vertices at a certain rate \citep{gao-lu-zhou-2015,tang-sussman-priebe-2013}.

Because the adjacency matrix has a low expected rank, spectral quantities such as the leading eigenvectors of the adjacency matrix and  those of the normalized Laplacian matrix, have been extensively used for low-rank random graph inference.
In particular, it is well known that the rows of these eigenvectors encode the cluster membership information when the underlying graph is generated from a stochastic block model \citep{abbe-fan-wang-zhong-2020, lyzinski-sussman-tang-athreya-priebe-2014, lei-rinaldo-2015, rohe-chatterjee-yu-2011, sussman-tang-fishkind-priebe-2012}.
There has been substantial recent development on the theory for spectral methods and the corresponding subsequent inference tasks in random dot product graphs. For an incomplete list of reference, see \cite{athreya-priebe-tang-2016, sussman-tang-priebe-2014, sarkar-bickel-2015, tang-priebe-2018, tang-sussman-priebe-2013, tang2017SPtest, tang2017NPtest}.
The readers are also referred to the survey paper \cite{rdpg-survey} for a review of the recent advances in this topic.

It has been pointed out \citep{xie-xu-2020-bayes, xie-xu-2021-onestep, xie-2022-entrywise} that, although the spectral methods for random dot product graphs have gained marvelous success and been broadly applied, the Bernoulli likelihood information contained in the graph distribution has been neglected. This observation has motivated the development of likelihood-based inference for random dot product graphs.
\cite{xie-xu-2020-bayes} proposed a fully Bayesian approach for estimating the latent positions in random dot product graphs, referred to as posterior spectral embedding, and established its global minimax optimality. 
\cite{li2023statistical} studied the maximum likelihood estimation for a general class of latent space networks and established the asymptotic normality of the resulting estimator under a slightly different setup.
\cite{xie-xu-2021-onestep} proposed a novel one-step procedure, which lead to a one-step estimator that took advantage of the Bernoulli likelihood information of the sampling model through the score function and the Fisher information matrix, to estimate random dot product graphs from the frequentist perspective. There, the authors further established the asymptotic efficiency of the one-step estimator and its smaller asymptotic sum of squared errors compared to that of the spectral estimators. 
The sparsity condition imposed in \cite{xie-xu-2021-onestep} was significantly weakened by \cite{xie-2022-entrywise} through a cleverly-designed leave-one-out analysis and delicate concentration analyses. 
Later, \cite{tang-cape-priebe-2022} applied the idea of the one-step refinement of spectral methods to stochastic block models when the block probability matrix is rank deficient.

Despite the success of the one-step estimator, a central question regarding likelihood-based inference for random dot product graphs remains open: What is the behavior of the frequentist maximum likelihood estimator? Also, a related question is: What is the behavior of the Bayes estimator? Efforts attempting to address these two questions aim to gain deeper insight into the likelihood-based inference for random dot product graphs from the frequentist and the Bayesian perspective, respectively. These two questions are also closely related through the Bernstein--von Mises phenomenon (see, for example, Section 10.2 in \citealp{van2000asymptotic}). Here, the major technical barrier is the complicated structure of the parameter space for the latent positions.
In this paper, we partially answer the aforementioned two questions by resorting to a cleverly-designed surrogate likelihood function that simplifies the parameter space enormously. Our work features the following novel contributions:
Firstly, the surrogate likelihood function has a separable structure, is log-concave, and the associated parameter space for the latent positions is a convex relaxation of the original latent space. These features greatly facilitate both the theoretical analyses and the related practical computations.
Secondly, we establish the existence, uniqueness, and asymptotic efficiency of the frequentist maximum surrogate likelihood estimator under the minimal sparsity condition. In particular, similar to the one-step estimator, the maximum surrogate likelihood estimator improves upon the baseline spectral estimators with a smaller sum of squared errors.
Furthermore, we have rigorously derived the second-order bias formulae of the maximum surrogate likelihood estimator and the one-step estimator, thereby providing insight into why the former typically outperforms the latter in certain finite sample problems.
Thirdly, we design a computationally efficient stochastic gradient descent algorithm for the maximum surrogate likelihood estimator with adaptive step sizes.
Fourthly, regarding the Bayes procedure, we establish the Bernstein--von Mises theorem for the posterior distribution with the surrogate likelihood function and show that the resulting credible sets have the correct frequentist coverage probabilities.


The remaining part of the article is structured as follows. In Section \ref{section:background-SL}, we review the background of random dot product graphs and introduce the surrogate likelihood function. Section \ref{sec:MSLE} elaborates on the theoretical properties and the computational algorithm of the maximum surrogate likelihood estimation. In Section \ref{sec:Bayesian_Estimation_Surrogate_Likelihood}, we establish the large sample properties of the Bayes procedure with the surrogate likelihood function. Section \ref{section:numerical} demonstrates the empirical performance of the proposed methods through simulation examples and the analysis of a real-world Wikipedia network dataset. We conclude the paper with a discussion in Section \ref{section:last}.

\noindent
\textit{Notations}: Let $[n]$ denote the set of consecutive integers from $1$ to $n$: $[n]=\{1,\ldots,n\}$. The symbol $\lesssim_\delta$ means an inequality up to a constant depending on $\delta$, that is, $a\lesssim_\delta b$ if $a\leq C_\delta b$ for some constant $C_\delta > 0$ depending on $\delta$; a similar definition also applies to the symbol $\gtrsim_\delta$. The notation $\|\bx\|$ denotes the Euclidean norm of a vector $\bx=[x_1,\ldots,x_d]\transpose\in\mathbb{R}^d$, that is, $\|\bx\|=(\sum_{k=1}^d x_k^2)\halfpower$. The $d\times d$ identity matrix is denoted by $\bI_d$. The notation $\mathbb{O}(n,d)=\{\bU\in\mathbb{R}^{n\times d}:\bU\transpose\bU=\bI_d\}$ denotes the set of all orthonormal $d$-frames in $\mathbb{R}^n$, where $d\leq n$, and we write $\mathbb{O}(d) = \mathbb{O}(d,d)$. For a matrix $\bX=[x_{ik}]_{n\times d}$, $\sigma_k(\bX)$ denotes its $k$th largest singular value, and when $\bX$ is square and symmetric, $\lambda_k(\bX)$ denotes its $k$th largest eigenvalue in magnitude. Matrix norms with following definitions are used: the spectral norm $\|\bX\|_2 = \sigma_1(\bX)$, the Frobenius norm $\|\bX\|\frobenius = (\sum_{i=1}^n\sum_{k=1}^d x_{ik}^2)\halfpower$, the matrix infinity norm $\|\bX\|_\infty=\max_{i\in[n]}\sum_{k=1}^d|x_{ik}|$, and the two-to-infinity norm $\|\bX\|_{2\to\infty}=\max_{i\in[n]}(\sum_{k=1}^d x_{ik}^2)\halfpower$. In particular, these norm notations apply to any Euclidean vector $\bx\in\mathbb{R}^d$ viewed as a $d\times 1$ matrix. Given two symmetric positive semidefinite matrices $\bA,\bB$ of the same dimension, we write $\bA\preceq\bB$ ($\bA\succeq \bB$, respectively) if $\bB - \bA$ ($\bA - \bB$, respectively) is positive semidefinite. 

\section{Background and the Surrogate Likelihood}
\label{section:background-SL}
\subsection{Background on random dot product graphs}
\label{subsection:background-on-rdpg}
We begin by briefly reviewing the background on random dot product graphs and adjacency spectral embedding. Consider a graph with $n$ vertices labeled as $[n] = \{1,\ldots,n\}$. Let $\calX$ be a subset of $\mathbb{R}^d$ such that $\bx_1\transpose{}\bx_2\in (0, 1)$ for all $\bx_1,\bx_2\in\calX$, where $d$ is fixed and $d\leq n$, and let $\rho_n\in (0, 1]$ be a sparsity factor. Each vertex $i\in [n]$ is associated with a vector $\bx_i\in\calX$, referred to as the latent position for vertex $i$. We say that a symmetric random matrix $\bA = [A_{ij}]_{n\times n}\in\{0, 1\}^{n\times n}$ is an adjacency matrix generated by a random dot product graph with latent position matrix $\bX = [\bx_1,\ldots,\bx_n]\transpose{}$ and sparsity factor $\rho_n$, denoted by $\bA\sim\mathrm{RDPG}(\rho_n^{1/2}\bX)$, if the random variables $A_{ij}\sim\mathrm{Bernoulli}(\rho_n\bx_{i}\transpose{}\bx_{j})$ independently for all $i,j\in [n]$, $i\leq j$, and $A_{ij} = A_{ji}$ for all $i > j$. The distribution of $\bA$ can thus be written as $p_{\bX}(\bA) = \prod_{i\leq j}(\rho_n\bx_i\transpose\bx_j)^{A_{ij}}(1-\rho_n\bx_i\transpose\bx_j)^{1-A_{ij}}$. 
The sparsity factor $\rho_n$ fundamentally controls the overall average graph expected degree through $n\rho_n$ when the entries of $\bX\bX\transpose$ are bounded away from $0$ and $\infty$.

\begin{remark}[Deterministic versus stochastic latent positions]
In this work, we consider the setup where the latent positions $\bx_1,\ldots,\bx_n$ are deterministic parameters to be estimated. Another slightly different modeling approach is to consider $\bx_1,\ldots,\bx_n$ as independent and identically distributed latent random variables
(see, for example, \citealp{athreya-priebe-tang-2016, tang2017NPtest, tang-priebe-2018}). This random formulation of the latent positions introduces implicit homogeneity and is connected to the infinite exchangeable random graphs \citep{janson2008graph}. The same homogeneity condition was retained in \cite{xie-xu-2021-onestep} using a Glivenko--Cantelli type condition when $\bx_1,\ldots,\bx_n$ are deterministic. The latter Glivenko--Cantelli type condition is also relaxed in the current work as we only require that $\sigma_d(\bX) > 0$ (see Remark \ref{remark:identifiability} below).
\end{remark}

\begin{remark}[Nonidentifiability]
\label{remark:identifiability}
The latent position matrix $\bX$ is not uniquely identified in the following two senses. Firstly, any low-rank positive semidefinite connection probability matrix $\bP = \bX\bX\transpose$ can have different factorizations because for any orthogonal matrix $\bW\in\mathbb{O}(d)$, $\bX\bX\transpose = (\bX\bW)(\bX\bW)\transpose$. Secondly, for any $d'>d$ and any latent position matrix $\bX\in\mathbb{R}^d$, there exists another matrix $\bX'\in\mathbb{R}^{d'}$ such that $\bX\bX\transpose=\bX'(\bX')\transpose$. The latter source of non-identifiability can be removed by requiring that $\sigma_d(\bX)>0$, while the former source is inevitable without further constraints. Thus, any estimator of the latent position matrix $\bX$ can only recover it up to an orthogonal transformation. 
\end{remark}

\begin{example}[Stochastic block model]
Random dot product graphs have connections with the popular stochastic block model \citep{holland-laskey-leinhardt-1983}. Consider a graph with $n$ vertices that are partitioned into $K$ communities, where $K$ is assumed to be much smaller than $n$. Let $\tau:[n]\to[K]$ be a cluster assignment function that assigns each vertex to a unique community. Let $\bB=[B_{kl}]_{K\times K}\in(0,1)^{K\times K}$ be a symmetric probability matrix and $A_{ij}$ be the binary indicator of the existence of an edge between vertices $i$ and $j$. Then the stochastic block model specifies that $A_{ij}\sim \mathrm{Bernoulli}(B_{\tau(i)\tau(j)})$ independently for all $i,j\in[n]$, $i \leq j$, and $A_{ij}=A_{ji}$ for all $i > j$. By converting the community assignment to a matrix $\bZ=[1\{\tau(i)=k\}]_{n\times K}$, we see that the expected adjacency matrix $\bZ\bB\bZ\transpose$ is symmetric and of low rank. Furthermore, if $\bB$ is positive semidefinite with rank $d\leq K$ and can be factorized as $\bB=\bV\bV\transpose$ for a $K\times d$ matrix $\bV$, then $\bA$ can be seen as an adjacency matrix generated by the random dot product graph with latent position matrix $\bX=\bZ\bV$, that is, $\bA\sim\mathrm{RDPG}(\bZ\bV)$.
\end{example}

Motivated by the low-rank structure of random dot product graphs, \cite{sussman-tang-fishkind-priebe-2012} proposed to estimate the latent position matrix $\bX$ by solving the least squares problem
    $
    \widetilde{\bX} = \argmin_{\bX\in\mathbb{R}^{n\times d}} \|\bA-\bX\bX\transpose\|_\mathrm{F}^2.
    $
The interpretation is that $\widetilde{\bX}\widetilde{\bX}\transpose$ can be viewed as the projection of the data matrix $\bA$ onto the space of all $n\times n$ rank-$d$ positive semidefinite matrices with regard to the Frobenius norm distance.
The solution $\widetilde{\bX}$ is referred to as the adjacency spectral embedding of $\bA$ into $\mathbb{R}^d$, and can be computed as the matrix of eigenvectors associated with the top $d$ eigenvalues of $\bA$, scaled by the square roots of the corresponding eigenvalues \citep{eckart-young-1936}. The asymptotic properties of $\widetilde{\bX}$ have been established in the literature \citep{sussman-tang-priebe-2014, athreya-priebe-tang-2016, tang-priebe-2018}. Notably, \citet{athreya-priebe-tang-2016}, \cite{tang-priebe-2018}, and \cite{xie-xu-2021-onestep} have shown that each row of the adjacency spectral embedding converges to a mean-zero multivariate normal distribution after appropriate standardization. 


\subsection{The surrogate likelihood function}
\label{subsection:intro-surr-lik}
In this subsection, we derive the surrogate likelihood function for the random dot product graph model. The motivation is that the exact likelihood function has a complicated structure, bringing challenges for developing the theory of maximum likelihood estimation. The difficulty partially comes from the fact that the random dot product graph model belongs to a curved exponential family, and the theory of the maximum likelihood estimation is much more difficult in curved exponential families than in canonical ones (see, for example, Section 2.3 in \citealp{bickel2007mathematical}). 
Also, the boundary of the parameter space renders the maximum likelihood estimation intractable, both computationally and analytically. To be more specific, consider the log-likelihood function 
\begin{align*}
\ell_\bA(\bX) = \sum_{1\leq i\leq j\leq n}\{A_{ij}\log(\rho_n\bx_i\transpose\bx_j) + (1 - A_{ij})\log(\rho_n\bx_i\transpose\bx_j)\}.
\end{align*}
The parameter space is defined by $\{\bX = [\bx_1,\ldots,\bx_n]\transpose\in\mathbb{R}^{n\times d}:0 < \bx_i\transpose\bx_j < 1\text{ for all }i,j\}$, and over the boundary, the log-likelihood function has an unbounded gradient. This is in sharp contrast with the requirement in \cite{li2023statistical}, where the log-likelihood functions are required to have bounded derivatives up to the second order over the entire parameter space. These challenges motivate the development of a more computationally and analytically tractable surrogate likelihood approach. 

To distinguish a generic latent position $\bx_i\in\mathbb{R}^d$ and its true value associated with the data generating distribution, let $\bx_{0i}$ denote the ground truth of $\bx_i$, $i\in [n]$, and $\bX_0 = [\bx_{01},\ldots,\bx_{0n}]\transpose$. Let us begin by considering the log-likelihood function of a single $\bx_i$ when the remaining latent positions $(\bx_{0j})_{j\neq i}$ are accessible:
\begin{equation}
\label{eqn:loglikelihood_ell_0in}
\begin{aligned}
\ell_{0in}(\bx_i)& = \sum_{j\neq i}^n\{A_{ij}\log(\rho_n\bx_i\transpose\bx_{0j}) + (1 - A_{ij})\log(1 - \rho_n\bx_i\transpose\bx_{0j})\}\\
&\quad + \{A_{ii}\log(\rho_n\bx_i\transpose\bx_i) + (1 - A_{ij})\log(1 - \rho_n\bx_i\transpose\bx_i)\}.
\end{aligned}
\end{equation}
We refer to $\ell_{0in}(\bx_i)$ in \eqref{eqn:loglikelihood_ell_0in} as the oracle log-likelihood function because it requires the true values of the remaining $\bx_j$'s with $j\neq i$. 
Theorem 2 in \cite{xie-xu-2021-onestep} established the consistency and asymptotic normality of the maximizer of the oracle log-likelihood function $\ell_{0in}(\bx_i)$ in \eqref{eqn:loglikelihood_ell_0in}. Nevertheless, the oracle log-likelihood is not computable because $(\bx_{0j})$ are not accessible in practice. Following the idea in \cite{xie-xu-2021-onestep}, we consider replacing the unknown latent positions by the corresponding rows of the adjacency spectral embedding. Formally, let $\widetilde{\bx}_j$ be the $j$th row of the adjacency spectral embedding $\widetilde{\bX}$, $j\in [n]$. Then we obtain the following approximation to the oracle log-likelihood:
\begin{align}
\label{eqn:loglikelihood_approximation}
\ell_{0in}(\bx_i) \approx \sum_{j = 1}^n\{A_{ij}\log(\rho_n^{1/2}\bx_i\transpose\widetilde{\bx}_j) + (1 - A_{ij})\log(1 - \rho_n^{1/2}\bx_i\transpose\widetilde{\bx}_j)\}.
\end{align}
Note that the last term in $\ell_{0in}$ is replaced by $A_{ii}\log(\rho_n^{1/2}\bx_i\transpose\widetilde{\bx}_i) + (1 - A_{ii})\log(1 - \rho_n^{1/2}\bx_i\transpose\widetilde{\bx}_i)$ for convenience, which is immaterial. This approximation step is motivated by the uniform consistency of the adjacency spectral embedding: There exists a $d\times d$ orthogonal $\bW$ such that $\|\widetilde{\bX}\bW - \rho_n^{1/2}\bX_0\|_2 = O\{\sqrt{(\log n)/n}\}$ with high probability \citep{lyzinski-sussman-tang-athreya-priebe-2014,xie-2022-entrywise}. 

With the approximation in \eqref{eqn:loglikelihood_approximation}, the constraints for the latent position $\bx_i$ become a system of linear inequalities: $0 < \rho_n^{1/2}\bx_i\transpose\widetilde{\bx}_j < 1$ for all $j\in [n]$. Geometrically, these constraints correspond to a convex polyhedron. Namely, given any vector $\bx_i\in\mathbb{R}^d$, checking whether $\bx_i$ is in such a convex polyhedron requires $O(n)$ operations, so that the relevant computation could be cumbersome. We now resolve this issue by applying a quadratic Taylor approximation to the terms $\log(\rho_n^{1/2}\bx_i\transpose\widetilde{\bx}_j)$ and relax the parameter space for $\bx_i$. Here we can drop the sparsity factor $\rho_n$ without loss of generality. Formally, write $g_{ij}(\bx_i) = A_{ij}\log(\bx_i\transpose\widetilde{\bx}_j)$. Then a quadratic Taylor approximation to $g_{ij}$ at $\bx_i = \widetilde{\bx}_i$ leads to 
\begin{align}\label{eqn:log_Taylor_approximation}
g_{ij}(\bx_i)= g_{ij}(\widetilde{\bx}_i) + \frac{A_{ij}\widetilde{\bx}_j\transpose(\bx_i - \widetilde{\bx}_i)}{\widetilde{\bx}_i\transpose\widetilde{\bx}_j} - \frac{A_{ij}(\bx_i - \widetilde{\bx}_i)\transpose\widetilde{\bx}_j\widetilde{\bx}_j\transpose(\bx_i - \widetilde{\bx}_i)}{2(\widetilde{\bx}_i\transpose\widetilde{\bx}_j)^2} + \mbox{remainder}.
\end{align}
Meanwhile, it is also conceivable that
\begin{align}\label{eqn:log_Hessian_approximation}
\sum_{j = 1}^n\frac{A_{ij}}{2(\widetilde{\bx}_i\transpose\widetilde{\bx}_j)^2}(\bx_i - \widetilde{\bx}_i)\transpose\widetilde{\bx}_j\widetilde{\bx}_j\transpose(\bx_i - \widetilde{\bx}_i) \approx \sum_{j = 1}^n\frac{1}{2\widetilde{\bx}_i\transpose\widetilde{\bx}_j}(\bx_i - \widetilde{\bx}_i)\transpose\widetilde{\bx}_j\widetilde{\bx}_j\transpose(\bx_i - \widetilde{\bx}_i)
\end{align}
because $\expect_0(A_{ij}) = \rho_n\bx_{0i}\transpose\bx_{0j}\approx \widetilde{\bx}_i\transpose\widetilde{\bx}_j$. Hence, ignoring the constant terms that are free of $\bx_i$, combining the approximations in \eqref{eqn:loglikelihood_approximation}, \eqref{eqn:log_Taylor_approximation}, and \eqref{eqn:log_Hessian_approximation} leads to the following surrogate log-likelihood function
\begin{equation}\label{def:surr-lik-fun}
\widetilde{\ell}_{in}(\bx_i) = \sum_{j=1}^n
\left\{
\frac{A_{ij} \widetilde{\bx}_j\transpose \bx_i}{\widetilde{\bx}_i\transpose\widetilde{\bx}_j} + \widetilde{\bx}_j\transpose \bx_i - \frac{1}{2\widetilde{\bx}_i\transpose\widetilde{\bx}_j}\bx_i\transpose\widetilde{\bx}_j\widetilde{\bx}_j\transpose \bx_i + (1-A_{ij})\log(1-\bx_i\transpose \widetilde{\bx}_j)
\right\}.
\end{equation}
Therefore, 
by Cauchy--Schwarz inequality, the the parameter space for $\bx_i$ associated with the surrogate likelihood of vertex $i$ can be taken as the unit ball $\{\bx_i\in\mathbb{R}^d:\|\bx_i\|\leq 1\}$ for all $i\in[n]$ when $\max_{j\in [n]}\|\widetilde{\bx}_j\|_2 < 1$ (which holds with high probability). Consequently, we relax the original complicated parameter space $\{\bx_i\in\mathbb{R}^d:0 < \bx_i\transpose\widetilde{\bx}_j < 1,j\in [n]\}$ to a simple unit ball, which is much more tractable to work with.
Moreover, the surrogate likelihood function has a separable structure because $\widetilde{\ell}_{in}(\bx_i)$ does not involve $\bx_j$ for $j\neq i$. This convenience enables parallelization when related computation is requested.
In addition, a simple algebra shows that the surrogate likelihood function is log-concave, a highly desired feature when optimization and Monte Carlo sampling are needed.


\subsection{Comparison with the one-step estimator}
\label{sub:comparison_with_OSE}
Recently, \cite{xie-xu-2021-onestep} proposed a one-step estimator $\widehat{\bX}^{\mathrm{OS}}=[\widehat\bx_1^{\mathrm{OS}}, \ldots, \widehat\bx_n^{\mathrm{OS}}]\transpose$ for random dot product graphs that improves upon the adjacency spectral embedding:
\begin{equation}\label{def:ose}
    \widehat\bx_i^{\mathrm{OS}} = \widetilde\bx_i + \left\{\frac{1}{n}\sum_{j=1}^n\frac{\widetilde\bx_j\widetilde\bx_j\transpose}{\widetilde\bx_i\transpose\widetilde\bx_j(1-\widetilde\bx_i\transpose\widetilde\bx_j)}\right\}^{-1} \left\{\frac{1}{n}\sum_{j=1}^n\frac{(A_{ij}-\widetilde\bx_i\transpose\widetilde\bx_j)\widetilde\bx_j}{\widetilde\bx_i\transpose\widetilde\bx_j(1-\widetilde\bx_i\transpose\widetilde\bx_j)}\right\}, \quad i\in[n].
\end{equation}
The one-step estimator originates from a one-step updating scheme of the Newton-Raphson method for maximizing the log-likelihood function with the initial guess being the adjacency spectral embedding (see, e.g., Section 5.7 in \citealp{van2000asymptotic}). It is clear from the construction that the one-step estimator takes advantage of the likelihood information of the sampling distribution through the Fisher information matrix and the score function. 

In Section \ref{subsection:intro-surr-lik}, we have shown the derivation of the surrogate log-likelihood function by applying a quadratic Taylor approximation to the logarithm function $\log(\bx_i\transpose\widetilde{\bx}_j)$. We now show that the same approximation treatment applied to the entire function in \eqref{eqn:loglikelihood_approximation} results in an approximate log-likelihood function whose maximizer is exactly the one-step estimator. Formally, applying a second-order Taylor expansion to the term $\log(1 - \bx_i\transpose\widetilde{\bx}_j)$ at $\bx_i = \widetilde{\bx}_i$ yields
\begin{align}\label{eqn:OSE_SL_approximation_I}
\log\frac{(1 - \bx_i\transpose\widetilde{\bx}_j)}{(1 - \widetilde{\bx}_i\transpose\widetilde{\bx}_j)}
& = - \frac{\widetilde{\bx}_j\transpose(\bx_i - \widetilde{\bx}_i)}{1 - \widetilde{\bx}_i\transpose\widetilde{\bx}_j} - 
\frac{(\bx_i - \widetilde{\bx}_i)\transpose\widetilde{\bx}_j\widetilde{\bx}_j\transpose(\bx_i - \widetilde{\bx}_i)}{2(1 - \widetilde{\bx}_i\transpose\widetilde{\bx}_j)^2} + \mbox{remainder}. 
\end{align}
Following the idea in \eqref{eqn:log_Hessian_approximation}, we can also conceive the following approximation:
\begin{align}\label{eqn:OSE_SL_approximation_II}
\sum_{j = 1}^n
\frac{(1 - A_{ij})(\bx_i - \widetilde{\bx}_i)\transpose\widetilde{\bx}_j\widetilde{\bx}_j\transpose(\bx_i - \widetilde{\bx}_i)}{2(1 - \widetilde{\bx}_i\transpose\widetilde{\bx}_j)^2} \approx \sum_{j = 1}^n
\frac{(\bx_i - \widetilde{\bx}_i)\transpose\widetilde{\bx}_j\widetilde{\bx}_j\transpose(\bx_i - \widetilde{\bx}_i)}{2(1 - \widetilde{\bx}_i\transpose\widetilde{\bx}_j)}.
\end{align}
We thus obtain the following quadratic approximation to \eqref{eqn:loglikelihood_approximation} modulus a constant term from \eqref{def:surr-lik-fun}, \eqref{eqn:OSE_SL_approximation_I}, and \eqref{eqn:OSE_SL_approximation_II}:
\begin{align}\label{eqn:OSE_surrogate_loglikeliood}
\widetilde{\ell}_{in}^{(\mathrm{OS})}(\bx_i)
& = \sum_{j = 1}^n\frac{(A_{ij} - \widetilde{\bx}_i\transpose\widetilde{\bx}_j)\widetilde{\bx}_j\transpose(\bx_i - \widetilde{\bx}_i)}{\widetilde{\bx}_i\transpose\widetilde{\bx}_j(1 - \widetilde{\bx}_i\transpose\widetilde{\bx}_j)} - \sum_{j = 1}^n\frac{(\bx_i - \widetilde{\bx}_i)\transpose\widetilde{\bx}_j\widetilde{\bx}_j\transpose(\bx_i - \widetilde{\bx}_i)}{2\widetilde{\bx}_i\transpose\widetilde{\bx}_j(1 - \widetilde{\bx}_i\transpose\widetilde{\bx}_j)} .
\end{align}
Then a simple algebra shows that the one-step estimator $\widehat{\bx}_i^{(\mathrm{OS})}$ maximizes $\widetilde{\ell}_{in}^{(\mathrm{OS})}$ defined in \eqref{eqn:OSE_surrogate_loglikeliood}. 

Clearly, the surrogate log-likelihood function in \eqref{def:surr-lik-fun} is constructed by applying a Taylor expansion to the term $\log(\bx_i\transpose\widetilde{\bx}_j)$, whereas the one-step estimator is obtained by applying the Taylor expansion to the entire function.
Thus, intuitively, the surrogate likelihood retains more likelihood information than the one-step procedure does.
Below, we visualize this heuristic using a toy numerical example.
\begin{example}\label{example:SL_function}
Consider the following random dot product graph model.
Let $n = 300$, $(t_i)_{i = 1}^n$ be equidistant points over $[0, 1]$, $x_{0i} = 0.2 + 0.6\sin(\pi t_i)$, $i\in [n]$, and $\bX_0 = [x_{01},\ldots,x_{0n}]\transpose\in\mathbb{R}^{n\times 1}$. Suppose $\bA\sim\mathrm{RDPG}(\bX_0)$ and we focus on the likelihood function for $\bx_i$ with $i = 100$.
Figure \ref{fig:SL_comparison} visualizes the comparison among the oracle log-likelihood $\ell_{0in}(\bx_i)$, the surrogate log-likelihood $\widetilde{\ell}_{in}(\bx_i)$, and the approximate log-likelihood $\widetilde{\ell}_{in}^{(\mathrm{OS})}(\bx_i)$ associated with the one-step estimator. 
\begin{figure}[htbp]
\centering
\includegraphics[width = 0.65\textwidth]{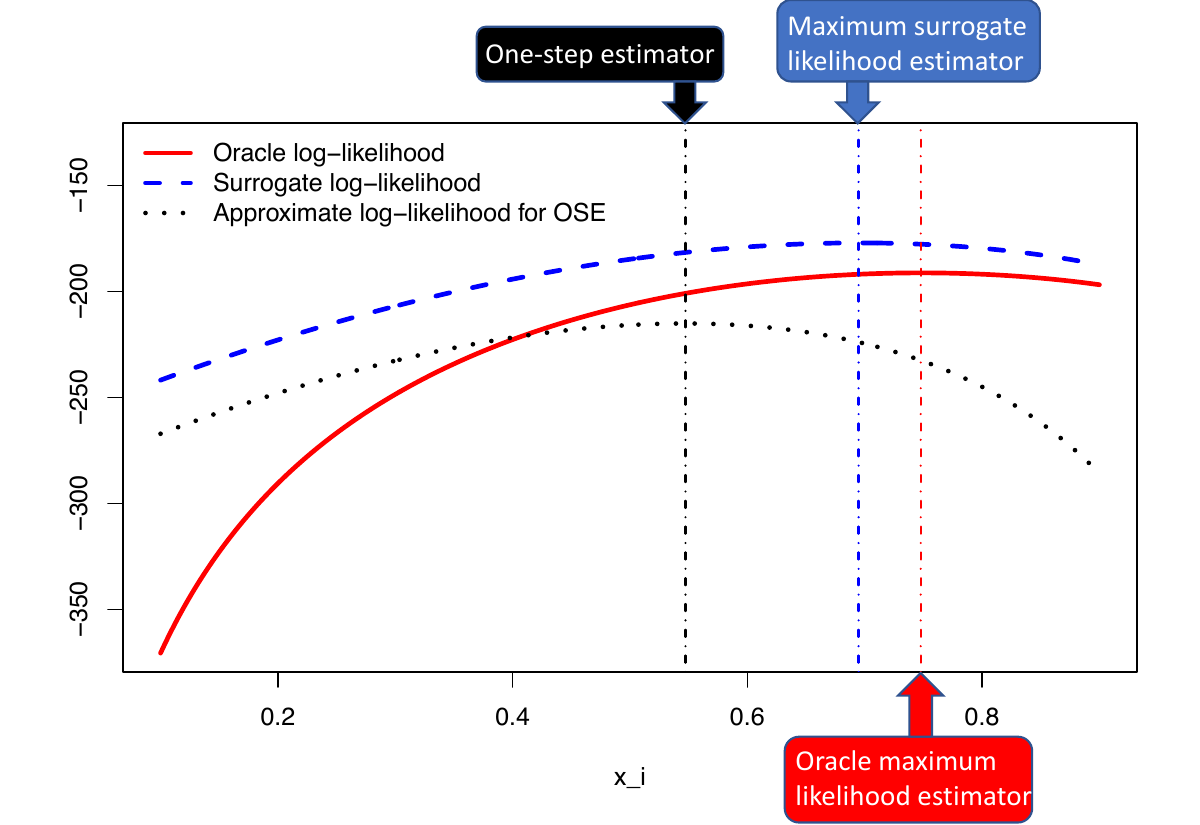}
\caption{Comparison among the oracle log-likelihood function $\ell_{0in}(\bx_i)$, the surrogate log-likelihood function $\widetilde{\ell}_{in}(\bx_i)$, and the approximate log-likelihood function $\widetilde{\ell}_{in}^{(\mathrm{OS})}(\bx_i)$ associated with the one-step estimator. The three vertical lines mark the one-step estimate, the maximum surrogate likelihood estimate, and the oracle maximum likelihood estimate, respectively.}
\label{fig:SL_comparison}
\end{figure}
The constant terms of these functions have been added to make them comparable. The vertical lines mark the maximizers of the three functions, respectively. 
It is visually clear that the maximizer of the surrogate log-likelihood estimate is closer to that of the oracle log-likelihood than the one-step estimate is, suggesting that the maximum surrogate likelihood estimator may outperform the one-step estimator in some practical finite sample problems.
\end{example}

\section{Maximum Surrogate Likelihood Estimation}
\label{sec:MSLE}


\subsection{Theoretical properties}
\label{sub:MSLE_theory}
This subsection elaborates on the theoretical properties of the frequentist inference with the surrogate likelihood. Below, Theorem \ref{theorem:existence_and_uniqueness} establishes the existence and uniqueness of the maximum  surrogate likelihood estimator. 
\begin{theorem}
\label{theorem:existence_and_uniqueness}
Suppose $\bA\sim\mathrm{RDPG}(\rho_n\halfpower\bX_0)$ and $(\log n)/(n\rho_n)\to0$ as $n\to\infty$. Assume $\lambda_d(\bX_0\transpose\bX_0/n)\geq\lambda$ for some constant $\lambda>0$ for all $n>d$, and $\min_{i,j\in[n]}(\bx_{0i}\transpose\bx_{0j}, 1-\bx_{0i}\transpose\bx_{0j}) \geq\delta$ for some constant $\delta>0$. Let $i\in [n]$ be a fixed vertex and consider the maximum surrogate likelihood estimator
$
\widehat\bx_i = \argmax_{\bx_i:\|\bx_i\|_2\leq1} \widetilde{\ell}_{in}(\bx_i).
$
Then for any $c>0$, there exists some constant $N_{c,\delta,\lambda}\in\mathbb{N}_+$ depending on $c,\delta,\lambda$ such that $\prob_0(\widehat\bx_i\mathrm{\ exists\ and\ is\ unique})\geq1-n^{-c}$ for all $n\geq N_{c,\delta,\lambda}$.
\end{theorem}
Let
$
\bG_{0in}=(1/n)\sum_{j=1}^n{\bx_{0j}\bx_{0j}\transpose}\{\bx_{0i}\transpose\bx_{0j}(1-\rho_n\bx_{0i}\transpose\bx_{0j})\}^{-1}
$
be the Fisher information matrix with regard to the latent position $\bx_i$.
Theorem \ref{theorem:asymptotic_properties_of_MSLE} below, which is one of the main results in this article, establishes the large sample properties of the maximum surrogate likelihood estimator.
\begin{theorem}
\label{theorem:asymptotic_properties_of_MSLE}
Suppose 
the conditions of Theorem \ref{theorem:existence_and_uniqueness} hold and the embedding dimension $d$ is fixed. For each $i\in [n]$, let $\widehat{\bx}_i = \argmax_{\bx_i:\|\bx_i\|_2\leq 1}\widetilde{\ell}_{in}(\bx_i)$ be the maximum surrogate likelihood estimator. Then there exists an orthogonal matrix $\bW\in\mathbb{O}(d)$ that depends on $n$, such that for each $i\in[n]$,
\[
\sqrt{n}\bG_{0in}\halfpower(\bW\transpose\widehat\bx_i-\rho_n\halfpower\bx_{0i}) \overset{\calL}{\to} \mathrm{N}_d(\zero_d, \bI_d).
\]
Furthermore, if $(\log n)^4/(n\rho_n)\to 0$, then
\[
\|\widehat\bX\bW-\rho_n\halfpower\bX_0\|\frobenius^2 - \frac{1}{n}\sum_{i = 1}^n\mathrm{tr}(\bG_{0in}\inverse) \to 0
\]
in probability, where $\widehat{\bX} = [\widehat{\bx}_1,\ldots,\widehat{\bx}_n]\transpose$.
\end{theorem}

\begin{remark}[Sparsity condition]
The sparsity condition that $(\log n)/(n\rho_n)\to 0$ required in Theorem \ref{theorem:existence_and_uniqueness} and in the asymptotic normality of Theorem \ref{theorem:asymptotic_properties_of_MSLE} is minimal in the following sense. It is well known that the random adjacency matrix $\bA$ no longer concentrates around its expected value $\expect_0(\bA)$ when $(\log n)/(n\rho_n)\to \infty$ \citep{tang-priebe-2018}. Furthermore, \cite{abbe-fan-wang-zhong-2020} and \cite{xie-2022-entrywise} showed that in order to have $\|\rho_n^{-1/2}\widetilde{\bX}\bW - \bX_0\|_{2\to\infty} = o(1)$ with high probability, which is an indispensable ingredient in our employed proof, it is necessary that $(\log n)/(n\rho_n)\to 0$. 
\end{remark}

\begin{remark}[Comparison with other estimators]
\cite{athreya-priebe-tang-2016}, \cite{tang-priebe-2018}, and \cite{xie-xu-2021-onestep} have establish the large sample properties of the adjacency spectral embedding and the one-step estimator as the following. 
Let
$\widetilde{\bX} = [\widetilde{\bx}_1,\ldots,\widetilde{\bx}_n]\transpose$ and 
$\widehat{\bX}^{(\mathrm{OS})} = [\widehat{\bx}_1^{(\mathrm{OS})},\ldots,\widehat{\bx}_n^{(\mathrm{OS})}]\transpose$.
Under appropriate conditions, for each vertex $i\in [n]$, 
\begin{align*}
&\sqrt{n}\bSigma_{in}^{-1/2}(\bW\transpose\widetilde{\bx}_i - \rho_n^{1/2}\bx_{0i})\overset{\calL}{\to}\mathrm{N}_d(\zero_d, \bI_d), \quad \sqrt{n}\bG_{0in}^{1/2}(\bW\transpose\widehat{\bx}_i^{(\mathrm{OS})} - \rho_n^{1/2}\bx_{0i})\overset{\calL}{\to}\mathrm{N}_d(\zero_d, \bI_d),\\
&\|\widetilde{\bX}\bW - \rho_n^{1/2}\bX_0\|_{\mathrm{F}}^2 - \frac{1}{n}\sum_{i = 1}^n\mathrm{tr}(\bSigma_{in})\overset{\prob_0}{\to} 0, \quad  \|\widehat{\bX}^{(\mathrm{OS})}\bW - \rho_n^{1/2}\bX_0\|_{\mathrm{F}}^2 -  \frac{1}{n}\sum_{i = 1}^n\mathrm{tr}(\bG_{0in}^{-1})\overset{\prob_0}{\to}0.
\end{align*}
where the covariance matrix $\bSigma_{in}$ satisfies $\bSigma_{in}\succeq \bG_{0in}^{-1}$. Theorem \ref{theorem:asymptotic_properties_of_MSLE} thus suggests that the maximum surrogate likelihood estimator improves upon the adjacency spectral embedding and is (first-order) asymptotically equivalent to the one-step estimator. This phenomenon is also known as the local efficiency \citep{xie-xu-2021-onestep} because the asymptotic covariance matrix for a single latent position $\bx_i$ is the same as that of the oracle maximum likelihood estimator.
\end{remark}




Given that both the one-step estimator and the maximum surrogate likelihood estimator achieve the local efficiency, the comparison at the first-order ($O(n^{-1/2})$) is unable to distinguish their performance. To further discern the difference between these two estimators, it is desirable to explore their second-order ($O(n^{-1})$) behavior. Such an idea can be formalized by the second-order ($O(n^{-1})$) bias of an estimator. Generically, given an asymptotic unbiased estimator $\widehat{\btheta}_n$ for an unknown parameter $\btheta$ (\emph{i.e.}, $\lim_{n\to\infty}\expect\widehat{\btheta}_n = \btheta$), if $\expect\widehat{\btheta}_n = \btheta + \bb_n + o(n^{-1})$ and $\bb_n = O(n^{-1})$, then the $O(n^{-1})$ bias of $\widehat{\btheta}_n$ is given by $\mathrm{Bias}(\widehat{\btheta}_n) = \bb_n$. Also see \cite{PFANZAGL19781,RILSTONE1996369,https://doi.org/10.1111/j.1468-0262.2004.00482.x,10.1214/009053606000001208} for the analyses of the $O(n^{-1})$ biases in the econometric literature. Below, Theorem \ref{thm:second_order_bias} establishes the formulae of the $O(n^{-1})$ biases of the one-step estimator and the maximum surrogate likelihood estimator. 
\begin{theorem}\label{thm:second_order_bias}
Suppose 
the conditions of Theorem \ref{theorem:existence_and_uniqueness} hold and the embedding dimension $d$ is fixed. Further assume that $\rho_n = 1$ and $\bX_0\transpose\bX_0$ is a diagonal matrix with different eigenvalues, and the differences of these are lower bounded by a constant multiple of $n$.
For each $i\in [n]$, let $\widehat{\bx}_i = \argmax_{\bx_i:\|\bx_i\|_2\leq 1}\widetilde{\ell}_{in}(\bx_i)$ be the maximum surrogate likelihood estimator. Then the $O(n^{-1})$ biases of $\widehat{\bx}_i$ and $\widehat{\bx}_i^{(\mathrm{OS})}$ are given by
\begin{align*}
\mathrm{Bias}(\widehat{\bx}_i) & = \bb_i^{(\mathrm{MSLE})} + \bb_i^{(\mathrm{ASE})} + \bb_i^{(\mathrm{base})},\\
\mathrm{Bias}(\widehat{\bx}_i^{(\mathrm{OS})}) &= \bb_i^{(\mathrm{OS})} + \bb_i^{(\mathrm{ASE})} + \bb_i^{(\mathrm{base})},
\end{align*}
where
\begin{align*}
\bb_i^{(\mathrm{OSE})}
& = \bG_{0in}^{-1}\frac{1}{n}\sum_{j = 1}^n\frac{(2\bx_{0i}\transpose\bx_{0j} - 1)\bx_{0j}\bx_{0j}\transpose(\bX\transpose\bX)^{-1}\bx_{0j}}{(\bx_{0i}\transpose\bx_{0j})(1 - \bx_{0i}\transpose\bx_{0j})},\\
\bb_i^{(\mathrm{MSLE})}
& = -\bG_{0in}^{-1}\frac{1}{n}\sum_{j = 1}^n\frac{(1 - \bx_{0i}\transpose\bx_{0j})\bx_{0j}\bx_{0j}\transpose(\bX\transpose\bX)^{-1}\bx_{0j}}{\bx_{0i}\transpose\bx_{0j}} + \bG_{0in}^{-1}\frac{1}{n^2}\sum_{j = 1}^n\frac{\bx_{0j}\bx_{0j}\transpose\bG_{0in}^{-1}\bx_{0j}}{(1 - \bx_{0i}\transpose\bx_{0j})^2},\\
\bb_i^{(\mathrm{base})}
& =  - \bG_{0in}^{-1}\frac{1}{n^2}\sum_{j = 1}^n\bigg\{\frac{-1}{(\bx_{0i}\transpose\bx_{0j})^2} + \frac{1}{(1 - \bx_{0i}\transpose\bx_{0j})^2}\bigg\}\bx_{0j}\bx_{0j}\transpose\bG_{0in}^{-1}\bx_{0j},\\
\bb_i^{(\mathrm{ASE})}
& = - \bG_{0in}^{-1}\frac{1}{n}\sum_{j = 1}^n\sum_{k = 1}^n\frac{\bx_{0j}x_{0ik}\beta_{jk}}{\bx_{0i}\transpose\bx_{0j}(1 - \bx_{0i}\transpose\bx_{0j})}\\
&\quad + \bG_{0in}^{-1}\frac{1}{n}\sum_{j = 1}^n\frac{2\bx_{0i}\transpose\bx_{0j} - 1}{\bx_{0i}\transpose\bx_{0j}(1 - \bx_{0i}\transpose\bx_{0j})}\bx_{0i}\transpose(\bX\transpose\bX)^{-1}\bx_{0i} + \bG_{in}^{-1}(\bX\transpose\bX)^{-1}\bx_{0i}\\
&\quad - \bG_{0in}^{-1}\frac{1}{n^2}\sum_{j = 1}^n\frac{(2\bx_{0i}\transpose\bx_{0j} - 1)\bx_{0j}\bx_{0i}\transpose\bSigma_{jn}\bx_{0i}}{(\bx_{0i}\transpose\bx_{0j})^2(1 - \bx_{0i}\transpose\bx_{0j})^2} - \bG_{0in}^{-1}\frac{1}{n^2}\sum_{j = 1}^n\frac{\bSigma_{jn}\bx_{0i}}{(\bx_{0i}\transpose\bx_{0j})^2(1 - \bx_{0i}\transpose\bx_{0j})^2},\\
\beta_{jk}& = \be_j\transpose\bigg(
\bI_n - \bu_k\bu_k\transpose - 
\sum_{m\in[d]\backslash\{k\}}\frac{\lambda_m\bu_m\bu_m\transpose}{\lambda_k - \lambda_m}\bigg)\mathrm{diag}\bigg\{\bigg(\sum_{b = 1}^n\frac{\bx_{0a}\transpose\bx_{0b}(1 - \bx_{0a}\transpose\bx_{0b})x_{0bk}}{\lambda_k^2}\bigg)_{a = 1}^n\bigg\},
\end{align*}
and $\bu_k$ is the eigenvector of $\bP_0 = \bX_0\bX_0\transpose$ associated with the $k$th largest eigenvalue $\lambda_k$.
\end{theorem}

\begin{remark}[Interpretation of the bias terms]
Each term in the $O(n^{-1})$ bias formula has an interesting interpretation. The term $[\beta_{j1},\ldots,\beta_{jd}]\transpose$ is precisely the $O(n^{-1})$ bias of the ASE and has been obtained in \cite{xie2024higher}. The term $\bb_i^{(\mathrm{ASE})}$ stems from the substitution of the unknown $\bx_{0j}$'s with the ASE $\widetilde{\bx}_j$'s and naturally connects to the $O(n^{-1})$ bias of the ASE. This term does not depend on whether the one-step estimator or the maximum surrogate likelihood estimator is used. It also vanishes in the hypothetical scenario where the oracle knowledge of $\bx_{0j}$'s is accessible. The term $\bb_i^{(\mathrm{base})}$ is intrinsic to the use of maximum likelihood principle and does not vanish even if the oracle maximum likelihood estimator (\emph{i.e.}, $\argmax_{\bx_i}\ell_{0in}(\bx_i)$) is used. The key difference of the $O(n^{-1})$ biases of these estimators lies in $\bb_i^{(\mathrm{MSLE})}$ and $\bb_i^{(\mathrm{OS})}$ , which explains how the former estimator retains more likelihood information than the latter estimator. 
\end{remark}

\begin{example}\label{example:second_order_bias_SBM}
Consider the following rank-one two-block stochastic block model with block probability matrix 
\begin{align*}
\bB = \begin{bmatrix}
p^2 & pq\\
pq & q^2
\end{bmatrix}
\end{align*} 
and cluster assignment function $\tau(\cdot)$ defined by $\tau(i) = 1$ if $i\in\{1,\ldots,n/2\}$ and $\tau(i) = 2$ if $i\in\{n/2 + 1,\ldots,n\}$ with $n = 300$. For this specific model, we compute the ratio of sum-of-squared biases $\sum_{i = 1}^n\mathrm{Bias}^2(\widehat{\bx}_i^{(\mathrm{OS})}) / \sum_{i = 1}^n\mathrm{Bias}^2(\widehat{\bx}_i)$ as a function of $p$ and $q$, where $p, q$ vary over $[0.05, 0.95]$ and visualize the ratio in \ref{fig:Second_order_bias_ratio_SBM}. The plot shows that the maximum surrogate likelihood estimator results in less bias compared to the one-step estimator for a broad range of $(p, q)$ values. 
\begin{figure}[htbp]
\centering
\includegraphics[width = 0.65\textwidth]{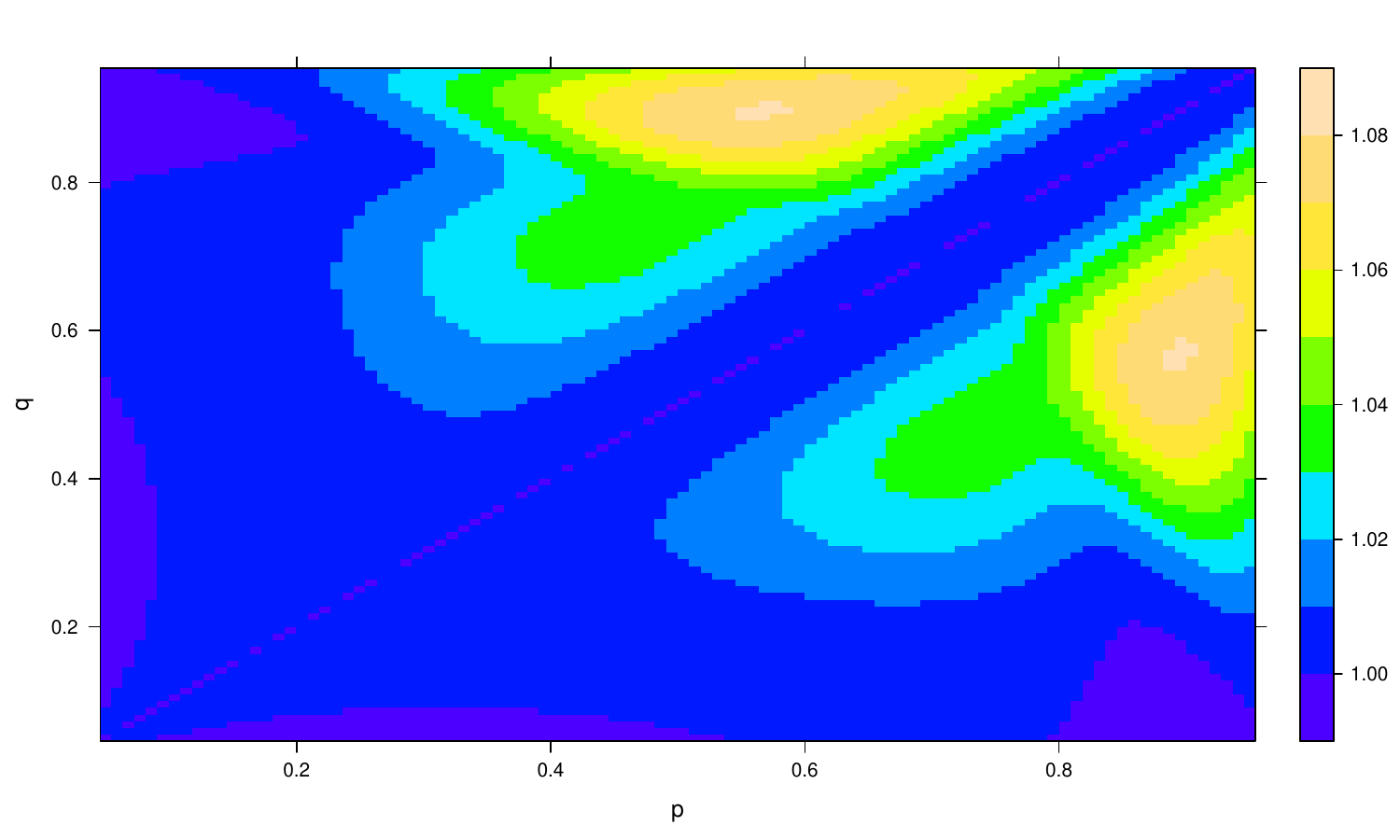}
\caption{Level plot of the ratio of  sum-of-squared biases $\sum_{i = 1}^n\mathrm{Bias}^2(\widehat{\bx}_i^{(\mathrm{OS})}) / \sum_{i = 1}^n\mathrm{Bias}^2(\widehat{\bx}_i)$ as a function of $p$ and $q$, where $p, q$ vary over $[0.05, 0.95]$ for Example \ref{example:second_order_bias_SBM}. }
\label{fig:Second_order_bias_ratio_SBM}
\end{figure}
\end{example}

\subsection{Computation details}
\label{sub:MSLE_computation}
This subsection discusses the detailed algorithm for computing the maximum surrogate likelihood estimator. For a given vertex $i\in [n]$, the estimator $\widehat{\bx}_i = \argmax_{\bx_i}\widetilde{\ell}_{in}(\bx_i)$ can be computed separately for each vertex $i\in [n]$. Thus, 
it is sufficient to design an algorithm for solving the optimization problem 
\begin{align}\label{eqn:MSLE_optimization}
\max_{\|\bx_i\|_2\leq 1}\frac{1}{n}\widetilde{\ell}_{in}(\bx_i).
\end{align}
Then the entire estimator $\widehat{\bX} = [\widehat{\bx}_1,\ldots,\widehat{\bx}_n]\transpose$ for all vertices can be obtained through a parallelization over $i\in [n]$.
Let us consider the optimization problem \eqref{eqn:MSLE_optimization}. 
Observe that the objective function $(1/n)\widetilde{\ell}_{in}(\bx_i)$ is concave and can be written in a sample average fashion, which motivates us to adopt the stochastic gradient descent algorithm \citep{robbins-monro-1951}.
Let $j_1,\ldots,j_s$ be independent $\mathrm{Unif}(1,\ldots,n)$ random variables, where $s\in\{1,\ldots,n\}$ is the so-called batch size, and for any $j\in [n]$, let
\[
m_i(\bx_i, j) = \frac{A_{ij} \widetilde\bx_j\transpose \bx_i}{\widetilde{p}_{ij}}
+ \widetilde{\bx}_j\transpose \bx_i
- \frac{1}{2\widetilde{p}_{ij}}\bx_i\transpose\widetilde{\bx}_j\widetilde{\bx}_j\transpose \bx_i
+ (1-A_{ij})\log(1-\bx_i\transpose \widetilde{\bx}_j).
\]
It is clear that for each $j_k$, $k\in [s]$, $m_i(\bx_i, j_k)$ can be viewed as a noisy measurement of the objective function $(1/n)\widetilde{\ell}_{in}(\bx_i)$ because $(1/n)\widetilde{\ell}_{in}(\bx_i) = \expect_{j_k}\{m_i(\bx_i, j_k)\}$. Then given a sequence of step sizes $\{\alpha_t\}_{t\geq 1}$ and a initial guess $\widehat{\bx}_i^{(0)}$, the stochastic gradient descent algorithm generates a sequence of iterates $\{\widehat{\bx}_i^{(t)}\}_{t\geq 1}$ using the updating scheme
\begin{align}\label{eqn:SGD_updating_scheme}
\widehat{\bx}_i^{(t + 1)} = \widehat{\bx}_i^{(t)} + \frac{\alpha_t}{s}\sum_{k = 1}^s\frac{\partial m_i}{\partial\bx}(\widehat{\bx}_i^{(t)}, j_k^{(t)}),
\end{align}
where $\{(j_1^{(t)},\ldots,j_s^{(t)})\}_{t\geq 1}$ are independent copies of $(j_1,\ldots,j_s)$. The advantage of the stochastic gradient descent method over the classical gradient descent algorithm is that, with a comparatively small batch size $s$, one only needs to compute $s$ gradient measurements of $m_i(\bx_i, j)$ rather than all the gradient measurements of $\{m_i(\bx_i, j)\}_{j = 1}^n$. This computational convenience is especially desired when the network contains large number of vertices. To implement the algorithm with adaptive step sizes, we follow the suggestion given by \cite{duchi2011adaptive} and \cite{li-orabona-2019} and take
\begin{align}\label{eqn:SGD_adaptive_stepsize}
\alpha_t = a_0\left\{b_0 + \sum_{l = 1}^{t - 1}\left\|\frac{1}{s}\sum_{k = 1}^s\frac{\partial m_i}{\partial\bx}(\widehat{\bx}_i^{(l)}, j_k^{(l)})\right\|_2^2\right\}^{-(\epsilon + 1/2)},
\end{align}
where $a_0,b_0>0$ and $0<\epsilon\leq1/2$ are constants. 

The key difference between our algorithm and the standard stochastic gradient descent algorithm is that the feasible region $\{\bx_i\in\mathbb{R}^d:\|\bx_i\|\leq 1\}$ is compact. Therefore, whenever an updated value $\widehat{\bx}_i^{(t + 1)}$ stays outside the feasible region, one repeats step-halving procedures until $\|\widehat{\bx}_i^{(t + 1)}\|\leq 1$. 
We present the detailed stochastic gradient descent algorithm for computing the maximum surrogate likelihood estimator in Algorithm \ref{alg:sgd-msle}, 
the convergence of which is guaranteed by Theorem \ref{theorem:SGD_MSLE_convergence} below. 
\begin{algorithm}[h]
\caption{Stochastic gradient descent for maximum surrogate likelihood estimation}
\label{alg:sgd-msle}
\begin{algorithmic}[1]
\State \textbf{Input:} The adjacency matrix $\bA = [A_{ij}]_{n\times n}$ and the embedding dimension $d$.

\State \textbf{Set:} Tuning parameters $a_0, b_0 > 0$, $\epsilon\in (0, 1/2]$, and batch size $1\leq s\leq n$.

\State Compute the spectral decomposition of the adjacency matrix
  $\bA = \sum_{i=1}^n \widehat\lambda_i \widehat\bu_i \widehat\bu_j\transpose$,
where $|\widehat{\lambda}_1| \geq |\widehat{\lambda}_2| \geq \ldots \geq |\widehat{\lambda}_n|$, and $\widehat{\textbf{u}}_i\transpose \widehat{\textbf{u}}_j = \mathbbm{1}(i=j)$ for all $i, j \in[n]$.

\State Compute the adjacency spectral embedding:
\[
  \widetilde{\mathbf{X}} = \widehat{\mathbf{X}}^{\mathrm{ASE}} = [\widehat{\mathbf{u}}_1, \ldots, \widehat{\mathbf{u}}_d] \cdot \mathrm{diag}(|\widehat{\lambda}_1|^{1/2}, \ldots, |\widehat{\lambda}_d|^{1/2}),
\]
and write $\widetilde{\mathbf{X}} = [\widetilde{\mathbf{x}}_1, \ldots, \widetilde{\mathbf{x}}_n]\transpose \in \mathbb{R}^{n\times d}$. Let = $\widetilde{p}_{ij} = \widetilde{\mathbf{x}}_i\transpose\widetilde{\mathbf{x}}_j$ for all $i,j\in[n]$.

\State For $i = 1,2,\ldots,n$

\State \quad Initialize $\widehat{\mathbf{x}}_i^{(1)} = \widetilde{\mathbf{x}}_i$.

\State \quad  Set the iteration counter $t = 1$.

\State \quad  While not converge

\State \quad  \quad Sample without replacement $j_1, j_2, \ldots, j_{s}\sim\mathrm{Unif}(1,2,\ldots,n)$.

\State \quad  \quad Compute the average gradient at $\widehat{\mathbf{x}}_i^{(t)}$
\[
\bar{\mathbf{g}}^{(t)}(\widehat{\mathbf{x}}_i^{(t)}) = \frac{1}{s}\sum_{k=1}^s \frac{\partial m_i}{\partial \mathbf{x}_i}(\mathbf{x}_i, j_k)\bigg\rvert_{\mathbf{x}_i = \widehat{\mathbf{x}}_i^{(t)}}.
\]


\State \quad\quad Compute the step size $\alpha_t$ using formula \eqref{eqn:SGD_adaptive_stepsize}.

\State \quad\quad Compute
$\widehat{\mathbf{x}}_i^{(t+1)} = \widehat{\mathbf{x}}_i^{(t)} + \alpha_t\bar{\mathbf{g}}^{(t)}(\widehat{\mathbf{x}}_i^{(t)})$.

\State \quad\quad If $\|\widehat{\bx}_i^{(t + 1)}\|_2 > 1$, then set $\alpha_t \longleftarrow \alpha_t/2$ and go to line 12

\State \quad\quad Set $t \longleftarrow t+1$.

\State \quad End While

\State End For

\State \textbf{Output: } The MSLE $\widehat{\mathbf{X}} = [\widehat{\mathbf{x}}_1, \ldots, \widehat{\mathbf{x}}_n]\transpose$.

\end{algorithmic}
\end{algorithm}

\begin{theorem}
\label{theorem:SGD_MSLE_convergence}
Let the vertex $i\in [n]$ be fixed and suppose $(1/n)\widetilde{\ell}_{in}(\bx_i)$ is well-defined. Assume that $\widehat{\bx}_i = \argmax_{\bx_i:\|\bx_i\|\leq 1}(1/n)\widetilde{\ell}_{in}(\bx_i)$ lies in the interior of $\{\bx_i\in\mathbb{R}^d:\|\bx_i\|\leq 1\}$. 
Then the sequence of iterates $\{\widehat{\bx}_i^{(t)}\}_{t\geq 1}$ generated by 
\eqref{eqn:SGD_updating_scheme} with step sizes $\{\alpha_t\}_{t\geq 1}$ given by \eqref{eqn:SGD_adaptive_stepsize} and step-halving converges to $\widehat{\bx}_i$ almost surely with regard to the distribution of $\{(j_1^{(t)},\ldots,j_s^{(t)})\}_{t\geq 1}$. 
\end{theorem}

\begin{remark}
\label{remark:ase-out-of-ball}
The surrogate log-likelihood function $\widetilde{\ell}_{in}(\bx_i)$ is well-defined only when $\bx\transpose\widetilde{\bx}_j < 1$ for all $j\in [n]$ because of the logarithm terms $\{\log(1 - \bx_i\transpose\widetilde{\bx}_j)\}_{j = 1}^n$. For sufficiently large $n$, the constraint is satisfied by requiring that $\|\bx_i\|_2\leq 1$ since the adjacency spectral embedding $\widetilde{\bX} = [\widetilde{\bx}_1,\ldots, \widetilde{\bx}_n]\transpose$ satisfies $\max_{j\in [n]}\|\widetilde{\bx}_j\|_2 < 1$ with high probability. However, this requirement may not hold in certain finite sample problems, in which case the surrogate log-likelihood function $\widetilde{\ell}_{in}(\bx_i)$ is no longer well-defined. 
This numerical issue can be practically addressed by the following smooth concatenation technique. Roughly speaking, for a fixed $j\in [n]$, when $1-\bx_i\transpose\widetilde\bx_j$ drops below a small threshold, we replace the objective function $(1/n)\widetilde{\ell}_{in}(\bx_i)$ by a quadratic function such that the two pieces of functions are concatenated smoothly. 
Formally, let $\tau>0$ be a small threshold and define
\begin{align}\label{eqn:smooth_concatenation}
    h_i(\bx_i, j) & = \left\{
    \begin{aligned}
    &m_i(\bx_i, j),&\quad&\text{if }1 - \bx_i\transpose\widetilde{\bx}_j \geq \tau,\\
    &\alpha_{ij}(\bx_i\transpose\widetilde{\bx}_j)^2 + \beta_{ij}(\bx_i\transpose\widetilde{\bx}_j) + \gamma_{ij},&\quad&\text{if }1 - \bx_i\transpose\widetilde{\bx}_j < \tau,
    \end{aligned}
    \right.
\end{align}
for each $j\in [n]$, where $\alpha_{ij},\beta_{ij},\gamma_{ij}$ are coefficients such that $h_i(\cdot, j)$ is twice continuously differentiable. Then the objective function $(1/n)\widetilde{\ell}_{in}(\bx_i)$ can be replaced by $(1/n)\sum_{j = 1}^nh_i(\bx_i, j)$ and the aforementioned stochastic gradient descent algorithm applies with $\partial m_i(\bx_i, j)/\partial\bx_i$ replaced by $\partial h_i(\bx_i, j)/\partial\bx_i$. 

\end{remark}

\section{Bayesian Estimation With Surrogate Likelihood}
\label{sec:Bayesian_Estimation_Surrogate_Likelihood}
This section explores Bayesian estimation of random dot product graphs with the proposed surrogate likelihood. Although \cite{xie-xu-2020-bayes} has established the minimax optimality of the Bayesian random dot product graph model with the exact likelihood,
the asymptotic shape of the posterior distribution is yet to be characterized because of the complicated structure of the exact likelihood function. In contrast, thanks to the separable and log-concave properties of the surrogate likelihood, we are able to completely characterize the asymptotic posterior distribution of the latent positions with the exact likelihood replaced by the surrogate. Formally, for any fixed vertex $i\in [n]$ and a prior distribution $\pi(\cdot)$ supported on $\{\bx\in\mathbb{R}^d:\|\bx\|_2\leq 1\}$, the posterior distribution of $\bx_i$ given $\bA$ with the surrogate log-likelihood function $\widetilde{\ell}_{in}(\bx_i)$ can be written as
\begin{align}\label{eqn:posterior_distribution_xi}
\widetilde{\pi}_{in}(\bx_i\mid\bA) = \frac{\exp\{\widetilde{\ell}_{in}(\bx_i)\}\pi(\bx_i)}{\int \exp\{\widetilde{\ell}_{in}(\bx_i)\}\pi(\bx_i)\mathrm{d}\bx_i}.
\end{align}
Then the joint posterior density of the entire latent position matrix $\bX = [\bx_1,\ldots,\bx_n]\transpose$ is taken as the product $\widetilde{\pi}_n(\bX\mid\bA) = \prod_{i = 1}^n\widetilde{\pi}_{in}(\bx_i\mid\bA)$ because the surrogate log-likelihood function is separable across different vertices.

When the exact likelihood function is not available or intractable for analysis or computation, 
the idea of using a general statistical criterion function to replace the likelihood in the Bayes formula is not entirely new, among which an influential work is \cite{chernozhukov-hong-2003}. There have also been several recent works addressing the large sample properties of the so-called quasi-posterior or Gibbs posterior distributions \citep{kleijn-vaart-2012, miller2021asymptotic, syring-martin-2018, syring-martin-2022}. One key difference is that unlike the well-specified exact posterior distributions, the frequentist coverage of the credible sets of the quasi-posterior distributions may not agree with their credibility level \citep{kleijn-vaart-2012}. Below, we show that, with the surrogate likelihood, the posterior distribution produces credible sets that have the correct frequentist coverage. This is achieved through the following Bernstein--von Mises theorem. 

\begin{theorem}\label{theorem:posterior_convergence_tvm}
Suppose
the conditions of Theorem \ref{theorem:existence_and_uniqueness} hold and the embedding dimension $d$ is fixed. Let $\pi(\cdot)$ be a prior density satisfying $c\leq\pi(\bx_i)\leq C$ and $|\pi(\bx)-\pi(\by)|\leq C'\|\bx-\by\|_2$ for any $\bx, \by$ with $\|\bx\|_2,\|\by\|_2\leq 1$ for some constants $0<c,C,C'<\infty$. Let 
$\bW$ be the $d\times d$ orthogonal matrix in Theorem \ref{theorem:asymptotic_properties_of_MSLE}. For any fixed vertex $i\in [n]$, let
$\widehat{\bx}_i = \argmax_{\bx_i:\|\bx_i\|_2\leq 1}\widetilde{\ell}_{in}(\bx_i)$,
$\bt = \sqrt{n}\bW\transpose(\bx_i - \widehat{\bx}_i)$, and $\widetilde{\pi}_{in}^*(\bt\mid\bA)$ be the density of $\bt$ induced from 
\eqref{eqn:posterior_distribution_xi}. Then for any $\alpha>0$,
\begin{align}\label{equation:posterior_convergence_tvm}
\max_{i\in[n]}\int (1 + \|\bt\|_2^\alpha)\left|\widetilde{\pi}_{in}^*(\bt\mid\bA) - \det(2\pi\bG_{0in}^{-1})^{-1/2}e^{-\bt\transpose\bG_{0in}\bt/2}\right|\mathrm{d}\bt\overset{\prob_0}{\to}0.
\end{align}
\end{theorem}

Below, Corollary \ref{corollary:posterior_inference} discusses the effect of Theorem \ref{theorem:posterior_convergence_tvm} on subsequent inference. In particular, it shows that for each vertex $i\in [n]$, the posterior mean has the same asymptotic distribution as the maximum surrogate likelihood estimator, and the asymptotic level-$\alpha$ credible set has the correct frequentist coverage probability.

\begin{corollary}
\label{corollary:posterior_inference}
Suppose the conditions of Theorem \ref{theorem:posterior_convergence_tvm} hold. For any $i\in [n]$, let $\bx_i^*=\int\bx_i\widetilde\pi_{in}(\bx_i\mid\bA)\mathrm{d}\bx_i$ and $\bSigma_{in}^*=\int(\bx_i-\bx_i^*)(\bx_i-\bx_i^*)\transpose\widetilde\pi_{in}(\bx_i\mid\bA)\mathrm{d}\bx_i$ be the posterior mean and covariance matrix of $\bx_i$, respectively, and $\bX^* = [\bx_1^*,\ldots,\bx_n^*]\transpose$. Let $q_{1-\alpha}$ be the $(1-\alpha)$ quantile of the $\chi^2_d$ distribution
and $\mathcal{C}_{in}(\alpha) = \{\bx_i:(\bx_{i}-\bx_i^*)\transpose(\bSigma_{in}^*)\inverse(\bx_{i}-\bx_i^*)\leq q_{1-\alpha}\}$ be the asymptotic $(1 - \alpha)$-credible set for $\bx_i$, where $\bW\in\mathbb{O}(d)$ is given in Theorem \ref{theorem:asymptotic_properties_of_MSLE}. Then
\[
\sqrt{n}\bG_{0in}\halfpower(\bW\transpose\bx_i^*-\rho_n\halfpower\bx_{0i}) \overset{\calL}{\to} \mathrm{N}_d(\zero_d,\bI_d)
\]
and
$
\prob_0\{\rho_n^{1/2}\bW\bx_{0i}\in\mathcal{C}_{in}(\alpha)\} \to 1-\alpha.
$
Furthermore, if $(\log n)^4/(n\rho_n)\to 0$, then 
\[
\|\bX^*\bW - \rho_n^{1/2}\bX_0\|_{\mathrm{F}}^2 - \frac{1}{n}\sum_{i = 1}^n\mathrm{tr}(\bG_{0in}^{-1}) \overset{\prob_0}{\to} 0.
\]
%
%
\end{corollary}

In practice, the posterior distribution based on the surrogate likelihood can be computed using a standard Metropolis--Hastings algorithm with parallelization over the vertices $i\in [n]$. The detailed algorithm is provided in the Supplementary Material. Note that in practice, we can also apply the smooth concatenation technique discussed in Remark \ref{remark:ase-out-of-ball} to the posterior computation by simply replacing the surrogate log-likelihood function $\widetilde{\ell}_{in}(\bx_i)$ in the Bayes formula \eqref{eqn:posterior_distribution_xi} by $\sum_{j = 1}^nh_i(\bx_i, j)$ defined in \eqref{eqn:smooth_concatenation}.

\section{Numerical Examples}
\label{section:numerical}

\subsection{A latent curve example}
\label{sub:latent_curve_example}
In this subsection, we study the empirical performance of the proposed estimation procedures through a simulated random dot product graph example, where the latent positions are generated from a one-dimensional curve.
Consider a random dot product graph with $n$ vertices and latent dimension $d = 1$. For each vertex $i\in [n]$, the latent position $x_{0i}$ for the $i$th vertex is set to $x_{0i} = 0.8\sin\{\pi(i-1)/(n-1)\} + 0.1$. Let $\bX_0 = [x_{01}, \ldots, x_{0n}]\transpose$, $n = 1000$. Given $\bA\sim\mathrm{RDPG}(\bX_0)$, we consider the following four estimation procedures for $\bX_0$: the adjacency spectral embedding (ASE), the one-step estimate (OSE), the maximum surrogate likelihood estimate (MSLE) obtained using the step-halving stochastic gradient descent algorithm,
and the Bayes estimate with the surrogate likelihood (BE). For the Bayes estimate, we use the uniform prior on the unit disk for all $\bx_i$. The Metropolis--Hastings sampler is implemented with parallelization over vertices $i\in [n]$, and each Markov chain contains $1000$ burn-in iterations and $2000$ post-burn-in samples with a thinning of $5$. The posterior mean is taken as the point estimate. The convergence diagnostics of the Markov chains are provided in the Supplementary Material, showing no signs of non-convergence.

The performance of the above estimates is investigated via the following two objectives: The recovery of the latent position matrix $\bX_0$; The empirical coverage probabilities of the vertex-wise confidence intervals based on the MSLE and credible intervals based on the BE. Specifically, for the first objective, given a generic estimate $\bar{\bX}$ for $\bX_0$, we use the sum of squared errors (SSEs) $\inf_{\bW\in\{\pm1\}}\|\bar{\bX}\bW - \bX_0\|\frobenius^2$ as the evaluation metric. For the second objective, we compute the vertex-wise asymptotic $95\%$ frequentist confidence intervals and Bayesian credible intervals. The vertex-wise $95\%$ confidence intervals based on the MSLE are computed as follows: 
Denote the $1-\alpha/2$ quantile of the standard normal distribution by $z_{1-\alpha/2}$. Then by Theorem \ref{theorem:asymptotic_properties_of_MSLE}, for each $i\in [n]$, the $(1-\alpha)$ confidence interval for $x_{0i}$ is
$(|\widehat{x}_i|-\{n\widehat\bG(\widehat{x}_i)\}^{-1/2}{z_{1-\alpha/2}}, |\widehat{x}_i|+ \{n\widehat\bG(\widehat{x}_i)\}^{-1/2}{z_{1-\alpha/2}})$, where $\widehat\bG_{in}(\widehat{x}_i)=(1/n)\sum_{j=1}^n\widehat{x}_j\{\widehat{x}_i(1-\widehat{x}_i\widehat{x}_j)\}\inverse$ is the plug-in estimate of the asymptotic variance. The vertex-wise $95\%$ credible intervals based on the posterior distribution with the surrogate likelihood function can be obtained directly from the Metropolis--Hastings samples. The same numerical experiment is repeated for $1000$ Monte Carlo replicates. 

\begin{table}[htbp]
  \centering
  \begin{tabular}{c | c c c c}
    \hline\hline
    Estimate & ASE & OSE & MSLE & BE\\
    \hline
    SSE & 0.4707 & 0.4592 & 0.4596  & 0.4608 \\
    Standard error for SSE & 0.0216 & 0.0209 & 0.0209 & 0.0210 \\
    \hline\hline
    Two-sample $t$-test & & ASE vs OSE & ASE vs MSLE & ASE vs BE \\
    \hline
    $p$-value & & $1.0\times 10^{-32}$ & $7\times 10^{-31}$ & $8\times 10^{-25}$ \\
    \hline\hline
  \end{tabular}
  \caption{The average SSEs, their standard errors, and the $p$-values of the two-sample $t$-tests between the SSEs of the ASE against the remaining three estimates for Section \ref{sub:latent_curve_example} with $n = 1000$. }
  \label{table:latent_curve_simulation}
\end{table}


For the first objective, the SSEs of the estimates are shown in Table \ref{table:latent_curve_simulation}. We can see that the SSEs of the adjacency spectral embedding is comparatively larger than those of the remaining competitors, while the likelihood-based estimates have smaller SSEs. The $p$-values of the pairwise two-sample $t$-tests among the SSEs of these estimates are tabulated in Table \ref{table:latent_curve_simulation} as well and they show that the differences between the ASE and the remaining likelihood-based estimates are statistically significant. 
This phenomenon empirically validate the conclusion that the likelihood-based estimates, namely, the OSE, the MSLE, and the BE, improve upon the the spectral-based adjacency spectral embedding. 

\begin{figure}[htbp]
\begin{subfigure}{0.5\linewidth}
\includegraphics[width=\linewidth]{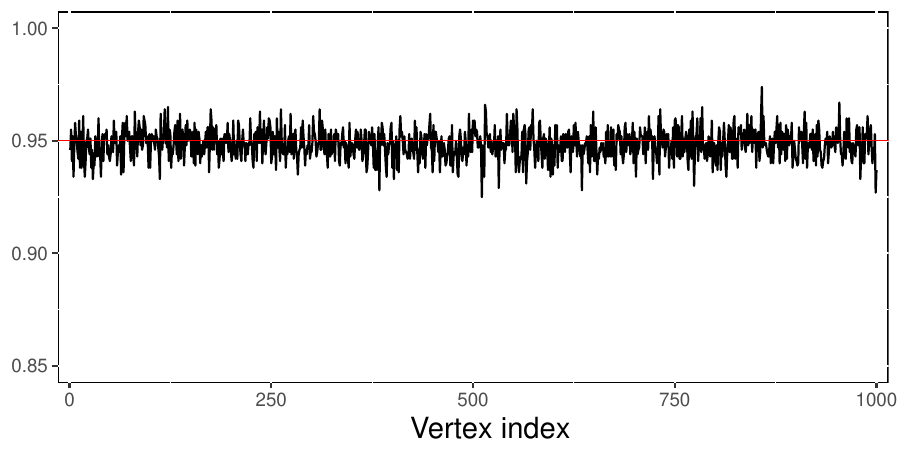}
\caption{Coverage probabilities of confidence intervals}
\end{subfigure}
\begin{subfigure}{0.5\linewidth}
\includegraphics[width=\linewidth]{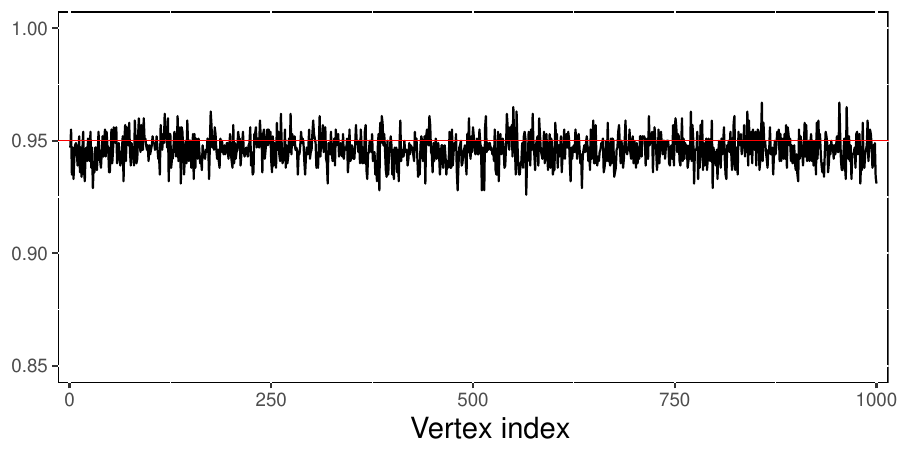}
\caption{Coverage probabilities of credible intervals}
\end{subfigure}
\caption{Numerical results for Section \ref{sub:latent_curve_example}: Panels (a) and (b) present the empirical coverage probabilities of the 95\% confidence intervals constructed based on the MSLE and the 95\% credible intervals constructed from the Metropolis--Hastings samples, respectively, where the red horizontal lines mark the 95\% nominal coverage probability.}
\label{figure:example-1}
\end{figure}

For the second objective, Figure \ref{figure:example-1} (a) and (b) visualize the empirical coverage probabilities of the vertex-wise $95\%$ confidence intervals based on the MSLE and the vertex-wise $95\%$ Bayesian credible intervals across the $1000$ Monte Carlo replicates, respectively. 
It is clear that the empirical coverage probabilities of these confidence intervals and credible intervals are close to the nominal $95\%$ level, validating the theory developed in Section \ref{sec:MSLE} and Section \ref{sec:Bayesian_Estimation_Surrogate_Likelihood}.

In addition to the above investigation in a large sample regime with $n = 1000$, we also explore the performance of the proposed estimation methods in a comparatively small sample regime with $n = 30$. Here, we focus on the performance of different estimates using the SSE as the evaluation metric. Besides the aforementioned four estimates, we also consider the maximum likelihood estimate (MLE). Note that although the theory of the MLE is still open, it is always possible to find a local maximizer of the likelihood function using any optimization toolkit. Here we use the \texttt{R} built-in \texttt{optim} function in practice. We repeat the same numerical experiment for $1000$ independent Monte Carlo replicates, visualize the boxplots of the SSEs in Figure \ref{fig:latent_curve_small_sample}, and tabulate the numeric values of the summary statistics of these SSEs in Table \ref{table:latent_curve_simulation_n30}.
We can see that in this small sample scenario, the MSLE and the OSE do not outperform the baseline ASE and the MLE as they have larger SSEs, while the BE has the least SSEs. The $p$-values of the pairwise $t$-test of the SSEs of the BE against those of the remaining competitors are reported in Table \ref{table:latent_curve_simulation_n30} as well, showing that the differences between BE and the other competitors are statistical significant. This observation shows the potential advantage of the Bayesian estimation procedure based on the Markov chain Monte Carlo sampling algorithm over the classical optimization-based estimation methods for finite-sample problems in practice.

\begin{table}[htbp!]
  \centering
  \begin{tabular}{c | c c c c c}
    \hline\hline
    Estimate & ASE & OSE & MSLE & MLE & BE\\
    \hline
    SSE & 0.4594 & 0.4608 & 0.5739 & 0.4451 & 0.3886\\
    Standard error for SSE & 0.1204 & 0.1439 & 0.2409 & 0.1139 & 0.1079\\
    {Computation time (seconds)} & $9\times10^{-3}$ & $4.8\times10^{-2}$ & 139 & $5.1$ & 148\\
    \hline\hline
    Two-sample $t$-test & BE vs ASE & BE vs OSE & BE vs MSLE & BE vs MLE &\\
    \hline
    $p$-value & $1.1\times 10^{-41}$ & $1.9\times 10^{-35}$ & $1.2\times 10^{-93}$ & $3.7\times 10^{-29}$ & \\
    \hline\hline
  \end{tabular}
  \caption{The average SSEs, their standard errors, and the $p$-values of the two-sample $t$-tests between the SSEs of the BE against the remaining estimates for Section \ref{sub:latent_curve_example} with $n = 30$, }
  \label{table:latent_curve_simulation_n30}
\end{table}

\subsection{A rank-two random dot product graph example}
\label{sub:a_rank_two_random_dot_product_graph_example}

We now consider a rank-two random dot product graph with $n = 300$ vertices and latent dimension $d = 2$, where the latent positions $\bX_0 = [\bx_{01},\ldots,\bx_{0n}]\transpose$ are given by $\bx_{0i} = [0.15\sin\{\pi(i - 1)/(n - 1)\} + 0.6, 0.15\cos\{\pi(i - 1)/(n - 1)\} + 0.6]\transpose$. Similar to Section \ref{sub:latent_curve_example}, for a given $\bA\sim\mathrm{RDPG}(\bX_0)$, we also implement the adjacency spectral embedding (ASE), the one-step estimate (OSE), the maximum surrogate likelihood estimate (MSLE), and the Bayes estimate with the surrogate likelihood (BE). The implementation details are the same as those of \ref{sub:latent_curve_example}, and we compute the sum-of-squared errors (SSEs) $\inf_{\bW\in\mathbb{O}(d)}\|\bar{\bX}\bW - \bX_0\|_{\mathrm{F}}^2$ as the evaluation metric. The same numerical experiment is repeated for $1000$ Monte Carlo replicates. The average SSEs, their standard errors, and the computation times across repeated experiments are summarized in Table \ref{table:Rank2_RDPG_simulation}. 
\begin{table}[htbp]
  \centering
  \begin{tabular}{c | c c c c}
    \hline\hline
    Estimate & ASE & OSE & MSLE & BE\\
    \hline
    SSE & 14.94 & 21.70 & 13.34 & 10.91\\
    Standard error for SSE & 0.4326 & 8.057 & 0.5929 & 0.6229 \\
    Computation time (seconds) &  0.005 & 0.01 & 4.80 & 41.43 \\
    \hline\hline
  \end{tabular}
  \caption{The average SSEs, their standard errors, and the computation times for Section \ref{sub:a_rank_two_random_dot_product_graph_example} with $n = 300$. }
  \label{table:Rank2_RDPG_simulation}
\end{table}
The differences in SSEs are statistically significant. It is clear that the Bayes estimate with the surrogate likelihood results in the best performance in terms of SSE, while both the maximum surrogate likelihood estimate and the Bayes estimate outperform the baseline adjacency spectral embedding and the one-step estimate. Note, nonetheless, that the performance improvement of the proposed methods is at the cost of additional computation times.



\begin{figure}[t]
\centering
\includegraphics[width=0.5\linewidth]{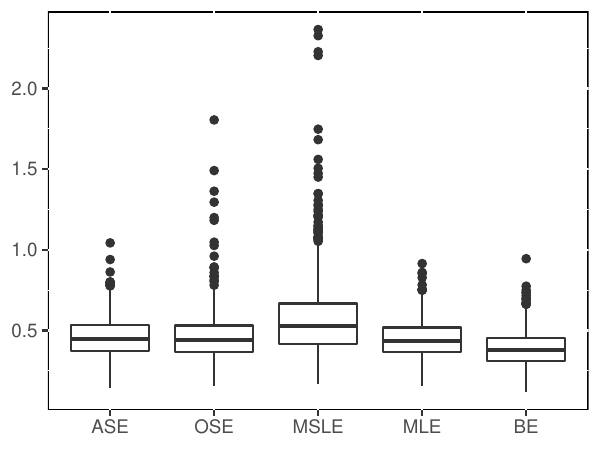}
\caption{Boxplots of the sum of squared errors (SSEs) $\inf_{\bW\in\{\pm1\}}\|\bar{\bX}\bW - \bX_0\|\frobenius^2$ for Section \ref{sub:latent_curve_example} with $n = 30$.}
\label{fig:latent_curve_small_sample}
\end{figure}

\subsection{A stochastic block model example}
\label{sub:SBM_example}
We now consider a stochastic block model in the context of a random dot product graph.
The latent dimension is $d = 2$, the number of communities is $K = 5$, and the unique latent positions are $\bv_1 = [0.3, 0.3]\transpose, \bv_2 = [0.5, 0.5]\transpose, \bv_3 = [0.7, 0.7]\transpose, \bv_4 = [0.3, 0.7]\transpose$, and $\bv_5 = [0.7, 0.3]\transpose$. 
The cluster assignments of the vertices $(z_i)_{i = 1}^n$ are drawn from a categorical distribution with probability vector $[1/K, \ldots, 1/K]\transpose$ and we set $\bx_{0i} = \bv_{z_i}$, $i\in [n]$. Note that $\bv_3$ is very close to the boundary of the parameter space. Let $\bX_0 = [\bx_{01}, \ldots, \bx_{0n}]\transpose$ and suppose an adjacency matrix $\bA$ is generated from $\mathrm{RDPG}(\bX_0)$. We consider two sample sizes, $n=2000$ and $n=3000$.

We consider the performance of the same estimates as in Section \ref{sub:latent_curve_example} given a realization $\bA\sim\mathrm{RDPG}(\bX_0)$: the ASE, the OSE, the MSLE 
and the BE with the surrogate likelihood computed using the Metropolis--Hastings sampler.
For the MSLE, we implement the step-halving stochastic gradient descent algorithm
with the batch size set to $s = 500$ and $s = n$ (giving rise to the classical gradient descent algorithm) to compare the computational costs. 
The setup of the Metropolis--Hastings sampler for the Bayesian estimation is the same as in Section \ref{sub:latent_curve_example}, and the convergence diagnostics are provided in the Supplementary Material. We take the posterior mean as the point estimate as before. The same experiment is repeated for $1000$ independent Monte Carlo replicates. 

Similar to Section \ref{sub:latent_curve_example}, given a generic estimate $\bar{\bX}$, we compute the SSEs of the estimates $\inf_{\bW\in\mathbb{O}(2)}\|\bar{\bX}\bW - \bX_0\|\frobenius^2$ to measure the estimation accuracy. 
%
The summary statistics of these results are visualized in Table \ref{table:SBM_simulation_2000} (for $n = 2000$) and Table \ref{table:SBM_simulation_3000} (for $n = 3000$), respectively. We see that the OSE is numerically unstable because $v_3$ is close to the boundary of the parameter space. Overall, the BE outperforms the other competitors with the least errors, while the ASE and the MSLE have similar performance in terms of the estimation error. The $p$-values of the pairwise two-sample $t$-tests among the SSEs of these estimates are reported in Table \ref{table:SBM_simulation_pvalue_2000} (for $n = 2000$) and Table \ref{table:SBM_simulation_pvalue_3000} (for $n = 3000$), showing that the differences between the BE and the remaining competitors are statistically significant. This phenomenon suggests that, when some latent positions are close to the boundary of the parameter space, the Bayesian estimation method based on the Markov chain Monte Carlo sampler is numerically more stable than the optimization-based frequentist ASE and the MSLE.
\begin{table}[t!]
  \centering
  \begin{tabular}{c | c c c c c}
    \hline\hline
    Estimate & ASE & OSE & MSLE-SGD & MSLE-GD & BE\\
    \hline
    SSEs & 8.570 & 31.646 &  8.510 &  8.513 &  7.970 \\
    Standard errors of SSEs & 0.250 & 18.242 &  0.250 &  0.250 &  0.378 \\
    Computation time (seconds) & 0.240 &  0.425 &  9.920 & 18.028 & 98.524 \\
    \hline\hline
  \end{tabular}
  \caption{Numerical results for Section \ref{sub:SBM_example}: The average sum of squared errors, their standard errors, and the computation time of a single experiment (in seconds). MSLE-SGD and MSLE-GD refer to the MSLE computed by the stochastic gradient descent with batch size being $500$ and the classical gradient descent. Sample size $n=2000$.}
  \label{table:SBM_simulation_2000}
\end{table}
\begin{table}[t!]
  \centering
  \begin{tabular}{c | c c c c c}
    \hline\hline
    Estimate & ASE & OSE & MSLE-SGD & MSLE-GD & BE\\
    \hline
    SSEs & 8.365 & 44.079 & 8.323 & 8.325 & 7.892 \\
    Standard errors of SSEs & 0.227 & 34.972 &  0.225 & 0.226 & 0.410 \\
    Computation time (seconds) & 0.599 & 1.127 & 14.666 & 28.760 & 225.676 \\
    \hline\hline
  \end{tabular}
  \caption{Numerical results for Section \ref{sub:SBM_example}: The average sum of squared errors, their standard errors, and the computation time of a single experiment (in seconds). MSLE-SGD and MSLE-GD refer to the MSLE computed by the stochastic gradient descent with batch size being $500$ and the classical gradient descent. Sample size $n=3000$.}
  \label{table:SBM_simulation_3000}
\end{table}

The computation times of a single experiment for different estimation procedures are reported in Table \ref{table:SBM_simulation_2000} and \ref{table:SBM_simulation_3000}. We see that the ASE and the OSE are faster to compute, whereas the MSLE obtained through the classical gradient descent algorithm and the BE are more computationally expensive. We also observe that the stochastic gradient descent algorithm is significantly faster than the classical gradient descent algorithm for finding the MSLE and gains computational efficiency at the cost of estimation accuracy compared to the BE. 
\begin{table}[t]
  \centering
  \begin{tabular}{c | c c c c }
    \hline\hline
    Two-sample $t$-test & BE vs ASE & BE vs OSE & BE vs MSLE-SGD & BE vs MSLE-GD\\
    \hline
    $p$-value & $6\times 10^{-28}$ & $9\times 10^{-17}$ & $3\times 10^{-24}$  & $2\times 10^{-24}$\\
    \hline\hline
  \end{tabular}
  \caption{$p$-values of the two-sample $t$-tests between the SSEs of the BE against the remaining estimates for Section \ref{sub:SBM_example}. Sample size $n=2000$.}
  \label{table:SBM_simulation_pvalue_2000}
\end{table}

The computation times of a single experiment for different estimation procedures are reported in Table \ref{table:SBM_simulation_2000} and \ref{table:SBM_simulation_3000}. We see that the ASE and the OSE are faster to compute, whereas the MSLE obtained through the classical gradient descent algorithm and the BE are more computationally expensive. We also observe that the stochastic gradient descent algorithm is significantly faster than the classical gradient descent algorithm for finding the MSLE and gains computational efficiency at the cost of estimation accuracy compared to the BE. 
\begin{table}[t]
  \centering
  \begin{tabular}{c | c c c c }
    \hline\hline
    Two-sample $t$-test & BE vs ASE & BE vs OSE & BE vs MSLE-SGD & BE vs MSLE-GD\\
    \hline
    $p$-value & $1.1\times 10^{-18}$ & $1.9\times 10^{-17}$ & $2.3\times 10^{-16}$  & $1.7\times 10^{-16}$\\
    \hline\hline
  \end{tabular}
  \caption{$p$-values of the two-sample $t$-tests between the SSEs of the BE against the remaining estimates for Section \ref{sub:SBM_example}. Sample size $n=3000$.}
  \label{table:SBM_simulation_pvalue_3000}
\end{table}

\subsection{Analysis of Wikipedia Graph Dataset}
\label{subsection:wikipedia-graph}
In this section, we apply the proposed surrogate likelihood estimation methods to a real-world Wikipedia graph dataset. The network data is structured as follows: The vertices represent $1382$ Wikipedia articles that are connected to the article named Algebraic Geometry within two hyperlinks, and an edge is assigned to link two articles if they are connected by a hyperlink. Besides the network itself, each Wikipedia article is also assigned with one of the following six class labels: people, places, dates, things, math and category. The dataset is  publicly available at at \url{http://www.cis.jhu.edu/~parky/Data/data.html}. 

The goal is to study the clustering accuracy using different estimates when the embedding dimension varies. Given a selected embedding dimension $d\geq 1$, we consider the following four estimates: the ASE, the OSE, the MSLE computed using the step-halving stochastic gradient descent algorithm, and the BE based on the surrogate likelihood (we consider the posterior mean as the point estimate) with the uniform prior distribution on the unit disk for all $\bx_i$. Unlike the scenarios in the simulated examples in Sections \ref{sub:latent_curve_example} and \ref{sub:SBM_example}, for this real dataset, the underlying ground truth of the latent positions is unknown. Rather, only the class labels of the vertices are available to us. To this end, we follow the suggestion in \cite{tang-priebe-2018} and apply the Gaussian-mixture-model-based clustering to the aforementioned four estimates. Namely, these estimates are regarded as the input for learning the clustering structure of the Wikipedia article network. We report the clustering accuracy using the Rand index \citep{rand-1971} as the evaluation metric. 

The Rand indices of the clustering results using different estimates across different embedding dimensions $d\in \{1,2,\ldots,10\}$ are shown in Table \ref{table:Wikipedia_RI}. On one hand, we can see that when $d\leq 2$, the adjacency spectral embedding yields better clustering accuracy with a higher Rand index value than the remaining competitors. On the other hand, as the embedding dimension $d$ increases from $2$ to $10$, the MSLE and the BE with the surrogate likelihood outperform the other two competitors. A plausible explanation of this phenomenon could be that the eigenvectors of the adjacency matrix with smaller eigenvalues are noisier than the top two eigenvectors, but this source of noise is reduced through the additional information introduced by the surrogate likelihood function. 
  
\begin{table}
  \centering
  \begin{tabular}{c c c c c c c c c c c}
    \hline\hline
    $d$ & 1 & 2 & 3 & 4 & 5 & 6 & 7 & 8 & 9 & 10\\
    \hline
    ASE  & 0.745 & 0.720 & 0.721 & 0.723 & 0.731 & 0.735 & 0.736 & 0.721 & 0.723 & 0.715 \\
    OSE  & 0.697 & 0.706 & 0.724 & 0.728 & 0.735 & 0.739 & 0.739 & 0.741 & 0.744 & 0.740 \\
    MSLE & 0.723 & 0.711 & 0.726 & 0.736 & 0.739 & 0.744 & 0.742 & 0.744 & 0.742 & 0.747 \\
    BE   & 0.718 & 0.715 & 0.724 & 0.735 & 0.735 & 0.742 & 0.743 & 0.744 & 0.744 & 0.745 \\
    \hline\hline
  \end{tabular}
  \caption{Numerical results of Wikipedia graph data: Rand indices between the class labels and the clustering results based on the four estimates, across embedding dimensions $d$ from $1$ to $10$, respectively. 
  }
  \label{table:Wikipedia_RI}
\end{table}



\subsection{Analysis of Political Blogs Network}
\label{sub:additional_real_world_data_examples}
We now consider the political blogs network \citep{10.1145/1134271.1134277}, a benchmark network data that has also been analyzed by \cite{PhysRevE.83.016107,10.1214/12-AOS1036,10.1214/13-AOS1138,10.1214/14-AOS1265,10.1111/rssb.12117,10.1214/15-AOS1360}. The network corresponds to the hyperlinks of blogs regarding U.S. politics after the 2004 presidential election. These blogs are manually classified as either liberal or conservative, which we use as the ground true communities. After following the rule of thumb by extracting the largest connected component and converting the resulting network with undirected edges, we obtain an $1224\times 1224$ adjacency matrix with $33430$ edges. We implement the proposed maximum surrogate likelihood estimate and the associated Bayes estimate, together with the adjacency spectral embedding and the one-step estimate as the competitors. We choose the embedding dimension to be $d = 2$ (the same as the number of clusters). Similar to the treatment in Section \ref{subsection:wikipedia-graph}, these latent position estimates are then applied to the Gaussian-mixture-model-based clustering, which we compare against the true community labels via the adjusted Rand index (ARI). 
Since the political blogs network is known to be closer to a degree-corrected stochastic block model (DCSBM) as opposed to the stochastic block model, we also consider the clustering algorithms designed for DCSBM after obtaining the latent position estimates. Specifically, we apply the spherical $k$-means \citep{lei-rinaldo-2015,lyzinski-sussman-tang-athreya-priebe-2014} and the spectral clustering on ratios-of-eigenvectors (SCORE) \citep{10.1214/14-AOS1265}. Note that in order to apply SCORE to latent positions estimates, we first compute the left singular vector matrix $\widehat{\bU}$ of an estimated latent position matrix $\widehat{\bX}$ and then apply SCORE to the orthonormal matrix $\widehat{\bU}$.
See Table \ref{table:polblogs_network} below for the detailed comparison, together with the computation time. When the clustering method is based on the Gaussian mixture model, the proposed Bayes estimate (BE) associated with the surrogate likelihood is more accurate in terms of recovering the liberal-versus-conservative community structure of these political blogs, although its computational cost is much more expensive compared to the maximum surrogate likelihood estimate (MSLE). In this case, both the BE and the MSLE result in significantly better ARI compared to the ASE and the OSE. When the clustering method is either the spherical $k$-means or SCORE, the ARI differences are marginal. Note that among all different clustering methods applied to different latent position estimates, SCORE applied to the ASE and the MSLE gives the best clustering results compared to other competitors, and they outperform the Gaussian-mixture-model-based clustering by a large margin. 
\begin{table}[t!]
  \centering
  \begin{tabular}{c | c c c c}
    \hline\hline
    \diagbox{Clustering method}{Estimate} & ASE & OSE & MSLE & BE\\
    \hline
    Gaussian mixture model & 0.1321 & 0.0416 & 0.4439 & 0.4660 \\
    Spherical $k$-means & 0.8104 & 0.8046 & 0.7726 & 0.8075\\
    SCORE & 0.8193 & 0.7900 & 0.8193 & 0.8134\\
    \hline
    Computation time (seconds) & 0.14 & 0.16 & 106.41 & 54373 \\
    \hline\hline
  \end{tabular}
  \caption{ARI and the computation times for Section \ref{sub:additional_real_world_data_examples}.}
  \label{table:polblogs_network}
\end{table}

\section{Discussion}
\label{section:last}
In this paper, we propose a novel surrogate likelihood estimation framework for random dot product graphs. 
The surrogate likelihood has several fascinating properties, including the separability and the log-concavity, that facilitate theoretical analyses and practical computation. 
We study the maximum surrogate likelihood estimation from the frequentist perspective and the Bayesian estimation using the surrogate likelihood. In particular, we establish the existence, uniqueness, and asymptotic normality of the maximum likelihood estimator, and propose a convenient stochastic gradient descent algorithm for the computation. 
Furthermore, we derive the $O(n^{-1})$ biases of the maximum surrogate likelihood estimator and the one-step estimator. These formulae illustrate how the former outperforms the latter in finite sample problems. 
We also establish the Bernstein--von Mises theorem of the posterior distribution with the surrogate likelihood function and show that the resulting credible sets have the correct frequentist coverage probabilities. It turns out that the maximum surrogate likelihood estimator and the Bayes estimator are asymptotically efficient in the sense of local efficiency \citep{xie-xu-2021-onestep}, and they outperform the baseline adjacency spectral embedding in terms of smaller asymptotic mean-squared errors. Our numerical examples also suggest that the proposed surrogate likelihood methodology is more favorable than the previously developed one-step estimator \citep{xie-xu-2021-onestep} in some finite sample problems. In particular, we have observed that for networks with comparatively small and moderate sizes, the empirical improvement of the maximum surrogate likelihood estimates over the one-step estimates, among other competitors, is more significant. Intuitively, such a phenomenon can be partially explained by the $O(n^{-1})$ bias derived in Theorem \ref{thm:second_order_bias} and Example \ref{example:second_order_bias_SBM}. When $n$ is small or moderate, the $O(n^{-1})$ bias difference is more observable as opposed to the case when $n$ is large. For example, in Section \ref{section:numerical}, the outperformance of the maximum surrogate likelihood estimates over the one-step estimates is more significant when $n = 30$ (Table \ref{table:latent_curve_simulation_n30}) and $300$ (Table \ref{table:Rank2_RDPG_simulation}) than the case when $n\in\{1000, 2000, 3000\}$.


Our current methodology and theory are designed for random dot product graphs with positive semidefinite edge probability matrices. These networks can only model the so-called assortative mixing networks and exclude many interesting examples, such as disassortative mixing stochastic block models with larger between-community connection probabilities and smaller within-community connection probabilities. Extending the current framework to random graphs with possibly indefinite edge probability matrices is straightforward by introducing the generalized random dot product graphs \citep{https://doi.org/10.1111/rssb.12509}, where $\expect A_{ij} = \bx_i\transpose\bI_{p, q}\bx_j$ and $\bI_{p, q}$ is a diagonal matrix whose first $p$ diagonals are $+1$ and the remaining $q$ diagonals $-1$ with $d = p + q$. The generalized random dot product graphs allow disassortative mixing networks. The trick is to replace the ASE $\widetilde{\bX} = [\widetilde{\bx}_1,\ldots,\widetilde{\bx}_n]\transpose$ with the sign-adjusted ASE $\widetilde{\bY} = [\widetilde{\by}_1,\ldots,\widetilde{\by}_n]\transpose = [|\widehat{\lambda}_1|\widehat{\bu}_1,\ldots,|\widehat{\lambda}_d|\widehat{\bu}_d]$, where $\widehat{\lambda}_1,\ldots,\widehat{\lambda}_d$ are the largest $d$-eigenvalues of $\bA$ in absolute value and $\widehat{\bu}_1,\ldots,\widehat{\bu}_d$ are the associated eigenvectors. The associated theory and computation methods can be easily extended accordingly. 

In Section \ref{sub:comparison_with_OSE}, we have seen that the one-step estimator also corresponds to the maximizer of an approximate likelihood function, but it has a worse approximation quality than the proposed surrogate likelihood near the oracle maximum likelihood estimator. Surprisingly, under a framework of generalized estimating equations proposed by \cite{xie-wu-2022-ea-spn}, the gradients of both the surrogate log-likelihood function and the approximate log-likelihood function associated with the one-step estimator can be viewed as some generalized estimating equations that take advantage of the likelihood function information.
 This intuition conforms to the fact that the estimators based on the approximation of likelihood are asymptotically equivalent up to the first order. We have also found in some finite sample problems that the maximum surrogate likelihood estimator outperforms the one-step estimator, which can be explained by the difference of their $O(n^{-1})$ bias. However, a more systematic way to study the performance difference of these two estimators requires the analysis of $O(n^{-3/2})$ mean-squared errors by following the spirit of \cite{PFANZAGL19781} and \cite{https://doi.org/10.1111/j.1468-0262.2004.00482.x}, which in turn requires the analysis of $O(n^{-3/2})$ behavior of the ASE beyond \cite{xie2024higher}. This is an interesting direction that we defer to future research.



\appendix
\section{Preliminary Results for the proofs}

\allowdisplaybreaks

\begin{lemma}
\label{lemma:two_to_infinity_norm_ASE}
Let $\bA\sim\mathrm{RDPG}(\rho_n^{1/2}\bX_0)$ with $n\rho_n\gtrsim \log n$. Denote by $\bDelta_n = (1/n)\bX_0\transpose{}\bX_0$. Assume $\lambda_d(\bDelta_n)\geq\lambda$ for some constant $\lambda > 0$ for all sufficiently large $n$, and $\min_{i,j\in [n]}(\bx_{0i}\transpose\bx_{0j}, 1 - \bx_{0i}\transpose\bx_{0j}) \geq\delta$ for some constant $\delta > 0$. Then for all $c > 0$, there exists some constant $N_{c,\lambda}\in\mathbb{N}_+$ depending on $c,\lambda$, such that for all $n\geq N_{c,\lambda}$, 
\begin{align*}
\|\widetilde\bX\bW - \rho_n^{1/2}\bX_0\|_{2\to\infty}
&\lesssim_{c,\lambda}\sqrt{\frac{\log n}{n}}.
\end{align*}
with probability at least $1 - n^{-c}$.
\end{lemma}

\begin{proof}
Denote by $\kappa(\bDelta_n) = \lambda_1(\bDelta_n)/\lambda_d(\bDelta_n)$. 
By Corollary 4.1 in \cite{xie-2022-entrywise}, for all $c > 0$, we can pick a constant $N_c\in\mathbb{N}_+$ such that for all $n\geq N_c$, with probability at least $1 - n^{-c}$,
\begin{align*}
\|\widetilde\bX\bW - \rho_n^{1/2}\bX_0\|_{2\to\infty}
&\lesssim_c \frac{\|\bU_\bP\|_{2\to\infty}}{(n\rho_n)^{1/2}\lambda_d(\bDelta_n)^{2}}\max\left\{
\frac{(\log n)^{1/2}}{\lambda_d(\bDelta_n)^2}, \frac{\kappa(\bDelta_n)}{\lambda_d(\bDelta_n)^2}, \log n
\right\}
\\&\quad
 +  \frac{(\log n)^{1/2}\|\bU_\bP\|_{2\to\infty}}{\lambda_d(\bDelta_n)^{1/2}}.
\end{align*}
Observe that $\lambda_d(\bDelta_n)$ is lower bounded by a constant $\lambda > 0$ for sufficiently large $n$, and $\lambda_1(\bDelta_n)\leq (1/n)\|\bX_0\|_{\mathrm{F}}^2\leq 1$. Also note that
\[
\|\bU_\bP\|_{2\to\infty} \leq \|\rho_n^{1/2}\bX_0\|_{2\to\infty}\|\bS_\bP^{-1/2}\|_{2}\leq 
\sqrt{\frac{\rho_n}{n\rho_n\lambda_d(\bDelta_n)}}\leq \frac{1}{\sqrt{n\lambda}}.
\]
Therefore, by the fact that $(\log n)/(n\rho_n)$ is bounded, we can pick a constant $N_{c,\lambda} \in\mathbb{N}_+$ depending on $c,\lambda$, such that for all $n\geq N_{c,\lambda}$, with probability at least $1 - n^{-c}$,
\begin{align*}
\|\widetilde\bX\bW - \rho_n^{1/2}\bX_0\|_{2\to\infty}
&\lesssim_c \|\bU_\bP\|_{2\to\infty}\frac{\log n}{(n\rho_n)^{1/2}\lambda^5} + \|\bU_\bP\|_{2\to\infty}\frac{(\log n)^{1/2}}{\lambda^{1/2}}\lesssim_{c,\lambda}\sqrt{\frac{\log n}{n}}.
\end{align*}
This completes the proof.
\end{proof}

\begin{lemma}[Some frequently used results]
\label{lemma:some_frequently_used_results}
Suppose $\bA\sim\mathrm{RDPG}(\rho_n^{1/2}\bX_0)$ and assume the conditions of Theorem \ref{theorem:existence_and_uniqueness} hold. Denote by $\widetilde{p}_{ij} = \widetilde{\bx}_i\transpose\widetilde{\bx}_j$, $i,j\in [n]$. Then for any $c>0$, there exists a constant $N_{c,\delta,\lambda}\in\mathbb{N}_+$ depending on $c, \delta, \lambda$ such that for all $n\geq N_{c,\delta,\lambda}$, the following hold with probability at least $1-n^{-c}$:
\begin{align*}
&\max_{j\in[n]}\|\widetilde\bx_j\|_2 \leq \rho_n(1-\frac{\delta}{2}),\\
&\max_{i,j\in[n]}|\widetilde p_{ij} - \rho_n\bx_{0i}\transpose\bx_{0j}| \lesssim_{c,\lambda} \rho_n\halfpower\sqrt{\frac{\log n}{n}},\\
&\frac{\rho_n\delta}{2} \leq \min_{i,j\in[n]}\widetilde p_{ij} \leq \max_{i,j\in[n]}\widetilde p_{ij} \leq \rho_n(1-\frac{\delta}{2}),\\
&\max_{j\in[n]}\|\bW\transpose\widetilde\bx_j\widetilde\bx_j\transpose\bW - \rho_n\bx_{0j}\bx_{0j}\transpose\|_2 \lesssim_{c,\lambda} \rho_n\halfpower\sqrt{\frac{\log n}{n}}.
\end{align*}
\end{lemma}

\begin{proof}
For the first result, by Lemma \ref{lemma:two_to_infinity_norm_ASE} and the condition that $\frac{\log n}{n\rho_n}\to 0$, we can pick a constant $N_{c,\delta,\lambda}\in\mathbb{N}_+$ depending on $c, \delta, \lambda$ such that for all $n\geq N_{c,\delta,\lambda}$, with probability at leat $1-n^{-c}$,
\[
\|\widetilde\bX\bW-\rho_n\halfpower\bX_0\|_{2\to\infty} = \max_{j\in[n]} \|\bW\transpose\widetilde\bx_j-\rho_n\halfpower\bx_{0j}\|_2 \leq \rho_n\halfpower\left(1-\frac{\delta}{2}-\sqrt{1-\delta}\right).
\]
This is because $(1-\delta/2)^2 = 1-\delta+\delta^2/4 > 1-\delta$. Then
\begin{align*}
\max_{j\in[n]}\|\widetilde\bx_j\|_2 &\leq \max_{j\in[n]}\|\bW\transpose\widetilde\bx_j-\rho_n\halfpower\bx_{0j}\|_2 + \max_{j\in[n]}\|\rho_n\halfpower\bx_{0j}\|_2 \\
&\leq \rho_n\halfpower\left(1-\frac{\delta}{2}-\sqrt{1-\delta}\right) + \rho_n\halfpower\sqrt{1-\delta}.
\end{align*}
For the second result, over the same event as above, we have
\begin{align*}
\max_{i,j\in[n]}|\widetilde p_{ij} - \rho_n\bx_{0i}\transpose\bx_{0j}| &\leq
\max_{i,j\in[n]}|\widetilde\bx_i\transpose\bW(\bW\transpose\widetilde\bx_j-\rho_n\halfpower\bx_{0j})| + \max_{i,j\in[n]}|(\bW\transpose\widetilde\bx_{i}-\rho_n\halfpower\bx_{0i})\transpose\rho_n\halfpower\bx_{0j}|\\
&\leq (\max_{j\in[n]}\|\widetilde\bx_j\|_2 + \rho_n\halfpower) \|\widetilde\bX\bW-\rho_n\halfpower\bX_0\|_{2\to\infty}
\lesssim_{c,\lambda} \rho_n\halfpower\sqrt{\frac{\log n}{n}}.
\end{align*}
For the third result, over the same event as above, we have
\[
\max_{i,j\in[n]}\widetilde p_{ij} \leq \max_{i,j\in[n]}|\widetilde p_{ij} - \rho_n\bx_{0i}\transpose\bx_{0j}| + \max_{i,j\in[n]}\rho_n\bx_{0i}\transpose\bx_{0j} \leq C_{c,\lambda}\rho_n\halfpower\sqrt{\frac{\log n}{n}} + \rho_n(1-\delta).
\]
Since $\frac{\log n}{n\rho_n}\to0$ and $\max_{i,j\in[n]}\bx_{0i}\transpose\bx_{0j}\leq1-\delta$, we can pick a  (possibly larger) constant $N_{c,\delta,\lambda}$ such that $C_{c,\lambda}\halfpower\sqrt{\frac{\log n}{n\rho_n}}\leq\delta/2$ for all $n\geq N_{c,\delta,\lambda}$. Then
\[
\max_{i,j\in[n]}\widetilde p_{ij} \leq \rho_n(1-\frac{\delta}{2}).
\]
Similarly,
\[
\min_{i,j\in[n]}\widetilde p_{ij} \geq \min_{i,j\in[n]}|\widetilde p_{ij} - \rho_n\bx_{0i}\transpose\bx_{0j}| - \max_{i,j\in[n]}\rho_n\bx_{0i}\transpose\bx_{0j} \geq \rho_n\delta - C_{c,\lambda}\rho_n\halfpower\sqrt{\frac{\log n}{n}} \geq \frac{\rho_n\delta}{2}.
\]

For the fourth one, over the same event as above, we have
\begin{align*}
&\max_{j\in[n]}\|\bW\transpose\widetilde\bx_j\widetilde\bx_j\transpose\bW - \rho_n\bx_{0j}\bx_{0j}\transpose\|_2\\
&\quad\leq \max_{j\in[n]}\|\bW\transpose\widetilde\bx_j(\widetilde\bx_j\bW-\rho_n\halfpower\bx_{0j}\transpose)\|_2 + \max_{j\in[n]}\|(\bW\transpose\widetilde\bx_j-\rho_n\halfpower\bx_{0j})\rho_n\halfpower\bx_{0j}\transpose\|_2\\
&\quad\leq (\max_{j\in[n]}\|\widetilde\bx_j\|_2 + \rho_n\halfpower)\|\widetilde\bX\bW-\rho_n\halfpower\bX_0\|_{2\to\infty}
\lesssim_{c,\lambda} \rho_n\halfpower\sqrt{\frac{\log n}{n}}.
\end{align*}
\end{proof}

\begin{lemma}[Concentration of Hessian matrices]
\label{lemma:concentration_of_hessian_matrices}
Suppose $\bA\sim\mathrm{RDPG}(\rho_n^{1/2}\bX_0)$ and assume the conditions of Theorem \ref{theorem:existence_and_uniqueness} hold. Denote by $\widetilde{p}_{ij} = \widetilde{\bx}_i\transpose\widetilde{\bx}_j$, $i,j\in [n]$ and let $\epsilon > 0$ be sufficiently small. Then for any $c>0$, there exists a constant $N_{c,\delta,\lambda}\in\mathbb{N}_+$ depending on $c, \delta, \lambda$ such that for all $n\geq N_{c,\delta,\lambda}$, the following hold with probability at least $1-n^{-c}$:
\begin{align*}
&\max_{i\in[n]}\sup_{\bx_i:\|\bW\transpose\bx_i-\rho_n^{\frac{1}{2}}\bx_{0i}\|_2\leq\epsilon} \Bigg\|\frac{1}{n}\sum_{j=1}^n\left\{\frac{1}{\widetilde p_{ij}}+\frac{1-A_{ij}}{(1-\bx_i\transpose\widetilde{\bx}_j)^2}\right\}\bW\transpose\widetilde\bx_j\widetilde\bx_j\transpose\bW \\
&\qquad\qquad\qquad\qquad - \frac{1}{n}\sum_{j=1}^n\frac{\bx_{0j}\bx_{0j}\transpose}{\bx_{0i}\transpose\bx_{0j}(1-\rho_n\bx_{0i}\transpose\bx_{0j})}\Bigg\|_2 \qquad\lesssim_{c,\delta,\lambda} \rho_n^{\frac{3}{2}}\epsilon_n+\sqrt{\frac{\log n}{n\rho_n}},\\
&\left\|\frac{1}{n}\sum_{j=1}^n\left\{\frac{1}{\widetilde p_{ij}}+\frac{1-A_{ij}}{(1-\widetilde p_{ij})^2}\right\}\widetilde\bx_j\widetilde\bx_j\transpose - \frac{1}{n}\sum_{j=1}^n\frac{1}{\widetilde p_{ij}(1-\widetilde p_{ij})}\widetilde\bx_j\widetilde\bx_j\transpose\right\|_2\lesssim_{c,\delta,\lambda} \rho_n^{\frac{3}{2}}\sqrt{\frac{\log n}{n}}.
\end{align*}
\end{lemma}
\begin{proof}
For simplicity of notation, denote by $p_{0ij}=\rho_n\bx_{0i}\transpose\bx_{0j}$. 
The large probability bounds below are with regard to $n\geq N_{c,\delta,\lambda}$ for some large constant $N_{c,\delta,\lambda}$ depending on $c,\delta,\lambda$. 

\noindent
$\blacksquare$
We show the first conclusion first. Write
\begin{align*}
&\left\|\frac{1}{n}\sum_{j=1}^n\left\{\frac{1}{\widetilde p_{ij}}+\frac{1-A_{ij}}{(1-\bx_i\transpose\widetilde{\bx}_j)^2}\right\}\bW\transpose\widetilde\bx_j\widetilde\bx_j\transpose\bW - \frac{1}{n}\sum_{j=1}^n\frac{\bx_{0j}\bx_{0j}\transpose}{\bx_{0i}\transpose\bx_{0j}(1-\rho_n\bx_{0i}\transpose\bx_{0j})}\right\|_2 \\
&\quad\leq\left\|\frac{1}{n}\sum_{j=1}^n\left(1-A_{ij}\right)\left\{\frac{1}{(1-\bx_i\transpose\widetilde\bx_j)^2}-\frac{1}{(1-p_{0ij})^2}\right\}\bW\transpose\widetilde\bx_j\widetilde\bx_j\transpose\bW\right\|_2\\
&\quad\quad+\left\|\frac{1}{n}\sum_{j=1}^n\frac{A_{ij}-p_{0ij}}{(1-p_{0ij})^2}\left(\bW\transpose\widetilde\bx_j\widetilde\bx_j\transpose\bW-\rho_n\bx_{0j}\bx_{0j}\transpose\right)\right\|_2
+\left\|\frac{1}{n}\sum_{j=1}^n\frac{A_{ij}-p_{0ij}}{(1-p_{0ij})^2}\rho_n\bx_{0j}\bx_{0j}\transpose\right\|_2\\
&\quad\quad+\left\|\frac{1}{n}\sum_{j=1}^n\left\{\frac{1}{\widetilde p_{ij}}-\frac{1}{p_{0ij}}\right\}\bW\transpose\widetilde\bx_j\widetilde\bx_j\transpose\bW\right\|_2
+\left\|\frac{1}{n}\sum_{j=1}^n\frac{\bW\transpose\widetilde\bx_j\widetilde\bx_j\transpose\bW-\rho_n\bx_{0j}\bx_{0j}\transpose}{p_{0ij}(1-p_{0ij})}\right\|_2.
\end{align*}
For the first term, with probability at least $1-n^{-c}$,
\begin{align*}
&\max_{i\in[n]}\sup_{\bx_i:\|\bW\transpose\bx_i-\rho_n\halfpower\bx_{0i}\|_2\leq\epsilon}\left\|\frac{1}{n}\sum_{j=1}^n\left(1-A_{ij}\right)\left\{\frac{1}{(1-\bx_i\transpose\widetilde\bx_j)^2}-\frac{1}{(1-p_{0ij})^2}\right\}\bW\transpose\widetilde\bx_j\widetilde\bx_j\transpose\bW\right\|_2\\
&\quad\leq \max_{i\in[n]}\sup_{\bx_i:\|\bW\transpose\bx_i-\rho_n\halfpower\bx_{0i}\|_2\leq\epsilon}\frac{1}{n}\sum_{j=1}^n 2\frac{|(\bx_i\transpose\widetilde\bx_j-p_{0ij})(2-\bx_i\transpose\widetilde\bx_j-p_{0ij})|}{(1-\bx_i\transpose\widetilde\bx_j)^2(1-p_{0ij})^2}\|\widetilde\bx_j\|_2^2\\
&\quad\lesssim_{c,\delta,\lambda}\max_{i\in[n]}\sup_{\bx_i:\|\bW\transpose\bx_i-\rho_n\halfpower\bx_{0i}\|_2\leq\epsilon}\rho_n\frac{1}{n}\sum_{j=1}^n\left|\bx_i\transpose\widetilde\bx_j-p_{0ij}\right|\\
&\quad\leq \max_{i\in[n]}\sup_{\bx_i:\|\bW\transpose\bx_i-\rho_n\halfpower\bx_{0i}\|_2\leq\epsilon}\rho_n\frac{1}{n}\sum_{j=1}^n\Bigg\{\left\|\bW\transpose\bx_i-\rho_n\halfpower\bx_{0i}\right\|_2\|\widetilde\bx_j\|_2 \\
&\qquad\qquad\qquad\qquad\qquad\qquad\qquad\qquad + \|\rho_n\halfpower\bx_{0i}\|_2\left\|\bW\transpose\widetilde\bx_j-\rho_n\halfpower\bx_{0j}\right\|_2\Bigg\}\\
&\quad\lesssim_{c,\delta,\lambda} \rho_n^{\frac{3}{2}}\epsilon + \rho_n^{\frac{3}{2}}\sqrt{\frac{\log n}{n}},
\end{align*}
where in the second inequality we use Lemma \ref{lemma:some_frequently_used_results}, in the third inequality triangle inequality and Cauchy--Schwarz inequality, and in the fourth inequality Lemma \ref{lemma:two_to_infinity_norm_ASE} and Lemma \ref{lemma:some_frequently_used_results}.

\noindent For the second term, with probability at least $1-n^{-c}$,
\begin{align*}
&\max_{i\in[n]}\left\|\frac{1}{n}\sum_{j=1}^n\frac{A_{ij}-p_{0ij}}{(1-p_{0ij})^2}\left(\bW\transpose\widetilde\bx_j\widetilde\bx_j\transpose\bW-\rho_n\bx_{0j}\bx_{0j}\transpose\right)\right\|_2\\
&\quad\lesssim_\delta\frac{1}{n}\left\|\bA-\rho_n\bX_0\bX_0\transpose\right\|_\infty\max_{j\in[n]}\left\|\bW\transpose\widetilde\bx_j\widetilde\bx_j\transpose\bW-\rho_n\bx_{0j}\bx_{0j}\transpose\right\|_2 \\
&\quad\leq \frac{1}{n}\left(\|\bA\|_\infty+\|\rho_n\bX_0\bX_0\transpose\|_\infty\right)\max_{j\in[n]}\left\|\bW\transpose\widetilde\bx_j\widetilde\bx_j\transpose\bW-\rho_n\bx_{0j}\bx_{0j}\transpose\right\|_2 
\lesssim_{c,\delta,\lambda}\rho_n^{\frac{3}{2}}\sqrt{\frac{\log n}{n}},
\end{align*}
by Lemma \ref{lemma:some_frequently_used_results} and the result that $\|\bA\|_\infty\lesssim_cn\rho_n$ with probability at least $1-n^{-c}$ which follows from triangle inequality and Bernstein's inequality.

\noindent For the third term, for a typical $(k,l)$th entry, by Bernstein's inequality and a union bound over $i\in [n]$, for any $t > 0$,
\begin{align*}
&\prob\left\{\max_{i\in[n]}\left|\frac{1}{n}\sum_{j=1}^n\left(A_{ij}-p_{0ij}\right)\frac{\rho_n x_{0jk}x_{0jl}}{(1-p_{0ij})^2}\right|\geq t\right\}\\
&\quad\leq 2n\exp\left\{\frac{-3n^2t^2}{6\sum_{j=1}^n\frac{\rho_n^2x_{0jk}^2x_{0jl}^2}{(1-p_{0ij})^4} p_{0ij}(1-p_{0ij})+2\max_{j\in[n]}\frac{\rho_nx_{0jk}x_{0jl}}{(1-p_{0ij})^2} nt}\right\}\\
&\quad\leq 2n\exp\left\{-K_\delta\frac{nt^2}{\rho_n^3+\rho_n t}\right\},
\end{align*}
where $K_\delta > 0$ is a constant depending on $\delta$. Taking $t=C\sqrt{(\rho_n^3\log n)/n}$ for an appropriate constant $C > 0$, we see that
\[
\max_{i\in[n]}\left|\frac{1}{n}\sum_{j=1}^n\left(A_{ij}-p_{0ij}\right)\frac{\rho_n x_{0jk}x_{0jl}}{(1-p_{0ij})^2}\right| \lesssim_{c,\delta} \sqrt{\frac{\rho_n^3\log n}{n}}
\]
with probability at least $1-n^{-c}$. Since $d$ is fixed (it implicitly depends on $\lambda$), we have
\[
\max_{i\in[n]}\left\|\frac{1}{n}\sum_{j=1}^n\left(A_{ij}-p_{0ij}\right)\frac{\rho_n \bx_{0j}\bx_{0j}\transpose}{(1-p_{0ij})^2}\right\|_2 \lesssim_{c,\delta,\lambda} \rho_n^{\frac{3}{2}}\sqrt{\frac{\log n}{n}}
\]
with probability at least $1-n^{-c}$.

\noindent For the fourth term, with probability at least $1-n^{-c}$,
\begin{align*}
\max_{i\in[n]}\left\|\frac{1}{n}\sum_{j=1}^n\left\{\frac{1}{\widetilde p_{ij}}-\frac{1}{p_{0ij}}\right\}\bW\transpose\widetilde\bx_j\widetilde\bx_j\transpose\bW\right\|_2 &\leq \max_{i,j\in[n]}\frac{|\widetilde p_{ij}-p_{0ij}|}{\widetilde p_{ij}p_{0ij}}\|\widetilde\bx_j\|_2^2
\lesssim_{c,\delta,\lambda}\sqrt{\frac{\log n}{n\rho_n}}
\end{align*}
by Lemma \ref{lemma:some_frequently_used_results}.

\noindent For the fifth term, with probability at least $1-n^{-c}$,
\begin{align*}
\max_{i\in[n]}\left\|\frac{1}{n}\sum_{j=1}^n\frac{\bW\transpose\widetilde\bx_j\widetilde\bx_j\transpose\bW-\rho_n\bx_{0j}\bx_{0j}\transpose}{p_{0ij}(1-p_{0ij})}\right\|_2 &\lesssim_{\delta} \rho_n\inverse\max_{j\in[n]}\|\bW\transpose\widetilde\bx_j\widetilde\bx_j\transpose\bW-\rho_n\bx_{0j}\bx_{0j}\transpose\|_2
\\&
\lesssim_{c,\delta,\lambda} \sqrt{\frac{\log n}{n\rho_n}}
\end{align*}
by Lemma \ref{lemma:some_frequently_used_results}.
So the first conclusion is shown by combining the above five bounds.

\noindent
$\blacksquare$ Next, we show the second conclusion. Write\begin{align*}
&\left\|\frac{1}{n}\sum_{j = 1}^n\left\{\frac{1}{\widetilde{p}_{ij}} + \frac{(1 - A_{ij})}{(1 - \widetilde{p}_{ij})^2}\right\}\widetilde{\bx}_j\widetilde{\bx}_j\transpose - \frac{1}{n}\sum_{j = 1}^n\frac{\widetilde{\bx}_j\widetilde{\bx}_j\transpose}{\widetilde{p}_{ij}(1 - \widetilde{p}_{ij})}\right\|_2\\
&\quad= \left\|\frac{1}{n}\sum_{j = 1}^n\frac{(A_{ij} - \widetilde{p}_{ij})\bW\transpose\widetilde{\bx}_j\widetilde{\bx}_j\transpose\bW}{(1 - \widetilde{p}_{ij})^2}\right\|_2\\
&\quad\leq \left\|\frac{1}{n}\sum_{j = 1}^n\frac{(A_{ij} - \rho_n\bx_{0i}\transpose\bx_{0j})\rho_n\bx_{0j}\bx_{0j}\transpose}{(1 - \rho_n\bx_{0i}\transpose\bx_{0j})^2}\right\|_2\\
&\quad\quad+\left\|\frac{1}{n}\sum_{j = 1}^n(A_{ij} - \rho_n\bx_{0i}\transpose\bx_{0j})
\left\{
\frac{\rho_n\bx_{0j}\bx_{0j}\transpose}{(1 - \rho_n\bx_{0i}\transpose\bx_{0j})^2} - \frac{\bW\transpose\widetilde{\bx}_j\widetilde{\bx}_j\transpose\bW}{(1 - \widetilde{p}_{ij})^2}
\right\}
\right\|_2\\
&\quad\quad + \left\|\frac{1}{n}\sum_{j = 1}^n(\widetilde{p}_{ij} - \rho_n\bx_{0i}\transpose\bx_{0j})\frac{\bW\transpose\widetilde{\bx}_j\widetilde{\bx}_j\transpose\bW}{(1 - \widetilde{p}_{ij})^2}\right\|_2.
\end{align*}
The first term is $O(\rho_n^{\frac{3}{2}}\sqrt{\frac{\log n}{n}})$ with probability at least $1-n^{-c}$ as previously shown.

\noindent For the second term, with probability at least $1-n^{-c}$,
\begin{align*}
&\left\|\frac{1}{n}\sum_{j = 1}^n(A_{ij} - p_{0ij})\left\{\frac{\rho_n\bx_{0j}\bx_{0j}\transpose}{(1 - p_{0ij})^2} - \frac{\bW\transpose\widetilde{\bx}_j\widetilde{\bx}_j\transpose\bW}{(1 - \widetilde{p}_{ij})^2}\right\}\right\|_2\\
&\quad\leq \frac{1}{n}\sum_{j = 1}^n|A_{ij}-p_{0ij}|\left\{\left\|\frac{\rho_n\bx_{0j}\bx_{0j}\transpose}{(1 - p_{0ij})^2} - \frac{\bW\transpose\widetilde{\bx}_j\widetilde{\bx}_j\transpose\bW}{(1 - {p}_{0ij})^2}\right\|_2 + \left\|\frac{\bW\transpose\widetilde{\bx}_j\widetilde{\bx}_j\transpose\bW}{(1 - p_{0ij})^2} - \frac{\bW\transpose\widetilde{\bx}_j\widetilde{\bx}_j\transpose\bW}{(1 - \widetilde{p}_{ij})^2}\right\|_2\right\}\\
&\quad\leq \frac{1}{n}\|\bA-\rho_n\bX_0\bX_0\transpose\|_\infty \\
&\qquad\qquad\qquad \times \max_{i,j\in[n]}\left\{\frac{\|\rho_n\bx_{0j}\bx_{0j}\transpose-\bW\transpose\widetilde\bx_j\widetilde\bx_j\transpose\bW\|_2}{(1-p_{0ij})^2} + \left|\frac{(\widetilde p_{ij}-p_{0ij})(2-\widetilde p_{ij}-p_{0ij})}{(1-p_{0ij})^2(1-\widetilde p_{ij})^2}\right|\|\widetilde\bx_j\|_2^2\right\}\\
&\quad\lesssim_{c,\delta,\lambda} \frac{1}{n}\rho_n\left(\rho_n\halfpower\sqrt{\frac{\log n}{n}} + \rho_n\halfpower\sqrt{\frac{\log n}{n}}\rho_n\right)\\
&\quad\lesssim_{c,\delta,\lambda} \rho_n^{\frac{3}{2}}\sqrt{\frac{\log n}{n}}
\end{align*}
by Cauchy--Schwarz inequality, Lemma \ref{lemma:some_frequently_used_results}, and $\|\bA-\rho_n\bX_0\bX_0\transpose\|_\infty\leq\|\bA\|_\infty+\|\rho_n\bX_0\bX_0\transpose\|_\infty\lesssim_c n\rho_n$ with probability at least $1-n^{-c}$.

\noindent For the third term, with probability at least $1-n^{-c}$,
\begin{align*}
\left\|\frac{1}{n}\sum_{j = 1}^n(\widetilde{p}_{ij} - \rho_n\bx_{0i}\transpose\bx_{0j})\frac{\bW\transpose\widetilde{\bx}_j\widetilde{\bx}_j\transpose\bW}{(1 - \widetilde{p}_{ij})^2}\right\|_2 &\leq \max_{i,j\in[n]}|\widetilde p_{ij} - p_{0ij}|\cdot\max_{i,j\in[n]}\frac{1}{(1-\widetilde p_{ij})^2}\cdot\max_{j\in[n]}\|\widetilde\bx_j\|_2^2\\
&\lesssim_{c,\delta,\lambda}\rho_n^{\frac{3}{2}}\sqrt{\frac{\log n}{n}}
\end{align*}
by Lemma \ref{lemma:some_frequently_used_results}.
So the second conclusion is shown by combining the above three bounds.
\end{proof}

\begin{lemma}\label{lemma:Frobenius_norm_concentration}
Suppose $\bA\sim\mathrm{RDPG}(\rho_n^{1/2}\bX_0)$ and assume the conditions of Theorem \ref{theorem:asymptotic_properties_of_MSLE} hold. Denote by
\[
Z = Z(\bA) = \sum_{i = 1}^n\left\|\frac{1}{n\rho_n^{1/2}}\sum_{j = 1}^n\frac{(A_{ij} - \rho_n\bx_{0i}\transpose\bx_{0j})\bG_{0in}^{-1}\bx_{0j}}{\bx_{0i}\transpose\bx_{0j}(1 - \rho_n\bx_{0i}\transpose\bx_{0j})}\right\|_2^2.
\]
Then $Z = \expect_0Z + o_{\prob_0}(1)$. 
\end{lemma}

\begin{proof}
Denote by
\[
\bgamma_{ij} = \frac{\bG_{0in}^{-1}\bx_{0j}}{n\rho_n^{1/2}\bx_{0i}\transpose\bx_{0j}(1 - \rho_n\bx_{0i}\transpose\bx_{0j})},\quad i,j\in [n].
\]
Then we have
\begin{align*}
Z - \expect_0Z
& = \sum_{i = 1}^n\sum_{a = 1}^n\sum_{b = 1}^nE_{ia}E_{ib}\bgamma_{ia}\transpose\bgamma_{ib}\mathbbm{1}(a\neq b)\\
& = \sum_{i = 1}^n\sum_{a\geq i}\sum_{b\geq i}E_{ia}E_{ib}\bgamma_{ia}\transpose\bgamma_{ib}\mathbbm{1}(a\neq b)
 + \sum_{i = 1}^n\sum_{a < i}\sum_{b\geq i}E_{ai}E_{ib}\bgamma_{ia}\transpose\bgamma_{ib}\mathbbm{1}(a\neq b)\\
&\quad + \sum_{i = 1}^n\sum_{a\geq i}\sum_{b < i}E_{ia}E_{bi}\bgamma_{ia}\transpose\bgamma_{ib}\mathbbm{1}(a\neq b)
 + \sum_{i = 1}^n\sum_{a < i}\sum_{b < i}E_{ai}E_{bi}\bgamma_{ia}\transpose\bgamma_{ib}\mathbbm{1}(a\neq b). 
\end{align*}
To prove the desired result, we need the following decoupling inequality for $U$-statistic-type random variables. 
\begin{lemma}[Theorem 1 in \cite{10.1214/aop/1176988291}]
\label{lemma:decoupling_inequality}
Let $(X_i)_i$ be a sequence of independent random variables on a measurable space and let $(X_i^{(1)})_i,(X_i^{(2)})_i$ be two independent copies of $(X_i)_i$. Let $f_{i_1i_2}$ be families of functions of $k$ variables taking $(S\times S)$ into a Banach space $(B,\|\cdot\|_2)$. Then, for all $n\geq 2$, $t > 0$, there exists a numerical constant $C$ such that
\begin{align*}
\prob\bigg\{\bigg\|\sum_{1\leq i_1\neq i_2\leq n}f_{i_1i_2}(X_{i_1}^{(1)}, X_{i_2}^{(1)})\bigg\|\geq t\bigg\}
\leq C\prob\bigg\{C\bigg\|\sum_{1\leq i_1\neq i_2\leq n}f_{i_1i_2}(X_{i_1}^{(1)}, X_{i_2}^{(2)})\bigg\|_2\geq t\bigg\}
\end{align*}
\end{lemma}
Now we apply Lemma \ref{lemma:decoupling_inequality} with $(X_i)_i = (E_{ij}:1\leq i\leq j\leq n)$, and $f_{(i_1,a),(i_2,b)}(E_{i_1a},E_{i_2b}) = \mathbbm{1}(i_1 = i_2)\bgamma_{i_1a}\transpose\bgamma_{i_2b}$, $f_{(i_1,a),(i_2,b)}(E_{i_1a},E_{i_2b}) = \mathbbm{1}(a = i_2)\bgamma_{ai_1}\transpose\bgamma_{i_2b}$, $f_{(i_1,a),(i_2,b)}(E_{i_1a},E_{i_2b}) = \mathbbm{1}(i_1 = b)\bgamma_{i_1a}\transpose\bgamma_{bi_2}$, and $f_{(i_1,a),(i_2,b)}(E_{i_1a},E_{i_2b}) = \mathbbm{1}(a = b)\bgamma_{ai_1}\transpose\bgamma_{bi_2}$, for the four terms, respectively. Without loss of generality, it is sufficient to work with the first term. Now let $(\bar{E}_{ij}:1\leq i\leq j\leq n)$ be an independent copy of $(E_{ij}:1\leq i\leq j\leq n)$. It is sufficient to show that
\begin{align*}
\sum_{i = 1}^n\sum_{a\geq i}\sum_{b\geq i}E_{ia}\bar{E}_{ib}\bgamma_{ia}\transpose\bgamma_{ib}\mathbbm{1}(a\neq b) = o_{\prob_0}(1).
\end{align*}
By Bernstein's inequality and the independence between $E_{ia}$ and $\bar{E}_{ib}$, for any $c > 0$, there exists a constant $K_c > 0$, such that
\begin{align*}
\bigg|\sum_{i = 1}^n\sum_{a\geq i}\sum_{b\geq i}E_{ia}\bar{E}_{ib}\bgamma_{ia}\transpose\bgamma_{ib}\mathbbm{1}(a\neq b)\bigg|
\leq K_cn(\rho_n\log n)^{1/2}\max_{1\leq i,a\leq n}\bigg|\sum_{b = 1}^n\bar{E}_{ib}\bgamma_{ia}\transpose\bgamma_{ib}\bigg|\leq \frac{K_c^2\log n}{\sqrt{n}}
\end{align*}
with probability at least $1 - O(n^{-c})$. The proof is thus completed.

\end{proof}

\begin{theorem}[Theorem 4.7 in \citealp{xie-2022-entrywise}]
\label{theorem:asymptotic_normality_OSE}
Suppose $\bA\sim\mathrm{RDPG}(\rho_n\halfpower\bX_0)$ and assume the conditions of Theorem \ref{theorem:existence_and_uniqueness} hold. Define the one-step estimator $\widehat\bx_i^{(\mathrm{OS})}$ by
\[
\widehat{\bx}_i^{(\mathrm{OS})}
 = \widetilde{\bx}_i + 
 \left\{\frac{1}{n}\sum_{j = 1}^n\frac{\widetilde\bx_j\widetilde\bx_j\transpose}{\widetilde{p}_{ij}(1 - \widetilde{p}_{ij})}\right\}^{-1}
 \left\{\frac{1}{n}\sum_{j = 1}^n\frac{(A_{ij} - \widetilde{p}_{ij})\widetilde\bx_j}{\widetilde{p}_{ij}(1 - \widetilde{p}_{ij})}\right\}.
\]
Then
\[
\bG_{0in}^{1/2}(\bW\transpose\widehat{\bx}_i^{(\mathrm{OS})} - \rho_n^{1/2}\bx_{0i})
 = \frac{1}{n\rho_n^{1/2}}\sum_{j = 1}^n\frac{(A_{ij} - \rho_n\bx_{0i}\transpose\bx_{0j})\bG_{0in}^{-1/2}\bx_{0j}}{\bx_{0i}\transpose\bx_{0j}(1 - \rho_n\bx_{0i}\transpose\bx_{0j})} + \br_{in}^{(\mathrm{OS})},
\]
where
\[
\bG_{0in} = \frac{1}{n}\sum_{j = 1}^n\frac{\bx_{0j}\bx_{0j}\transpose}{\bx_{0i}\transpose\bx_{0j}(1 - \rho_n\bx_{0i}\transpose\bx_{0j})},
\]
and for all $c > 0$, there exists a constant $N_{c,\delta,\lambda}\in\mathbb{N}_+$ depending on $c,\delta,\lambda$, such that for all $n\geq N_{c,\delta,\lambda}$, with probability at least $1 - (n\rho_n)^{-c}$, $\|\br_{in}^{(\mathrm{OS})}\|_2\lesssim (\log(n\rho_n))^2/(n\rho_n^{1/2})$. Furthermore, 
\begin{align*}
&\sqrt{n}\bG_{0in}^{1/2}(\bW\transpose\widehat{\bx}_i^{(\mathrm{OS})} - \rho_n^{1/2}\bx_{0i})\overset{\calL}{\to}\mathrm{N}_d(\mathbf{0}_d,\bI_d),
\end{align*}
\end{theorem}

\section{Proofs of the Main Results}
\subsection{Proof of Theorem \ref{theorem:existence_and_uniqueness}}
\begin{proof}
$\blacksquare$ We first prove existence. For any $c>0$, there exists $N_{c,\delta,\lambda}\in\mathbb{N}_+$ such that
\[
\sup_{\|\bx_i\|_2\leq1}\max_{j\in[n]}|\bx_i\transpose\widetilde\bx_j| \leq \max_{j\in[n]}\|\widetilde\bx_j\|_2 \leq \rho_n(1-\frac{\delta}{2}) < 1
\]
with probability at least $1-n^{-c}$, where the first inequality follows from Cauchy--Schwarz inequality, the second from Lemma \ref{lemma:some_frequently_used_results}. By definition of $\widetilde M_{in}(\bx_i)$, it is continuous over the closed unit ball $\{\bx_i\in\mathbb{R}^d:\|\bx_i\|_2\leq1\}$ over this event. Hence the maximizer $\widehat\bx_i$ of $\widetilde M_{in}(\bx_i)$ exists with probability at least $1-n^{-c}$.

\noindent
$\blacksquare$ Next we prove uniqueness. By definition, with probability at least $1-n^{-c}$, $\widetilde M_{in}(\bx_i)$ is twice continuously differentiable, with
\begin{align*}
-\frac{\partial\widetilde M_{in}}{\partial\bx_i\partial\bx_i\transpose}(\bx_i)
&= \frac{1}{n}\sum_{j=1}^n\left\{\frac{1}{\widetilde p_{ij}} + \frac{1-A_{ij}}{(1-\bx_i\transpose\widetilde\bx_j)^2}\right\}\widetilde\bx_j\widetilde\bx_j\transpose
\succeq \frac{1}{n}\sum_{j=1}^n\frac{\widetilde\bx_j\widetilde\bx_j\transpose}{\widetilde p_{ij}}\\
&\succeq \frac{1}{n\rho_n}\sum_{j=1}^n\widetilde\bX\transpose\widetilde\bX
\succeq \frac{1}{n\rho_n}\sigma_d(\widetilde\bX)^2\mathbf{I}_d.
\end{align*}
By Theorem 5.2 in \cite{lei-rinaldo-2015} and Weyl's inequality, there exists a constant depending on $c$, such that with probability at least $1-n^{-c}$,
\[
\sigma_d(\widetilde\bX)^2 = \lambda_d(\bA) \geq \frac{1}{2}\lambda_d\left(\rho_n\bX_0\bX_0\transpose\right) = \frac{1}{2}n\rho_n\lambda_d\left(\frac{1}{n}\bX_0\bX_0\transpose\right) \geq \frac{1}{2}n\rho_n\lambda > 0.
\]
Therefore, for any $c>0$, there exists $N_{c,\delta,\lambda}\in\mathbb{N}_+$ such that for all $n\geq N_{c,\delta,\lambda}$, with probability at least $1-n^{-c}$, $\widetilde M_{in}(\bx_i)$ is strictly concave. Hence it has a unique maximizer $\widehat\bx_i$.
\end{proof}

\subsection{Proof of Theorem \ref{theorem:asymptotic_properties_of_MSLE}}
\begin{proof}
$\blacksquare$ We first establish the following consistency result: 
For any $c>0$, there exists some constant $N_{c,\delta,\lambda}\in\mathbb{N}_+$ depending on $c,\delta,\lambda$ such that for all $n\geq N_{c,\delta,\lambda}\in\mathbb{N}_+$, there exists an orthogonal matrix $\bW\in\mathbb{O}(d)$, such that
with probability at least $1-n^{-c}$,
\[
\max_{i\in[n]}\|\bW\transpose\widehat\bx_i - \rho_n\halfpower\bx_{0i}\|_2\lesssim_{c,\delta,\lambda} \{{\log n}/(n\rho_n)\}\halfpower.\]
Define $\widetilde{M}_{in}(\bx_i) = (1/n)\widetilde{\ell}_{in}(\bx_i)$ and the population counterpart of $\widetilde M_{in}(\bx_i)$ as
\[
M_{in}(\bx_i) = \frac{1}{n}\sum_{j=1}^n\left\{2\rho_n\halfpower\bx_i\transpose\bx_{0j} - \frac{\bx_i\transpose\bx_{0j}\bx_{0j}\transpose\bx_i}{2\bx_{0i}\transpose\bx_{0j}} + (1-\rho_n\bx_{0i}\transpose\bx_{0j})\log(1-\rho_n\halfpower\bx_i\transpose\bx_{0j})\right\}.
\]
Simple calculation shows that
\begin{align*}
\frac{\partial M_{in}}{\partial\bx_i}(\bx_i)&= \frac{1}{n}\sum_{j=1}^n\rho_n\halfpower\bx_{0j}\transpose\left(\rho_n\halfpower\bx_{0i}-\bx_i\right)\left\{\frac{1}{\rho_n\bx_{0i}\transpose\bx_{0j}} + \frac{1}{1-\rho_n\halfpower\bx_i\transpose\bx_{0j}}\right\}\rho_n\halfpower\bx_{0j},\\
\frac{\partial^2 M_{in}}{\partial\bx_i\partial\bx_i\transpose}(\bx_i)&= -\frac{1}{n}\sum_{j=1}^n\left\{\frac{1}{\rho_n\bx_{0i}\transpose\bx_{0j}} + \frac{1-\rho_n\bx_{0i}\transpose\bx_{0j}}{(1-\rho_n\halfpower\bx_i\transpose\bx_{0j})^2}\right\}\rho_n\bx_{0j}\bx_{0j}\transpose,
\end{align*}
and
\begin{align*}
\frac{\partial\widetilde M_{in}}{\partial\bx_i}(\bx_i)&= \frac{1}{n}\sum_{j=1}^n\left(A_{ij}-\bx_i\transpose\widetilde\bx_j\right)\left\{\frac{1}{\widetilde p_{ij}} + \frac{1}{1-\bx_i\widetilde\bx_j}\right\}\widetilde\bx_j,\\
\frac{\partial^2\widetilde M_{in}}{\partial\bx_i\partial\bx_i\transpose}(\bx_i)&= -\frac{1}{n}\sum_{j=1}^n\left\{\frac{1}{\widetilde p_{ij}} + \frac{1-A_{ij}}{(1-\bx_i\widetilde\bx_j)^2}\right\}\widetilde\bx_j\widetilde\bx_j\transpose.
\end{align*}
For simplicity of notation, in what follows the large probability bounds are with regard to $n\geq N_{c,\delta,\lambda}$ for some large constant $N_{c,\delta,\lambda}$ depending on $c,\delta,\lambda$. 

\noindent
Claim I (identifiability): For all $\epsilon>0$,
\[
\inf_{\|\bx_i-\rho_n\halfpower\bx_{0i}\|_2\geq\epsilon}\left\|\frac{\partial M_{in}}{\partial\bx_i}(\bx_i)\right\|_2 \geq \lambda\epsilon > \left\|\frac{\partial M_{in}}{\partial\bx_i}(\rho_n\halfpower\bx_{0i})\right\|_2 = 0.
\]
Claim II (uniform convergence): With probability at least $1-n^{-c}$,
\[
\max_{i\in[n]}\sup_{\|\bx_i\|_2\leq1} \left\|\bW\transpose\frac{\partial\widetilde M_{in}}{\partial\bx_i}(\bW\bx_i) -\frac{\partial M_{in}}{\partial\bx_i}(\bx_i)\right\|_2 \lesssim_{c,\delta,\lambda} \sqrt{\frac{\log n}{n\rho_n}}.
\]
Now we show Claim I. It is obvious that $\frac{\partial M_{in}}{\partial\bx_i}(\rho_n\halfpower\bx_{0i})=\mathbf{0}_d$. Because $\rho_n\leq1$, $\|\bx_i\|_2\leq1$, and $\max_{j\in[n]}\|\bx_{0j}\|_2\leq1$, we have
\[
-\frac{\partial^2 M_{in}}{\partial\bx_i\partial\bx_i\transpose}(\bx_i) \succeq \frac{1}{n}\sum_{j=1}^n\frac{\bx_{0j}\bx_{0j}\transpose}{\bx_{0i}\transpose\bx_{0j}} \succeq \frac{1}{n}\sum_{j=1}^n\bx_{0j}\bx_{0j}\transpose \succeq \frac{1}{n}\bX_0\transpose\bX \succeq \lambda_d\left(\frac{1}{n}\bX_0\transpose\bX\right)\mathbf{I}_d \succeq \lambda\mathbf{I}_d,
\]
which implies that $M_{in}(\bx_i)$ is strictly concave with $\rho_n\halfpower\bx_{0i}$ as a unique maximizer. By Taylor's theorem, $\frac{\partial M_{in}}{\partial\bx_i}(\bx_i) = \frac{\partial^2 M_{in}}{\partial\bx_i\partial\bx_i\transpose}(\bar\bx_i)(\bx_i-\rho_n\halfpower\bx_{0i})$, where $\bar\bx_i=\theta\rho_n\halfpower\bx_{0i}+(1-\theta)\bx_i$ for some $\theta\in[0,1]$. It follows that
\begin{align*}
\left\|\frac{\partial M_{in}}{\partial\bx_i}(\bx_i)\right\|_2 &= \left\|\frac{\partial^2 M_{in}}{\partial\bx_i\partial\bx_i\transpose}(\bar\bx_i)\left(\bx_i-\rho_n\halfpower\bx_{0i}\right)\right\|_2 \\
&\geq \lambda_d\left(-\frac{\partial^2 M_{in}}{\partial\bx_i\partial\bx_i\transpose}(\bar\bx_i)\right) \left\|\bx_i-\rho_n\halfpower\bx_{0i}\right\|_2 
\geq \lambda\left\|\bx_i-\rho_n\halfpower\bx_{0i}\right\|_2,
\end{align*}
so $\inf_{\|\bx_i-\rho_n\halfpower\bx_{0i}\|_2\geq\epsilon}\left\|\frac{\partial M_{in}}{\partial\bx_i}(\bx_i)\right\|_2 \geq \lambda\epsilon$. Thus Claim I is shown.
Now we show Claim II. By triangle inequality,
\begin{align*}
&\left\|\bW\transpose\frac{\partial\widetilde M_{in}}{\partial\bx_i}(\bW\transpose\bx_i) -\frac{\partial M_{in}}{\partial\bx_i}(\bx_i)\right\|_2\\
&\quad \leq\left\|\frac{1}{n}\sum_{j=1}^n\left(A_{ij}-\rho_n\bx_{0i}\transpose\bx_{0j}\right)\left(\frac{1}{\rho_n\bx_{0i}\transpose\bx_{0j}} + \frac{1}{1-\rho_n\halfpower\bx_i\transpose\bx_{0j}}\right)\rho_n\halfpower\bx_{0j}\right\|_2 \\
&\quad\quad+ \left\|\frac{1}{n}\sum_{j=1}^n\left(\bx_i\transpose\bW\transpose\widetilde\bx_j-\rho_n\halfpower\bx_i\transpose\bx_{0j}\right)\left(\frac{1}{\rho_n\bx_{0i}\transpose\bx_{0j}} + \frac{1}{1-\rho_n\halfpower\bx_i\transpose\bx_{0j}}\right)\rho_n\halfpower\bx_{0j}\right\|_2\\
&\quad\quad+ \left\|\frac{1}{n}\sum_{j=1}^n\left(A_{ij}-\bx_i\transpose\bW\transpose\widetilde\bx_j\right)\left\{\left(\frac{1}{\widetilde p_{ij}}+\frac{1}{1-\bx_i\transpose\bW\transpose\widetilde\bx_j}\right) - \left(\frac{1}{\rho_n\bx_{0i}\transpose\bx_{0j}}+\frac{1}{1-\rho_n^{\frac{1}{2}}\bx_i\transpose\bx_{0j}}\right)\right\}\widetilde\bx_j\right\|_2\\
&\quad\quad+ \left\|\frac{1}{n}\sum_{j=1}^n\left(A_{ij}-\bx_i\transpose\bW\transpose\widetilde\bx_j\right)\left(\frac{1}{\rho_n\bx_{0i}\transpose\bx_{0j}}+\frac{1}{1-\rho_n^{\frac{1}{2}}\bx_i\transpose\bx_{0j}}\right)\left(\bW\transpose\widetilde\bx_j-\rho_n^{\frac{1}{2}}\bx_{0j}\right)\right\|_2.
\end{align*}
For the second term,
\begin{align*}
&\max_{i\in[n]}\sup_{\|\bx_i\|_2\leq1}\left\|\frac{1}{n}\sum_{j=1}^n\left(\bx_i\transpose\bW\transpose\widetilde\bx_j-\rho_n\halfpower\bx_i\transpose\bx_{0j}\right)\left(\frac{1}{\rho_n\bx_{0i}\transpose\bx_{0j}} + \frac{1}{1-\rho_n\halfpower\bx_i\transpose\bx_{0j}}\right)\rho_n\halfpower\bx_{0j}\right\|_2\\
&\quad\leq\max_{i\in[n]}\sup_{\|\bx_i\|_2\leq1}\frac{1}{n}\sum_{j=1}^n\left\|\bW\transpose\widetilde\bx_j-\rho_n^{\frac{1}{2}}\bx_{0j}\right\|_2\|\bx_i\|_2\left(\left|\frac{1}{\rho_n\bx_{0i}\transpose\bx_{0j}}\right| + \left|\frac{1}{1-\rho_n\halfpower\bx_i\transpose\bx_{0j}}\right|\right)\rho_n\halfpower\|\bx_{0j}\|_2\\
&\quad\lesssim_\delta \left\|\widetilde\bX\bW-\rho_n\halfpower\bX_0\right\|_{2\to\infty}\rho_n\invhalfpower\\
&\quad\lesssim_{c,\delta,\lambda}\sqrt{\frac{\log n}{n\rho_n}}
\end{align*}
with probability at least $1-n^{-c}$. For the third term,
\begin{align*}
&\max_{i\in[n]}\sup_{\|\bx_i\|_2\leq1}\Bigg\|\frac{1}{n}\sum_{j=1}^n\left(A_{ij}-\bx_i\transpose\bW\transpose\widetilde\bx_j\right) \\
&\qquad\qquad\qquad \times \left\{\left(\frac{1}{\widetilde p_{ij}}+\frac{1}{1-\bx_i\transpose\bW\transpose\widetilde\bx_j}\right) - \left(\frac{1}{\rho_n\bx_{0i}\transpose\bx_{0j}}+\frac{1}{1-\rho_n^{\frac{1}{2}}\bx_i\transpose\bx_{0j}}\right)\right\}\widetilde\bx_j\Bigg\|_2\\
&\quad\leq\max_{i\in[n]}\sup_{\|\bx_i\|_2\leq1}\frac{1}{n}\sum_{j=1}^n\left(A_{ij}+\|\bx_i\|_2\|\widetilde\bx_j\|_2\right) \\
&\qquad\qquad\qquad \times \left(\frac{|\widetilde p_{ij}-\rho_n\bx_{0i}\transpose\bx_{0j}|}{|\widetilde p_{ij}\rho_n\bx_{0i}\transpose\bx_{0j}|} + \frac{\left\|\bW\transpose\widetilde\bx_j-\rho_n^{\frac{1}{2}}\bx_{0j}\right\|_2\|\bx_i\|_2}{(1-\rho_n^{\frac{1}{2}}\bx_i\transpose\bx_{0j})(1-\bx_i\transpose\bW\transpose\widetilde\bx_j)}\right)\|\widetilde\bx_j\|_2\\
&\quad\lesssim_{c,\delta,\lambda}\max_{i\in[n]}\frac{1}{n}\sum_{j=1}^n(A_{ij}+\rho_n\halfpower)\left(\rho_n^{-3/2}\sqrt{\frac{\log n}{n}} + \sqrt{\frac{\log n}{n}}\right)\rho_n\halfpower \\
&\quad\lesssim_{c,\delta,\lambda}\left(\frac{1}{n}\|\bA\|_\infty+\rho_n\halfpower\right)\rho_n\inverse\sqrt{\frac{\log n}{n}} \\
&\quad\lesssim_{c,\delta,\lambda}\sqrt{\frac{\log n}{n\rho_n}}
\end{align*}
with probability at least $1-n^{-c}$, where the second inequality follows from Lemma \ref{lemma:some_frequently_used_results}, and the last one from $\|\bA\|_\infty\lesssim_cn\rho_n$ with probability at least $1-n^{-c}$, which follows from Bernstein's inequality and triangle inequality.

\noindent For the fourth term,
\begin{align*}
&\max_{i\in[n]}\sup_{\|\bx_i\|_2\leq1}\left\|\frac{1}{n}\sum_{j=1}^n\left(A_{ij}-\bx_i\transpose\bW\transpose\widetilde\bx_j\right)\left(\frac{1}{\rho_n\bx_{0i}\transpose\bx_{0j}}+\frac{1}{1-\rho_n^{\frac{1}{2}}\bx_i\transpose\bx_{0j}}\right)\left(\bW\transpose\widetilde\bx_j-\rho_n^{\frac{1}{2}}\bx_{0j}\right)\right\|_2\\
&\quad\leq\max_{i\in[n]}\sup_{\|\bx_i\|_2\leq1}\frac{1}{n}\sum_{j=1}^n (A_{ij}+\|\bx_i\|_2\|\widetilde\bx_j\|_2) \left(\frac{1}{\rho_n\bx_{0i}\transpose\bx_{0j}}+\frac{1}{1-\rho_n^{\frac{1}{2}}\bx_i\transpose\bx_{0j}}\right) \left\|\bW\transpose\widetilde\bx_j-\rho_n^{\frac{1}{2}}\bx_{0j}\right\|_2\\
&\quad\lesssim_{c,\delta,\lambda}\max_{i\in[n]} \frac{1}{n}\sum_{j=1}^n (A_{ij}+\rho_n\halfpower)\rho_n\inverse\left\|\widetilde\bX\bW-\rho_n\halfpower\bX_0\right\|_{2\to\infty}\\
&\quad\lesssim_{c,\delta,\lambda}\left(\frac{1}{n}\|\bA\|_\infty+\rho_n\halfpower\right)\rho_n\inverse\sqrt{\frac{\log n}{n}} 
\lesssim_{c,\delta,\lambda}\sqrt{\frac{\log n}{n\rho_n}}
\end{align*}
with probability at least $1-n^{-c}$.

\noindent In order to bound the first term, a maximal inequality is required. We use the results in Chapter 8 of \cite{kosorok-2008-book}. Define a stochastic process on $\{\by\in\mathbb{R}^d:\|\by\|_2\leq1\}$ for each $k\in[d]$,
\[
J_{ink}(\by) = \frac{1}{n}\sum_{j=1}^n\left(A_{ij}-\rho_n\bx_{0i}\transpose\bx_{0j}\right)\left(\frac{1}{\rho_n\bx_{0i}\transpose\bx_{0j}}+\frac{1}{1-\rho_n\halfpower\bx_i\transpose\bx_{0j}}\right)\rho_n\halfpower x_{0jk}.
\]
Then for any $\by, \by'$ with $\|\by\|_2\leq 1,\|\by'\|_2\leq 1$,
\[
\left|J_{ink}(\by)-J_{ink}(\by')\right| = \left|\frac{1}{n}\sum_{j=1}^n\left(A_{ij}-\rho_n\bx_{0i}\transpose\bx_{0j}\right)\frac{\rho_n\halfpower\bx_{0j}\transpose(\by-\by')}{(1-\rho_n\halfpower\overline{\by}\transpose\bx_{0j})^2}\rho_n\halfpower x_{0jk}\right|,
\]
where $\overline\by=\theta\by+(1-\theta)\by'$ for some $\theta\in[0,1]$.
By Hoeffding's inequality,
\begin{align*}
\prob\{\left|J_{ink}(\by)-J_{ink}(\by')\right|\geq t\} &\leq 2\exp\left\{-\frac{2n^2t^2}{\sum_{j=1}^n(\rho_n\halfpower\bx_{0j}\transpose(\by-\by'))^2\rho_nx_{0jk}^2/(1-\rho_n\halfpower\overline{\by}\transpose\bx_{0j})^4}\right\}\\
&\leq 2\exp\left\{-\frac{nt^2}{C_\delta}\rho_n^2\|\by-\by'\|_2^2\right\},
\end{align*}
where $C_\delta > 0$ is a constant depending on $\delta$, which indicates that $J_{ink}(\by)$ is a sub-Gaussian process on $\{\by\in\mathbb{R}^d:\|\by\|_2\leq1\}$ with respect to the metric $d_n(\by,\by')=\|\by-\by'\|_2\sqrt C_\delta\rho_n^2/n$. The metric entropy of the metric space $(\{\by\in\mathbb{R}^d:\|\by\|_2\leq1\}, d_n)$ can be bounded by
\[
\log D(\epsilon, \{\by\in\mathbb{R}^d:\|\by\|_2\leq1\}, d_n) \leq d\log\left(\frac{K_\delta}{\epsilon}\sqrt{\frac{\rho_n^2}{n}}\right),
\]
where $K_\delta$ is a constant depending on $\delta$. Recall that the $\psi_2$-Orlicz norm (sub-Gaussian norm) of a random variable $X$ is defined as
\[
\left\|X\right\|_{\psi_2} = \inf\left\{c>0:\expect\psi_2\left(\frac{X}{c}\right)\leq1\right\},
\]
where $\psi_2(x)=e^{x^2}-1$ (see Chapter 8 of \citealp{kosorok-2008-book}).

\noindent By Theorem 8.4 in \cite{kosorok-2008-book},
\begin{align*}
\left\|\sup_{\|\by\|_2\leq1}J_{ink}(\by)\right\|_{\psi_2} &\lesssim \int_0^{2\sqrt{\frac{4\rho_n^2}{n\delta^4}}}\sqrt{\log D(\epsilon, \{\by\in\mathbb{R}^d:\|\by\|_2\leq1\}, d_n)}\mathrm{d}\epsilon\\
&\leq \int_0^{2\sqrt{\frac{4\rho_n^2}{n\delta^4}}}\sqrt{d\log\left(\frac{K_\delta}{\epsilon}\sqrt{\frac{\rho_n^2}{n}}\right)}\mathrm{d}\epsilon
\\&
= \int_{K_\delta}^\infty K_\delta\sqrt{d}\sqrt{\frac{\rho_n^2}{n}}\sqrt{u}e^{-u}\mathrm{d}u
\lesssim_{\delta,\lambda} \sqrt{\frac{\rho_n^2}{n}},
\end{align*}
where we note that $d$ depends on $\lambda$ implicitly. Then by Lemma 8.1 in \cite{kosorok-2008-book} and a union bound over $i\in[n]$, $\max_{i\in[n]}\sup_{\|\bx_i\|_2\leq1}|J_{ink}(\bx_i)|\lesssim_{c,\delta,\lambda}\sqrt{(\rho_n^2\log n)/n}$ with probability at least $1-n^{-c}$. So
\begin{align*}
&\max_{i\in[n]}\sup_{\|\bx_i\|_2\leq1}\left\|\frac{1}{n}\sum_{j=1}^n\left(A_{ij}-\rho_n\bx_{0i}\transpose\bx_{0j}\right)\left(\frac{1}{\rho_n\bx_{0i}\transpose\bx_{0j}}+\frac{1}{1-\rho_n\halfpower\bx_i\transpose\bx_{0j}}\right)\rho_n\halfpower\bx_{0j}\right\|_2\\
&\quad\leq \max_{i\in[n]}\sup_{\|\bx_i\|_2\leq1}\sum_{k=1}^d\left|J_{ink}(\bx_i)\right|
\lesssim_{c,\delta,\lambda} \rho_n\sqrt{\frac{\log n}{n}}
\end{align*}
with probability at least $1-n^{-c}$. Thus Claim II is shown.

By Theorem \ref{theorem:existence_and_uniqueness}, $\widehat\bx_i$ is the unique zero of $\|\partial\widetilde M_{in}/\partial\bx_i(\bx_i)\|_2$ with probability at least $1-n^{-c}$. Now
\begin{align*}
&\max_{i\in[n]}\left(\left\|\frac{\partial M_{in}}{\partial \bx_i}(\bW\transpose\widehat\bx_i)\right\|_2 - \left\|\frac{\partial M_{in}}{\partial \bx_i}(\rho_n\halfpower\bx_{0i})\right\|_2\right) \\
&\quad\leq \max_{i\in[n]}\left(\left\|\frac{\partial M_{in}}{\partial \bx_i}(\bW\transpose\widehat\bx_i)\right\|_2 - \left\|\frac{\partial\widetilde M_{in}}{\partial \bx_i}(\widehat\bx_i)\right\|_2\right)\\
&\quad\quad + \max_{i\in[n]}\left(\left\|\frac{\partial\widetilde M_{in}}{\partial \bx_i}(\bW\rho_n\halfpower\bx_{0i})\right\|_2 - \left\|\frac{\partial M_{in}}{\partial \bx_i}(\rho_n\halfpower\bx_{0i})\right\|_2\right) \\
&\quad\leq 2\max_{i\in[n]}\sup_{\|\bx_i\|_2\leq1}\left\|\bW\transpose\frac{\partial\widetilde M_{in}}{\partial\bx_i}(\bW\bx_i) -\frac{\partial M_{in}}{\partial\bx_i}(\bx_i)\right\|_2 
\lesssim_{c,\delta,\lambda}\sqrt{\frac{\log n}{n\rho_n}},
\end{align*}
where the first inequality follows from $\widehat\bx_i$ being the unique zero of $\|\partial\widetilde M_{in}/\partial\bx_i(\bx_i)\|_2$ with probability at least $1-n^{-c}$, the second inequality from triangle inequality, and the third inequality from Claim II.

\noindent By Claim I, take $\epsilon=K_{c,\delta,\lambda}\sqrt{(\log n)/(n\rho_n)}$, we have
\[
\max_{i\in[n]}\left\|\bW\transpose\widehat\bx_i-\rho_n\halfpower\bx_{0i}\right\|_2 \lesssim_{c,\delta,\lambda} \sqrt{\frac{\log n}{n\rho_n}}
\]
with probability at least $1-n^{-c}$.

\vspace*{2ex}\noindent
$\blacksquare$ We next establish the asymptotic normality. 
We utilize the asymptotic normality of the one-step estimator $\widehat\bx_i^{\mathrm{(OS)}}$ (Theorem \ref{theorem:asymptotic_normality_OSE}) to establish the asymptotic normality of the maximum surrogate likelihood estimator $\widehat\bx_i$. By the previous part of the theorem, we know that with probability at least $1-n^{-c}$, $\widehat\bx_i$ is in the interior of the closed unit ball $B(\mathbf{0}_d,1)=\{\bx\in\mathbb{R}^d:\|\bx\|_2\leq1\}$. For each $k\in[d]$, we apply Taylor's theorem to $(\partial\widetilde M_{in})/(\partial x_{ik})(\widehat\bx_i)=0$ at $\bx_i=\widetilde\bx_i$ to obtain
\begin{align*}
0 &= \frac{\partial\widetilde M_{in}}{\partial x_{ik}}(\widehat\bx_i) = \frac{\partial\widetilde M_{in}}{\partial x_{ik}}(\widetilde\bx_i) + \frac{\partial}{\partial\widetilde\bx_i\transpose}\frac{\partial\widetilde M_{in}}{\partial x_{ik}}(\widetilde\bx_i)(\widehat\bx_i-\widetilde\bx_i)\\
&\quad + \frac{1}{2}(\widehat\bx_i-\widetilde\bx_i)\transpose\frac{\partial^2}{\partial\bx_i\partial\bx_i\transpose}\frac{\partial\widetilde M_{in}}{\partial x_{ik}}(\bar\bx_i)(\widehat\bx_i-\widetilde\bx_i),
\end{align*}
where $\bar\bx_i=\theta\widehat\bx_i+(1-\theta)\widetilde\bx_i$ for some $\theta\in[0,1]$. It is easy to compute
\[
\frac{\partial^2}{\partial\bx_i\partial\bx_i\transpose}\frac{\partial\widetilde M_{in}}{\partial x_{ik}}(\bx_i) = -\frac{2}{n}\sum_{j=1}^n\frac{(1-A_{ij})\widetilde x_{jk}}{(1-\bx_i\transpose\widetilde\bx_j)^3}\widetilde\bx_j\widetilde\bx_j\transpose,
\]
then
\begin{align*}
\sup_{\|\bx_i\|_2\leq1}\left\|\frac{\partial^2}{\partial\bx_i\partial\bx_i\transpose}\frac{\partial\widetilde M_{in}}{\partial x_{ik}}(\bx_i)\right\|_2 &= \sup_{\|\bx_i\|_2\leq1}\left\|\frac{2}{n}\widetilde\bX\transpose\mathrm{diag}\left\{\frac{1-A_{i1}}{(1-\bx_i\transpose\widetilde\bx_1)^3},\ldots,\frac{1-A_{in}}{(1-\bx_i\transpose\widetilde\bx_n)^3}\right\}\widetilde\bX\right\|_2\\
&\lesssim_\delta\|\frac{1}{n}\widetilde\bX\transpose\widetilde\bX\|_2 = \|\frac{1}{n}\bA\|_2 \leq \frac{1}{n}\left(\|\bA-\bP\|_2 + \|\bP\|_2\right) \lesssim_c \rho_n,
\end{align*}
where in the last inequality we applied the fact that $\|\bA-\bP\|_2\lesssim_c\sqrt{n\rho_n}$ with probability at least $1-n^{-c}$ (Theorem 5.2 in \citealp{lei-rinaldo-2015}). By Lemma \ref{lemma:two_to_infinity_norm_ASE} and the previous part of the theorem, with probability at least $1-n^{-c}$,
\[
\|\widehat\bx_i-\widetilde\bx_i\|_2 \leq \|\bW\transpose\widehat\bx_i-\rho_n\halfpower\bx_{0i}\|_2 + \|\bW\transpose\widetilde\bx_i-\rho_n\halfpower\bx_{0i}\|_2 \lesssim_{c,\delta,\lambda} \sqrt{\frac{\log n}{n\rho_n}}.
\]
So the Taylor expansion of $(\partial\widetilde M_{in})/(\partial \bx_i)$ mentioned above can be written as
\[
-\left(\frac{\partial^2\widetilde M_{in}}{\partial\bx_i\partial\bx_i\transpose}(\widetilde\bx_i)+\bR_{in1}\right)(\widehat\bx_i-\widetilde\bx_i) = \frac{\partial\widetilde M_{in}}{\partial\bx_i}(\widetilde\bx_i),
\]
where $\bR_{in1}\in\mathbb{R}^{d\times d}$ is a random matrix with $\|\bR_{in1}\|_2\lesssim_{c,\delta,\lambda}\rho_n\halfpower\sqrt{(\log n)/n}$ with probability at least $1-n^{-c}$. By definition of $\widetilde M_{in}(\bx_i)$ and Lemma \ref{lemma:concentration_of_hessian_matrices},
\[
\left(\frac{1}{n}\sum_{j=1}^n\frac{1}{\widetilde p_{ij}(1-\widetilde p_{ij})}\widetilde\bx_j\widetilde\bx_j\transpose + \bR_{in2}\right)(\widehat\bx_i-\widetilde\bx_i) = \frac{1}{n}\sum_{j=1}^n\frac{A_{ij}-\widetilde p_{ij}}{\widetilde p_{ij}(1-\widetilde p_{ij})}\widetilde\bx_j,
\]
where $\bR_{in2}\in\mathbb{R}^{d\times d}$ is a random matrix with $\|\bR_{in2}\|_2\lesssim_{c,\delta,\lambda}\rho_n\halfpower\sqrt{(\log n)/n}$ with probability at least $1-n^{-c}$ and $\widetilde{p}_{ij} = \widetilde{\bx}_i\transpose\widetilde{\bx}_j$, $i, j\in [n]$.

\noindent Denote $\widetilde\bG_{in}=\frac{1}{n}\sum_{j=1}^n\frac{\widetilde\bx_j\widetilde\bx_j\transpose}{\widetilde p_{ij}(1-\widetilde p_{ij})}$. Similarly as in the proof of Theorem \ref{theorem:existence_and_uniqueness},
\begin{align*}
\frac{\lambda}{2} &\leq \frac{1}{n\rho_n}\lambda_d(\bA) = \lambda_d\left(\frac{1}{n\rho_n}\widetilde\bX\transpose\widetilde\bX\right) \leq \lambda_d(\widetilde\bG_{in}) \leq \lambda_1(\widetilde\bG_{in}) \\
&\lesssim_\delta \lambda_1\left(\frac{1}{n\rho_n}\widetilde\bX\transpose\widetilde\bX\right) = \frac{1}{n\rho_n}\lambda_1(\bA) \lesssim_c 1,
\end{align*}
i.e., $\widetilde\bG_{in}$ is finite and positive definite with probability at least $1-n^{-c}$.

\noindent Now write
\begin{align*}
\widehat\bx_i-\widetilde\bx_i &= \left(\widetilde\bG_{in}+\bR_{in2}\right)\inverse \frac{1}{n}\sum_{j=1}^n\frac{A_{ij}-\widetilde p_{ij}}{\widetilde p_{ij}(1-\widetilde p_{ij})}\widetilde\bx_j\\
&=\left(\bI_d+\widetilde\bG_{in}\inverse\bR_{in2}\right)\inverse\widetilde\bG_{in}\inverse\frac{1}{n}\sum_{j=1}^n\frac{A_{ij}-\widetilde p_{ij}}{\widetilde p_{ij}(1-\widetilde p_{ij})}\widetilde\bx_j\\
&=\sum_{m=0}^\infty(-\widetilde\bG_{in}\inverse\bR_{in2})^m(\widehat\bx_i^{\mathrm{(OS)}}-\widetilde\bx_i)\\
&=(\widehat\bx_i^{\mathrm{(OS)}}-\widetilde\bx_i) + \sum_{m=1}^\infty(-\widetilde\bG_{in}\inverse\bR_{in2})^m(\widehat\bx_i^{\mathrm{(OS)}}-\widetilde\bx_i),
\end{align*}
then
\begin{align*}
\|\widehat\bx_i-\widehat\bx_i^{\mathrm{(OS)}}\|_2 &\leq \sum_{m=1}^\infty\|\widetilde\bG_{in}\inverse\|_2^m\|\bR_{in2}\|_2^m\|\widehat\bx_i^{\mathrm{(OS)}}-\widetilde\bx_i\|_2\\
&= \frac{\|\widetilde\bG_{in}\inverse\|_2\|\bR_{in2}\|_2}{1-\|\widetilde\bG_{in}\inverse\|_2\|\bR_{in2}\|_2}\|\widehat\bx_i^{\mathrm{(OS)}}-\widetilde\bx_i\|_2 \\
&\lesssim_{c,\delta\lambda}\rho_n\halfpower\sqrt{\frac{\log n}{n}}\left(\|\bW\transpose\widehat\bx_i^{\mathrm{(OS)}}-\rho_n\halfpower\bx_{0i}\|_2 + \|\widetilde\bX\bW-\rho_n\halfpower\bX_0\|_{2\to\infty}\right).
\end{align*}
Note that
\begin{equation}\label{eqn:finite_pd_fisher_info}
\lambda \leq \lambda_d\left(\frac{1}{n}\bX_0\transpose\bX_0\right) \leq \lambda_d(\bG_{0in}) \leq \lambda_1(\bG_{0in}) \leq \frac{1}{\delta^2} \lambda_1\left(\frac{1}{n}\bX_0\transpose\bX_0\right) \leq \frac{1}{\delta^2},
\end{equation}
i.e., $\bG_{0in}$ is positive definite with eigenvalues bounded away from $0$ and $\infty$. By Theorem \ref{theorem:asymptotic_normality_OSE} and Bernstein's inequality, $\|\bW\transpose\widehat\bx_i^{\mathrm{(OS)}}-\rho_n\halfpower\bx_{0i}\|_2\lesssim_{c,\delta,\lambda}\sqrt{\frac{\log n}{n}}$ with probability at least $1-(n\rho_n)^{-c}$. By Lemma \ref{lemma:two_to_infinity_norm_ASE}, $\|\widetilde\bX\bW-\rho_n\halfpower\bX_0\|_{2\to\infty}\lesssim_{c,\delta,\lambda}\sqrt{\frac{\log n}{n}}$ with probability at least $1-n^{-c}$. So $\|\widehat\bx_i-\widehat\bx_i^{\mathrm{(OS)}}\|_2\lesssim_{c,\delta,\lambda}\rho_n\halfpower\frac{\log n}{n}$ with probability at least $1-(n\rho_n)^{-c}$. By Theorem \ref{theorem:asymptotic_normality_OSE} and Slutsky's theorem, we have
\[
\sqrt{n}\bG_{0in}\halfpower\left(\bW\transpose\widehat\bx_i-\rho_n\halfpower\bx_{0i}\right)\overset{\calL}{\to}N(\mathbf{0}_d,\bI_d),
\]
and
\[
\bG_{0in}\halfpower(\bW\transpose\hat\bx_i-\rho_n\halfpower\bx_{0i}) = \frac{1}{n\rho_n\halfpower}\sum_{j=1}^n\frac{(A_{ij}-\rho_n\bx_{0i}\transpose\bx_{0j})\bG_{0in}\invhalfpower\bx_{0j}}{\bx_{0i}\transpose\bx_{0j}(1-\rho_n\bx_{0i}\transpose\bx_{0j})} + \br_{in},
\]
where
\[
\|\br_{in}\|_2\lesssim_{c,\delta,\lambda}\rho_n\halfpower\frac{\log n}{n}+\frac{1}{\sqrt{n}}\sqrt{\frac{(\log(n\rho_n))^4}{n\rho_n}}
\]
with probability at least $1-(n\rho_n)^{-c}$.

\vspace*{2ex}\noindent
$\blacksquare$ We finally show the convergence of the sum of squares errors, that is
\[
\|\widehat{\bX}\bW - \rho_n^{1/2}\bX_0\|_{\mathrm{F}}^2 - \frac{1}{n}\sum_{i = 1}^n\mathrm{tr}(\bG_{0in}^{-1})\overset{\prob_o}{\to} 0.
\]
By the previous result, we have
\begin{align*}
\|\widehat{\bX}\bW - \rho_n^{1/2}\bX_{0}\|_{\mathrm{F}}^2
& = \sum_{i = 1}^n\left\|
\frac{1}{n\rho_n^{1/2}}\sum_{j = 1}^n\frac{(A_{ij} - \rho_n\bx_{0i}\transpose\bx_{0j})\bG_{0in}^{-1}\bx_{0j}}{\bx_{0i}\transpose\bx_{0j}(1 - \rho_n\bx_{0i}\transpose\bx_{0j})}
 + \bG_{0in}\invhalfpower\br_{in}\right\|_{2}^2\\
& =
\sum_{i = 1}^n\left\|
\frac{1}{n\rho_n^{1/2}}\sum_{j = 1}^n\frac{(A_{ij} - \rho_n\bx_{0i}\transpose\bx_{0j})\bG_{0in}^{-1}\bx_{0j}}{\bx_{0i}\transpose\bx_{0j}(1 - \rho_n\bx_{0i}\transpose\bx_{0j})}
\right\|_{2}^2 + 
\sum_{i = 1}^n\|\bG_{0in}^{-1/2}\br_{in}\|_2^2\\
&\quad
+ 2\sum_{i = 1}^n\left\langle
\frac{1}{n\rho_n^{1/2}}\sum_{j = 1}^n\frac{(A_{ij} - \rho_n\bx_{0i}\transpose\bx_{0j})\bG_{0in}^{-1}\bx_{0j}}{\bx_{0i}\transpose\bx_{0j}(1 - \rho_n\bx_{0i}\transpose\bx_{0j})}
,
\bG_{0in}\invhalfpower\br_{in}
\right\rangle.
\end{align*}
By Lemma \ref{lemma:Frobenius_norm_concentration}, the first term equals
\begin{align*}
&\sum_{i = 1}^n\expect_0\left\|
\frac{1}{n\rho_n^{1/2}}\sum_{j = 1}^n\frac{(A_{ij} - \rho_n\bx_{0i}\transpose\bx_{0j})\bG_{0in}^{-1}\bx_{0j}}{\bx_{0i}\transpose\bx_{0j}(1 - \rho_n\bx_{0i}\transpose\bx_{0j})}
\right\|_{2}^2 + o_{\prob_0}(1)\\
&\quad = \frac{1}{n^2\rho_n}\sum_{i = 1}^n\sum_{a = 1}^n\sum_{b = 1}^n\frac{\expect_0\{(A_{ia} - \rho_n\bx_{0i}\transpose\bx_{0a})(A_{ib} - \rho_n\bx_{0i}\transpose\bx_{0b})\}}{\bx_{0i}\transpose\bx_{0a}(1 - \rho_n\bx_{0i}\transpose\bx_{0a})\bx_{0i}\transpose\bx_{0b}(1 - \rho_n\bx_{0i}\transpose\bx_{0b})}\bx_{0a}\transpose\bG_{0in}^{-2}\bx_{0b}+ o_{\prob_0}(1)\\
&\quad = \frac{1}{n^2}\sum_{i = 1}^n\sum_{j = 1}^n\frac{\bx_{0j}\transpose\bG_{0in}^{-2}\bx_{0j}}{\bx_{0i}\transpose\bx_{0j}(1 - \rho_n\bx_{0i}\transpose\bx_{0j})}+ o_{\prob_0}(1)\\
&\quad = \frac{1}{n}\sum_{i = 1}^n\mathrm{tr}\left\{\frac{1}{n}\sum_{j = 1}^n\frac{\bx_{0j}\bx_{0j}\transpose}{\bx_{0i}\transpose\bx_{0j}(1 - \rho_n\bx_{0i}\transpose\bx_{0j})}\bG_{0in}^{-2}\right\}+ o_{\prob_0}(1)\\
&\quad = \frac{1}{n}\sum_{i = 1}^n\mathrm{tr}(\bG_{0in}^{-1})+ o_{\prob_0}(1).
\end{align*}
For the second term, by Theorem 4.7 in \cite{xie-2022-entrywise}, we have
\[
\sum_{i = 1}^n\|\br_{in}\|_2^2 \leq n\max_{i\in [n]}\|\br_{in}\|_2^2\lesssim_{c,\delta,\lambda}\frac{\rho_n(\log n)^2}{n} + \frac{(\log n)^4}{n\rho_n}
\]
with probability at least $1 - n^{-c}$ for all $n\geq N_{c,\delta,\lambda}$, so $\sum_{i = 1}^n\|\br_{in}\|_2^2 = o_{\prob_0}(1)$ by the condition that $(\log n)^4 = o(n\rho_n)$.
For the third term, by Cauchy--Schwarz inequality, we have
\begin{align*}
&\left|
\sum_{i = 1}^n\left\langle
\frac{1}{n\rho_n^{1/2}}\sum_{j = 1}^n\frac{(A_{ij} - \rho_n\bx_{0i}\transpose\bx_{0j})\bG_{0in}^{-1}\bx_{0j}}{\bx_{0i}\transpose\bx_{0j}(1 - \rho_n\bx_{0i}\transpose\bx_{0j})}
,
\bG_{0in}\invhalfpower\br_{in}
\right\rangle
\right|\\
&\quad\leq \sum_{i = 1}^n\left\|\frac{1}{n\rho_n^{1/2}}\sum_{j = 1}^n\frac{(A_{ij} - \rho_n\bx_{0i}\transpose\bx_{0j})\bG_{0in}^{-1}\bx_{0j}}{\bx_{0i}\transpose\bx_{0j}(1 - \rho_n\bx_{0i}\transpose\bx_{0j})}\right\|_2\|\bG_{0in}\invhalfpower\br_{in}\|_2\\
&\quad\leq\left\{\sum_{i = 1}^n\left\|\frac{1}{n\rho_n^{1/2}}\sum_{j = 1}^n\frac{(A_{ij} - \rho_n\bx_{0i}\transpose\bx_{0j})\bG_{0in}^{-1}\bx_{0j}}{\bx_{0i}\transpose\bx_{0j}(1 - \rho_n\bx_{0i}\transpose\bx_{0j})}\right\|_2^2\right\}^{1/2}\left\{\sum_{i = 1}^n\|\bG_{0in}\invhalfpower\br_{in}\|_2^2\right\}^{1/2}\\
&\quad = O_{\prob_0}(1)\times o_{\prob_0}(1) = o_{\prob_0}(1).
\end{align*}
Hence, we conlcude that
\[
\|\widehat{\bX}\bW - \rho_n^{1/2}\bX_{0}\|_{\mathrm{F}}^2
= \frac{1}{n}\sum_{i = 1}^n\mathrm{tr}(\bG_{0in}^{-1}) + o_{\prob_0}(1).
\]
\end{proof}
{
\subsection{Proof of Theorem \ref{thm:second_order_bias}}
\begin{proof}
For convenience, in this proof, we introduce the following $\Optilde(\cdot)$ notation. Given a sequence of random matrices $(\bX_n)_{n = 1}^\infty$ and a deterministic positive sequence $(\eps_n)_{n = 1}^\infty$, we write $\bX_n = \widetilde{O}_{\prob}(\eps_n)$, if for any $c > 0$, there exists constants $C_c, N_c > 0$, such that $\prob(\|\bX_n\|_2 > C_c\eps_n)\leq n^{-c}$ for any $n \geq N_c$. It is clear that $\Optilde(\cdot)$ is a stronger notion than $O_{\prob}(\cdot)$. 
Let $\psi(s, t):(0, 1)^2\to\mathbb{R}_+$ be a function such that $\psi(\bx_{0i}\transpose\bx_{0j}, \bx_{0i}\transpose\bx_{0j}) = 1/\mathrm{var}(A_{ij})$. Denote by $\bg_{ij}(\bx, \bu, \bv) = (A_{ij} - \bx\transpose\bv)\psi(\bx\transpose\bv, \bu\transpose\bv)$. Consider a generic estimating equation $(1/n)\sum_{j = 1}^n\bg_{ij}(\bx_i, \widetilde{\bx}_i, \widetilde{\bx}_j)$ (also known as the eigenvector-assisted estimating equation in \citealp{xie-wu-2022-ea-spn}). A simple algebra shows that the solution to this estimating equation corresponds to the one-step estimator and the maximum surrogate likelihood estimator when $\psi(s, t) = 1/t + 1/(1 - t)$ and $\psi(s, t) = 1/t + 1/(1 - s)$, respectively. Therefore, it is sufficient to work with the generic estimating equation. With a slight abuse of notation, we denote by $\widehat{\bx}_i$ the associated estimating equation estimator, and depending on the context, $\widehat{\bx}_i$ may represent the one-step estimator or the maximum surrogate likelihood estimator in this subsection. Following the proof of Theorem 1 in \cite{xie-wu-2022-ea-spn}, we have $\bW^*\widehat{\bx}_i - \bx_{0i} = \bgamma_i + \widehat{\br}_i$, where $\bgamma_i = \bG_{0in}^{-1}(1/n)\sum_{j = 1}^n{(A_{ij} - \bx_{0i}\transpose\bx_{0j})\bx_{0j}}/\{\bx_{0i}\transpose\bx_{0j}(1 - \bx_{0i}\transpose\bx_{0j})\}$, $\widehat{\br}_i = \widetilde{O}_{\prob}\{(\log n)^{2\xi}/n\}$ for any $\xi > 1$, and $\bW^*$ is a diagonal matrix whose $k$th diagonal entry is the sign of $\bu_k\transpose\widehat{\bu}_k$, where $\widehat{\bu}_k$ is the eigenvector of $\bA$ corresponding to its $k$th largest eigenvalue $\widehat{\lambda}_k$. For the ASE, we have a similar first-order stochastic expansion $\bW^*\widetilde{\bx}_i - \bx_{0i} = \be_i\transpose\bE\bX_0(\bX_0\transpose\bX_0)^{-1} + \widetilde{\br}_i$, where $\widetilde{\br}_i = \widetilde{O}_{\prob}\{(\log n)^{2\xi} / n\}$ for any $\xi > 1$. Furthermore, by Theorem 1 in \cite{xie2024higher}, the proof of Theorem 1 in \cite{athreya2022eigenvalues}, Lemma S2.3 of \cite{xie-2022-entrywise}, and the equation 
\begin{align*}
\widetilde{\bX}\bW^* - \bX_0
& = \bE\bX_0(\bX_0\transpose\bX_0)^{-1} + (\bU_\bA\bW^* - \bU_\bP - \bE\bU_\bP\bS_\bP^{-1})\bS_\bP^{1/2}\\
&\quad + \bU_\bP(\bS_\bA^{1/2} - \bS_\bP^{1/2}) + (\bU_\bA\bW^* - \bU_\bP)(\bS_\bA^{1/2} + \bS_\bP^{1/2}),
\end{align*}
we obtain the following second-order stochastic expansion of the ASE:
\begin{align*}
\bW^*\widetilde{\bx}_i - \bx_{0i} & = \widetilde{\balpha}_i + \widetilde{\bbeta}_i + \widetilde{O}_{\prob}\bigg\{\frac{(\log n)^{3\xi}}{n^{3/2}}\bigg\}
\end{align*}
for any $\xi > 1$, where $\widetilde{\balpha}_i = [\widetilde{\alpha}_{i1},\ldots,\widetilde{\alpha}_{id}]\transpose$, and $\widetilde{\bbeta}_i = [\widetilde{\beta}_{i1},\ldots,\widetilde{\beta}_{id}]\transpose$ are given by
\begin{align*}
\widetilde{\alpha}_{ik} & = \be_i\transpose\bigg(\bI_n - \frac{1}{2}\bu_k\bu_k\transpose - \sum_{m\in[d]\backslash\{k\}}\frac{\lambda_m\bu_m\bu_m\transpose}{\lambda_k - \lambda_m}\bigg)\frac{\bE\bu_k}{\lambda_k^{1/2}},\\
\widetilde{\beta}_{ik} & = \be_i\transpose\bigg(\bI_n - \bu_k\bu_k\transpose - \sum_{m\in[d]\backslash\{k\}}\frac{\lambda_m\bu_m\bu_m\transpose}{\lambda_k - \lambda_m}\bigg)\frac{\bE^2\bu_k}{\lambda_k^{3/2}}.
\end{align*}
Here, $\bS_\bA = \mathrm{diag}(\widehat{\lambda}_1,\ldots,\widehat{\lambda}_d)$, $\bS_\bP = \mathrm{diag}(\lambda_1,\ldots,\lambda_d)$, $\bE = [E_{ij}]_{n\times n} = \bA - \bX_0\bX_0\transpose$, $\bU_\bA = [\widehat{\bu}_1,\ldots,\widehat{\bu}_d]$, and $\bU_\bP = [\bu_1,\ldots,\bu_d]$. For any $a, b\in \{0, 1, 2\}$, denote by $D^{(a, b)}\psi(s, t) = \partial^{a + b}\psi(s, t)/(\partial^as\partial^bt)$ and write $D^{(a, b)}\psi_{ij} = D^{(a, b)}\psi(\bx_{0i}\transpose\bx_{0j}, \bx_{0i}\transpose\bx_{0j})$. In particular, when $a = b = 0$, we write $D^{(0, 0)}\psi_{ij} = \psi_{ij} = \psi(\bx_{0i}\transpose\bx_{0j}, \bx_{0i}\transpose\bx_{0j})$. Since $\bg_{ij}(\cdot, \cdot, \cdot)$ is continuously three-times differentiable and $\|\widetilde{\bX}\bW^* - \bX_0\|_{2\to\infty} + \|\widehat{\bX}\bW^* - \bX_0\|_{2\to\infty} = \widetilde{O}_{\prob}\{(\log n)^{3\xi}/\sqrt{n}\}$ for any $\xi > 1$, it follows Taylor's expansion that 
\begin{align}
&\bG_{0in}(\bW^*\widehat{\bx}_i - \bx_{0i})\nonumber\\
&\quad = \frac{1}{n}\sum_{j = 1}^n\bg_{ij}(\bx_{0i}, \bx_{0i}, \bx_{0j})
\\
\label{eqn:gx}
&\quad\quad
 + \bigg\{\frac{1}{n}\sum_{j = 1}^n\frac{\partial\bg_{ij}}{\partial\bx\transpose}(\bx_{0i}, \bx_{0i}, \bx_{0j}) + \bG_{0in}\bigg\}(\bW^*\widehat{\bx}_i - \bx_{0i})\\
\label{eqn:gu}
&\quad\quad + \frac{1}{n}\sum_{j = 1}^n\frac{\partial\bg_{ij}}{\partial\bu\transpose}(\bx_{0i}, \bx_{0i}, \bx_{0j})(\bW^*\widetilde{\bx}_i - \bx_{0i})\\
\label{eqn:gv}
&\quad\quad + \frac{1}{n}\sum_{j = 1}^n\frac{\partial\bg_{ij}}{\partial\bv\transpose}(\bx_{0i}, \bx_{0i}, \bx_{0j})(\bW^*\widetilde{\bx}_j - \bx_{0j})\\
\label{eqn:gxx}
&\quad\quad + \frac{1}{2n}\sum_{j = 1}^n\sum_{k,l = 1}^d\frac{\partial\bg_{ij}}{\partial x_k\partial x_l}(\bx_{0i}, \bx_{0i}, \bx_{0j})(w_k^*\widehat{x}_{ik} - x_{0ik})(w_l^*\widehat{x}_{il} - x_{0il})\\
\label{eqn:guu}
&\quad\quad + \frac{1}{2n}\sum_{j = 1}^n\sum_{k,l = 1}^d\frac{\partial\bg_{ij}}{\partial u_k\partial u_l}(\bx_{0i}, \bx_{0i}, \bx_{0j})(w_k^*\widetilde{x}_{ik} - x_{0ik})(w_l^*\widetilde{x}_{il} - x_{0il})\\
\label{eqn:gvv}
&\quad\quad + \frac{1}{2n}\sum_{j = 1}^n\sum_{k,l = 1}^d\frac{\partial\bg_{ij}}{\partial v_k\partial v_l}(\bx_{0i}, \bx_{0i}, \bx_{0j})(w_k^*\widetilde{x}_{jk} - x_{0jk})(w_l^*\widetilde{x}_{jl} - x_{0jl})\\
\label{eqn:gxu}
&\quad\quad + \frac{1}{n}\sum_{j = 1}^n\sum_{k,l = 1}^d\frac{\partial\bg_{ij}}{\partial x_k\partial u_l}(\bx_{0i}, \bx_{0i}, \bx_{0j})(w_k^*\widehat{x}_{ik} - x_{0ik})(w_l^*\widetilde{x}_{il} - x_{0il})\\
\label{eqn:gxv}
&\quad\quad + \frac{1}{n}\sum_{j = 1}^n\sum_{k,l = 1}^d\frac{\partial\bg_{ij}}{\partial x_k\partial v_l}(\bx_{0i}, \bx_{0i}, \bx_{0j})(w_k^*\widehat{x}_{ik} - x_{0ik})(w_l^*\widetilde{x}_{jl} - x_{0jl})\\
\label{eqn:guv}
&\quad\quad + \frac{1}{n}\sum_{j = 1}^n\sum_{k,l = 1}^d\frac{\partial\bg_{ij}}{\partial u_k\partial v_l}(\bx_{0i}, \bx_{0i}, \bx_{0j})(w_k^*\widetilde{x}_{ik} - x_{0ik})(w_l^*\widetilde{x}_{jl} - x_{0jl})\\
&\quad\quad + \widetilde{O}_{\prob}\bigg\{\frac{(\log n)^{3\xi}}{n^{3/2}}\bigg\}\nonumber,
\end{align}
where $\bW^* = \mathrm{diag}(w_1^*,\ldots, w_d^*)$, $\bu = [u_1,\ldots,u_d]\transpose$, $\bv = [v_1,\ldots,v_d]\transpose$, $\widehat{\bx}_i = [\widehat{x}_{i1},\ldots,\widehat{x}_{id}]\transpose$, and $\widetilde{\bx}_i = [\widetilde{x}_{i1},\ldots,\widetilde{x}_{id}]\transpose$. Now we analyze each term separately. For the terms on \eqref{eqn:gx} and \eqref{eqn:gu}, by Bernstein's inequality and the first-order expansions of $\widehat{\bx}_i$ and $\widetilde{\bx}_i$, we have
\begin{align*}
&\bigg\{\frac{1}{n}\sum_{j = 1}^n\frac{\partial\bg_{ij}}{\partial\bx\transpose}(\bx_{0i}, \bx_{0i}, \bx_{0j}) + \bG_{0in}\bigg\}(\bW^*\widehat{\bx}_i - \bx_{0i})\\
&\quad = \frac{1}{n}\sum_{j = 1}^nE_{ij}D^{(1, 0)}\psi_{ij}\bx_{0j}\bx_{0j}\transpose\bgamma_i + \widetilde{O}_{\prob}\bigg\{\frac{(\log n)^{3\xi}}{n^{3/2}}\bigg\},\\
&\frac{1}{n}\sum_{j = 1}^n\frac{\partial\bg_{ij}}{\partial\bu\transpose}(\bx_{0i}, \bx_{0i}, \bx_{0j})(\bW^*\widetilde{\bx}_i - \bx_{0i})\\
&\quad = \frac{1}{n}\sum_{j = 1}^n E_{ij}D^{(0, 1)}\psi_{ij}\bx_{0j}\bx_{0j}\transpose\be_i\transpose\bE\bX_0(\bX_0\transpose\bX_0)^{-1/2} + \widetilde{O}_{\prob}\bigg\{\frac{(\log n)^{3\xi}}{n^{3/2}}\bigg\}
\end{align*}
For the term on \eqref{eqn:gv}, by the second-order stochastic expansion for the ASE, Bernstein's inequality, Result S3 of \cite{xie-wu-2022-ea-spn}, Result B.1 of \cite{xie2024higher}, and a union bound over $j\in\{1,\ldots,n\}$, we have
\begin{align*}
&\frac{1}{n}\sum_{j = 1}^n\frac{\partial\bg_{ij}}{\partial\bv\transpose}(\bx_{0i}, \bx_{0i}, \bx_{0j})(\bW^*\widetilde{\bx}_j - \bx_{0j})\\
&\quad = \frac{1}{n}\sum_{j = 1}^n\sum_{k = 1}^d\frac{\partial\bg_{ij}}{\partial v_k}(\bx_{0i}, \bx_{0i}, \by_{0j})(w_k^*\widetilde{x}_{jk} - x_{0jk})\\
&\quad = \frac{1}{n}\sum_{j = 1}^n\sum_{k = 1}^dE_{ij}\{D^{(1, 0)}\psi_{ij}\bx_{0j}x_{0ik} + D^{(0, 1)}\psi_{ij}\bx_{0j}x_{0ik} + \psi_{ij}\be_k\}\widetilde{\alpha}_{jk}\\
&\quad\quad + \frac{1}{n}\sum_{j = 1}^n\sum_{k = 1}^dE_{ij}\{D^{(1, 0)}\psi_{ij}\bx_{0j}x_{0ik} + D^{(0, 1)}\psi_{ij}\bx_{0j}x_{0ik} + \psi_{ij}\be_k\}\\
&\quad\quad\quad\times\be_j\transpose\bigg(\bI_n - \bu_k\bu_k - \sum_{m\in[d]\backslash\{k\}}\frac{\lambda_m\bu_m\bu_m\transpose}{\lambda_k - \lambda_m}\bigg)\frac{\bE^2\bu_k}{\lambda_k^{3/2}}\\
&\quad\quad - \frac{1}{n}\sum_{j = 1}^n\sum_{k = 1}^d\psi_{ij}\bx_{0j}x_{0ik}(\widetilde{\alpha}_{jk} + \widetilde{\beta}_{jk}) + \widetilde{O}_{\prob}\bigg\{\frac{(\log n)^{3\xi}}{n^{3/2}}\bigg\}\\
&\quad = \frac{1}{n}\sum_{j = 1}^n\sum_{k = 1}^dE_{ij}\{D^{(1, 0)}\psi_{ij}\bx_{0j}x_{0ik} + D^{(0, 1)}\psi_{ij}\bx_{0j}x_{0ik} + \psi_{ij}\be_k\}\frac{\be_j\transpose\bE\bu_k}{\lambda_k^{1/2}}\\
&\quad\quad - \frac{1}{n}\sum_{j = 1}^n\sum_{k = 1}^d\psi_{ij}\bx_{0j}x_{0ik}(\widetilde{\alpha}_{jk} + \widetilde{\beta}_{jk}) + \widetilde{O}_{\prob}\bigg\{\frac{(\log n)^{3\xi}}{n^{3/2}}\bigg\}.
\end{align*}
Next, we work with the second-order derivative terms. For the terms on \eqref{eqn:gxx} and \eqref{eqn:guu}, by the first-order expansion of $\widehat{\bx}_i$ and the fact that $\|\widehat{\bX}\bW^* - \bX_0\|_{2\to\infty} = \widetilde{O}_{\prob}\{(\log n)/\sqrt{n}\}$, we have
\begin{align*}
&\frac{1}{n}\sum_{j = 1}^n\sum_{k = 1}^d\sum_{l = 1}^d
\frac{\partial^2 \bg_{ij}}{\partial x_k\partial x_l}(\bx_{0i}, \bx_{0i}, \bx_{0j})(w_l^*\widehat{x}_{il} - x_{0il})(w_k^*\widehat{x}_{ik} - x_{0ik})\\
&\quad = -\frac{2}{n}\sum_{j = 1}^n\sum_{k = 1}^d\sum_{l = 1}^dD^{(1, 0)}\psi_{ij}\bx_{0j} x_{0jk}x_{0jl}(w_l^*\widehat{x}_{il} - x_{0il})(w_k^*\widehat{x}_{ik} - x_{0ik})\\
&\quad\quad + \frac{1}{n}\sum_{j = 1}^n\sum_{k = 1}^d\sum_{l = 1}^dE_{ij}D^{(2, 0)}\psi_{ij}\bx_{0j}x_{0jk}x_{0jl}(w_l^*\widehat{x}_{il} - x_{0il})(w_k^*\widehat{x}_{ik} - x_{0ik})\\
&\quad = -\frac{2}{n}\sum_{j = 1}^nD^{(1, 0)}\psi_{ij}\bx_{0j} \bx_{0j}\transpose\bigg[\bgamma_{i} + \widetilde{O}_{\prob}\bigg\{\frac{(\log n)^{2\xi}}{n}\bigg\}\bigg]\bx_{0j}\transpose\bigg[\bgamma_{i} + \widetilde{O}_{\prob}\bigg\{\frac{(\log n)^{2\xi}}{n}\bigg\}\bigg]\\
&\quad\quad + \frac{1}{n}\sum_{j = 1}^n\sum_{k = 1}^d\sum_{l = 1}^dE_{ij}D^{(2, 0)}\psi_{ij}\bx_{0j}x_{0jk}x_{0jl}\times \widetilde{O}_{\prob}\bigg\{\frac{(\log n)^{2\xi}}{n}\bigg\}\\
&\quad = -\frac{2}{n}\sum_{j = 1}^nD^{(1, 0)}\psi_{ij}\bx_{0j} (\bx_{0j}\transpose\bgamma_i)^2 + \widetilde{O}_{\prob}\bigg\{\frac{(\log n)^{3\xi}}{n^{3/2}}\bigg\}
\end{align*}
and
\begin{align*}
&\frac{1}{n}\sum_{j = 1}^n\sum_{k = 1}^d\sum_{l = 1}^d\frac{\partial^2 \bg_{ij}}{\partial u_k\partial u_l}(\bx_{0i}, \bx_{0i}, \bx_{0j})(w_l^*\widetilde{x}_{il} - x_{0il})(w_k^*\widetilde{x}_{ik} - x_{0ik})\\
&\quad = \frac{1}{n}\sum_{j = 1}^n\sum_{k = 1}^d\sum_{l = 1}^dE_{ij}D^{(0, 2)}\psi_{ij}\bx_{0j} x_{0jk}x_{0jl}(w_l^*\widetilde{x}_{il} - x_{0il})(w_k^*\widetilde{x}_{ik} - x_{0ik})\\
&\quad = \frac{1}{n}\sum_{j = 1}^n\sum_{k = 1}^d\sum_{l = 1}^dE_{ij}D^{(0, 2)}\psi_{ij}\bx_{0j} x_{0jk}x_{0jl}\times\widetilde{O}_{\prob}\bigg\{\frac{(\log n)^{2}}{n}\bigg\}
 = \widetilde{O}_{\prob}\bigg\{\frac{(\log n)^{3\xi}}{n^{3/2}}\bigg\}.
\end{align*}
For the term on \eqref{eqn:gvv}, we first observe that, for any $(c_{ij}:i,j\in\{1,\ldots,n\})$ with $\sup_{i,j\in\{1,\ldots,n\}}|c_{ij}| = O(1)$, the following bound holds:
\begin{align*}
\frac{1}{n}\sum_{j = 1}^nc_{ij}E_{ij}\frac{\be_j\transpose\bE\bu_k}{\lambda_k^{1/2}}\frac{\be_j\transpose\bE\bu_l}{\lambda_l^{1/2}}
& = \frac{c_{ii}E_{ii}}{n}\frac{\be_j\transpose\bE\bu_k}{\lambda_k^{1/2}}\frac{\be_j\transpose\bE\bu_l}{\lambda_l^{1/2}} + \frac{1}{n}\sum_{j\neq i}^nc_{ij}E_{ij}\sum_{a\neq i}^n\frac{E_{ja}x_{ak}}{\lambda_k}\sum_{b\neq i}^n\frac{E_{jb}x_{bl}}{\lambda_l}\\
&\quad + \frac{1}{n}\sum_{j\neq i}^n\frac{c_{ij}x_{ik}E_{ij}^2}{\lambda_k}\sum_{b\neq i}\frac{E_{jb}x_{bl}}{\lambda_l} + \frac{1}{n}\sum_{j\neq i}^n\frac{c_{ij}x_{il}E_{ij}^2}{\lambda_l}\sum_{a\neq i}\frac{E_{ja}x_{ak}}{\lambda_k}\\
&\quad + \frac{1}{n}\sum_{j\neq i}^n\frac{c_{ij}x_{ik}x_{il}E_{ij}^3}{\lambda_k\lambda_l} = \widetilde{O}_{\prob}\bigg\{\frac{(\log n)^{3\xi}}{n^{3/2}}\bigg\}
\end{align*}
for any $\xi > 1$, where we have used the Bernstein's inequality, the independence between $(E_{ij}:j\in\{1,\ldots,n\}\backslash\{i\})$ and $(E_{ja}:j,a\in\{1,\ldots,n\}\backslash\{i\})$, the fact that $|E_{ij}|\leq 1$ with probability one, and a union bound over $j\in\{1,\ldots,n\}$.
Then, we write $\widetilde{\br}_i = [\widetilde{r}_{i1},\ldots,\widetilde{r}_{id}]\transpose$ and obtain the following decomposition on the term on \eqref{eqn:gvv}:
\begin{align*}
&\frac{1}{n}\sum_{j = 1}^n\sum_{k = 1}^d\sum_{l = 1}^d\frac{\partial^2 \bg_{ij}}{\partial v_k\partial v_l}(\bx_{0i}, \bx_{0i}, \bx_{0j})(w_l^*\widetilde{x}_{jl} - x_{0jl})(w_k^*\widetilde{x}_{jk} - x_{0jk})\\
&\quad = \frac{1}{n}\sum_{j = 1}^n\sum_{k = 1}^d\sum_{l = 1}^dE_{ij}
\{D^{(2, 0)}\psi_{ij}\bx_{0j} x_{0ik}x_{0il} + 2D^{(1, 1)}\psi_{ij}\bx_{0j}x_{0ik}x_{0il} + D^{(0, 2)}\psi_{ij}\bx_{0j} x_{0ik}x_{0il}\\
&\quad\quad\quad + D^{(1, 0)}\psi_{ij}(\be_kx_{0il} + \be_lx_{0ik}) + D^{(0, 1)}\psi_{ij}(\be_k x_{0il} + \be_l x_{0ik})\}
\\&\quad\quad\times
\bigg(\frac{\be_j\transpose\bE\bu_k}{\lambda_k^{1/2}} + \widetilde{r}_{jk}\bigg)\bigg(\frac{\be_j\transpose\bE\bu_l}{\lambda_l^{1/2}} + \widetilde{r}_{jl}\bigg)\\
&\quad\quad - \frac{1}{n}\sum_{j = 1}^n\sum_{k = 1}^d\sum_{l = 1}^d\{2D^{(1, 0)}\psi_{ij}\bx_{0j} x_{0ik}x_{0il} + 2D^{(0, 1)}\psi_{ij}\bx_{0j}x_{0ik}x_{0il} + \psi_{ij}(\be_l x_{0ik} + \be_k x_{0il})\}\\
&\quad\quad\times\frac{\be_j\transpose\bE\bu_k\be_j\bE\bu_l}{\lambda_k^{1/2}\lambda_l^{1/2}} + 
\Optilde\bigg\{\frac{(\log n)^{3\xi}}{n^{3/2}}\bigg\}\\
&\quad = \frac{1}{n}\sum_{j = 1}^n\sum_{k = 1}^d\sum_{l = 1}^dE_{ij}
\{D^{(2, 0)}\psi_{ij}\bx_{0j} x_{0ik}x_{0il} + 2D^{(1, 1)}\psi_{ij}\bx_{0j}x_{0ik}x_{0il} + D^{(0, 2)}\psi_{ij}\bx_{0j} x_{0ik}x_{0il}\\
&\quad\quad\quad + D^{(1, 0)}\psi_{ij}(\be_kx_{0il} + \be_lx_{0ik}) + D^{(0, 1)}\psi_{ij}(\be_k x_{0il} + \be_l x_{0ik})\}\frac{\be_j\transpose\bE\bu_k}{\lambda_k^{1/2}}\frac{\be_j\transpose\bE\bu_l}{\lambda_l^{1/2}}
\\&\quad\quad
 - \frac{1}{n}\sum_{j = 1}^n\sum_{k = 1}^d\sum_{l = 1}^d\{2D^{(1, 0)}\psi_{ij}\bx_{0j} x_{0ik}x_{0il} + 2D^{(0, 1)}\psi_{ij}\bx_{0j}x_{0ik}x_{0il} + \psi_{ij}(\be_l x_{0ik} + \be_k x_{0il})\}\\
&\quad\quad\times \frac{\be_j\transpose\bE\bu_k\be_j\bE\bu_l}{\lambda_k^{1/2}\lambda_l^{1/2}} + 
\Optilde\bigg\{\frac{(\log n)^{3\xi}}{n^{3/2}}\bigg\}\\
&\quad = - \frac{1}{n}\sum_{j = 1}^n\sum_{k = 1}^d\sum_{l = 1}^d\{2D^{(1, 0)}\psi_{ij}\bx_{0j} x_{0ik}x_{0il} + 2D^{(0, 1)}\psi_{ij}\bx_{0j}x_{0ik}x_{0il} + \psi_{ij}(\be_l x_{0ik} + \be_k x_{0il})\} \\&\quad\quad
\times\frac{\be_j\transpose\bE\bu_k\be_j\bE\bu_l}{\lambda_k^{1/2}\lambda_l^{1/2}}
 + 
\Optilde\bigg\{\frac{(\log n)^{3\xi}}{n^{3/2}}\bigg\},
\end{align*}
where we have used the first-order stochastic expansion for $\widetilde{\bx}_i$ and a union bound over $j\in\{1,\ldots,n\}$. For the term on \eqref{eqn:gxu}, by the first-order stochastic expansions of $\widehat{\bx}_i$ and $\widetilde{\bx}_i$ and Bernstein's inequality, we obtain
\begin{align*}
&\frac{1}{n}\sum_{j = 1}^n\sum_{k,l = 1}^d\frac{\partial^2 \bg_{ij}}{\partial x_k\partial u_l}(\bx_{0i}, \bx_{0i}, \bx_{0j})(w_k^*\widehat{x}_{ik} - x_{0ik})(w_l^*\widetilde{x}_{il} - x_{0il})\\
&\quad = \frac{1}{n}\sum_{j = 1}^n\sum_{k,l = 1}^dE_{ij}D^{(1, 1)}\psi_{ij}\bx_{0j} x_{0jk}x_{0jl}(w_k^*\widehat{x}_{ik} - x_{0ik})(w_l^*\widetilde{x}_{il} - x_{0il})\\
&\quad\quad - 
\frac{1}{n}\sum_{j = 1}^n\sum_{k,l = 1}^dD^{(0, 1)}\psi_{ij}\bx_{0j} x_{0jk}x_{0jl}(w_k^*\widehat{x}_{ik} - x_{0ik})(w_l^*\widetilde{x}_{il} - x_{0il})\\
&\quad = \Optilde\bigg\{\frac{(\log n)^{2\xi}}{n^{3/2}}\bigg\}
- \frac{1}{n}\sum_{j = 1}^n\sum_{k,l = 1}^dD^{(0, 1)}\psi_{ij}\bx_{0j} x_{0jk}x_{0jl}\bigg[\gamma_{ik}\frac{\be_i\transpose\bE\bu_l}{\lambda_l^{1/2}} + \Optilde\bigg\{\frac{(\log n)^{3\xi}}{n^{3/2}}\bigg\}\bigg]
\\
&\quad = -\frac{1}{n}\sum_{j = 1}^nD^{(0, 1)}\psi_{ij}\bx_{0j} \bx_{0j}\transpose\bgamma_i\bx_{0j}\transpose(\bX_0\transpose\bX_0)^{-1}\bX_0\transpose\bE\be_i + \Optilde\bigg\{\frac{(\log n)^{2\xi}}{n^{3/2}}\bigg\}
\end{align*}
For the term on \eqref{eqn:gxv}, by the first-order stochastic expansions of $\widehat{\bx}_i$, $\widetilde{\bx}_j$, Result 3 in \cite{xie-wu-2022-ea-spn}, the fact that $|E_{ij}|\leq 1$ with probabiltiy one, and a union bound over $j\in\{1,\ldots,n\}$, 
\begin{align*}
&\frac{1}{n}\sum_{j = 1}^n\sum_{k,l = 1}^d\frac{\partial^2 \bg_{ij}}{\partial x_k\partial v_l}(\bx_{0i}, \bx_{0i}, \by_{0j})(w_k^*\widehat{x}_{ik} - x_{0ik})(w_l^*\widetilde{x}_{jl} - x_{0jl})\\
&\quad = \sum_{j = 1}^n\sum_{k,l = 1}^d\{D^{(2, 0)}\psi_{ij}\bx_{0j} x_{0jk}x_{0il} + D^{(1, 1)}\psi_{ij}\bx_{0j} x_{0jk}x_{0il} + D^{(1, 0)}\psi_{ij}(\bx_{0j}\be_k\transpose\be_l + \be_l x_{0jk})\}\\
&\quad\quad\times E_{ij}\bigg(\frac{\be_j\transpose\bE\bu_l}{\lambda_l^{1/2}} + \widetilde{r}_{jl}\bigg)\Optilde\bigg\{\frac{(\log n)^{\xi}}{n^{3/2}}\bigg\}\\
&\quad\quad - \sum_{j = 1}^n\sum_{k,l = 1}^d\{2D^{(1, 0)}\psi_{ij}\bx_{0j} x_{0jk}x_{0il} + D^{(0, 1)}\psi_{ij}\bx_{0j} x_{0jk}x_{0il} + \psi_{ij}(\be_l x_{0jk} + \bx_{0j}\be_k\transpose\be_l)\}\\
&\quad\quad\times 
\bigg(\frac{\be_j\transpose\bE\bu_l}{\lambda_l^{1/2}} + \widetilde{r}_{jl}\bigg)\Optilde\bigg\{\frac{(\log n)^{\xi}}{n^{3/2}}\bigg\}
\\
&\quad = \sum_{j = 1}^n\sum_{k,l = 1}^d\{D^{(2, 0)}\psi_{ij}\bx_{0j} x_{0jk}x_{0il} + D^{(1, 1)}\psi_{ij}\bx_{0j} x_{0jk}x_{0il} + D^{(1, 0)}\psi_{ij}(\bx_{0j}\be_k\transpose\be_l + \be_l x_{0jk})\}\\
&\quad\quad\times \frac{E_{ij}\be_j\transpose\bE\bu_l}{\lambda_l^{1/2}}\Optilde\bigg\{\frac{(\log n)^{\xi}}{n^{3/2}}\bigg\}\\
&\quad\quad + \sum_{j = 1}^n\sum_{k,l = 1}^d\{D^{(2, 0)}\psi_{ij}\bx_{0j} x_{0jk}x_{0il} + D^{(1, 1)}\psi_{ij}\bx_{0j} x_{0jk}x_{0il} + D^{(1, 0)}\psi_{ij}(\bx_{0j}\be_k\transpose\be_l + \be_l x_{0jk})\}\\
&\quad\quad\times E_{ij}\widetilde{r}_{jl}\Optilde\bigg\{\frac{(\log n)^\xi}{n^{3/2}}\bigg\}\\
&\quad\quad - \sum_{j = 1}^n\sum_{k = 1}^d\sum_{l = 1}^d\{2D^{(1, 0)}\psi_{ij}\bx_{0j} x_{0jk}x_{0il} + D^{(0, 1)}\psi_{ij}\bx_j x_{0jk}x_{0il} + \psi_{ij}(\be_l x_{0jk} + \bx_{0j}\be_k\transpose\be_l)\}\\
&\quad\quad\times \frac{\be_j\transpose\bE\bu_l}{\lambda_l^{1/2}}\Optilde\bigg\{\frac{(\log n)^{\xi}}{n^{3/2}}\bigg\}\\
&\quad\quad - \sum_{j = 1}^n\sum_{k = 1}^d\sum_{l = 1}^d\{2D^{(1, 0)}\psi_{ij}\bx_{0j} x_{0jk}x_{0il} + D^{(0, 1)}\psi_{ij}\bx_{0j} x_{0jk}x_{0il} + \psi_{ij}(\be_l x_{0jk} + \bx_{0j}\be_k\transpose\be_l)\}\\
&\quad\quad\times \widetilde{r}_{jl}\Optilde\bigg\{\frac{(\log n)^\xi}{n^{3/2}}\bigg\}\\
&\quad = \Optilde\bigg\{\frac{(\log n)^{3\xi}}{n^{3/2}}\bigg\}.
\end{align*}
Finally, for the term on \eqref{eqn:guv}, by the first-order stochastic expansions of $\widetilde{\bx}_i$, $\widetilde{\bx}_j$, Result 3 of \cite{xie-wu-2022-ea-spn}, the fact that $|E_{ij}|\leq 1$ with probability one, Bernstein's inequality, and a union bound over $j\in\{1,\ldots,n\}$, we have
\begin{align*}
&\frac{1}{n}\sum_{j = 1}^n\sum_{k, l = 1}^d\frac{\partial^2 \bg_{ij}}{\partial u_k\partial v_l}(\bx_{0i}, \bx_{0i}, \bx_{0j})(w_k^*\widetilde{x}_{ik} - x_{0ik})(w_k^*\widetilde{x}_{jl} - x_{0jl})\\
&\quad = \sum_{j = 1}^n\sum_{k, l = 1}^d\{D^{(1, 1)}\psi_{ij}\bx_{0j} x_{0jk}x_{0il} + D^{(0, 2)}\psi_{ij}\bx_{0j} x_{0jk}x_{0il} + D^{(0, 1)}\psi_{ij}(\bx_{0j}\be_k\transpose\be_l + \be_l x_{0jk})\}\\
&\quad\quad\times E_{ij}\bigg(\frac{\be_j\transpose\bE\bu_l}{\lambda_l^{1/2}} + \widetilde{r}_{jl}\bigg)\Optilde\bigg\{\frac{(\log n)^\xi}{n^{3/2}}\bigg\}\\
&\quad\quad - \frac{1}{n}\sum_{j = 1}^n\sum_{k = 1}^d\sum_{l = 1}^dD^{(0, 1)}\psi_{ij}\bx_{0j} x_{0jk}x_{0il}\bigg[\frac{\be_i\transpose\bE\bu_k}{\lambda_k^{1/2}} + \Optilde\bigg\{\frac{(\log n)^{2\xi}}{n}\bigg\}\bigg]\bigg(\frac{\be_j\transpose\bE\bu_l}{\lambda_l^{1/2}} + \widetilde{r}_{jl}\bigg)\\
&\quad = \Optilde\bigg\{\frac{(\log n)^{3\xi}}{n^{3/2}}\bigg\} - \frac{1}{n}\sum_{j = 1}^n\sum_{k = 1}^d\sum_{l = 1}^dD^{(0, 1)}\psi_{ij}\bx_{0j} x_{0jk}x_{0il}\frac{\be_i\transpose\bE\bu_k}{\lambda_k^{1/2}}\frac{\be_j\transpose\bE\bu_l}{\lambda_l^{1/2}} = \Optilde\bigg\{\frac{(\log n)^{3\xi}}{n^{3/2}}\bigg\}.
\end{align*}
Combining the above results, we obtain the following second-order expansion for $\widehat{\bx}_i$:
\begin{align*}
\bW^*\widehat{\bx}_i - \bx_{0i} = \bgamma_i - \bG_{0in}^{-1}\frac{1}{n}\sum_{j = 1}^n\psi_{ij}\bx_{0j}\bx_{0i}\transpose(\bX\transpose\bX)^{-1}\widetilde{\balpha}_j + \bq_i + \Optilde\bigg\{\frac{(\log n)^{3\xi}}{n^{3/2}}\bigg\}, 
\end{align*}
where
\begin{align*}
\bq_i & = \bG_{0in}^{-1}\frac{1}{n}\sum_{j = 1}^nE_{ij}D^{(0, 1)}\psi_{ij}\bx_{0j}
\bx_{0j}\transpose(\bX_0\transpose\bX_0)^{-1}\bX_0\transpose\bE\be_i
 + \bG_{0in}^{-1}\frac{1}{n}\sum_{j = 1}^nE_{ij}D^{(1, 0)}\psi_{ij}\bx_{0j}\bx_{0j}\transpose\bgamma_i\\
&\quad + \bG_{0in}^{-1}\frac{1}{n}\sum_{j = 1}^n\sum_{k = 1}^dE_{ij}\{D^{(1, 0)}\psi_{ij}\bx_{0j}x_{0ik} + D^{(0, 1)}\psi_{ij}\bx_{0j}x_{0ik} + \psi_{ij}\be_k\}\frac{\be_j\transpose\bE\bu_k}{\lambda_k^{1/2}}\\
&\quad - \bG_{0in}^{-1}\frac{1}{n}\sum_{j = 1}^n\psi_{ij}\bx_{0j}\bx_{0i}\transpose\widetilde{\bbeta}_{j} - \bG_{in}^{-1}\frac{1}{n}\sum_{j = 1}^nD^{(1, 0)}\psi_{ij}\bx_{0j} (\bx_{j}\transpose\bgamma_i)^2\\
&\quad - \bG_{0in}^{-1}\frac{1}{2n}\sum_{j = 1}^n\sum_{k,l = 1}^d\{2D^{(1, 0)}\psi_{ij}\bx_{0j} x_{0ik}x_{0il} + 2D^{(0, 1)}\psi_{ij}\bx_{0j}x_{0ik}x_{0il} + \psi_{ij}(\be_l x_{0ik} + \be_k x_{0il})\}\\
&\quad\quad\times\frac{\be_j\transpose\bE\bu_k\be_j\bE\bu_l}{\lambda_k^{1/2}\lambda_l^{1/2}}
 -\bG_{0in}^{-1}\frac{1}{n}\sum_{j = 1}^nD^{(0, 1)}\psi_{ij}\bx_{0j} \bx_{0j}\transpose\bgamma_i\bx_{0j}\transpose(\bX_0\transpose\bX_0)^{-1}\bX_0\transpose\bE\be_i.
\end{align*}
Since $|E_{ij}|\leq 1$ with probability one, it follows that $\expect \bW^*\widehat{\bx}_i = \bx_{0i} + \expect \bq_i + o(n^{-1})$. Furthermore, a simple algebra shows that 
\begin{align*}
\expect \bq_i& = \bG_{0in}^{-1}\frac{1}{n}\sum_{j = 1}^nD^{(0, 1)}\psi_{ij}\psi_{ij}\bx_{0j}
\bx_{0j}\transpose(\bX_0\transpose\bX_0)^{-1}\bx_{0j}
 + \bG_{0in}^{-1}\frac{1}{n}\sum_{j = 1}^nD^{(1, 0)}\psi_{ij}\bx_{0j}\bx_{0j}\transpose\bG_{0in}^{-1}\bx_{0j}\\
 &\quad + \bb_i^{(\mathrm{ASE})} + \bb_i^{(\mathrm{base})}.
\end{align*}
The proof is completed by substituting the generic function $\psi(s, t)$ above with $\psi(s, t) = 1 / t + 1 / (1 - t)$ for the one-step estimator and $\psi(s, t) = 1 / t + 1 / (1 - s)$ for the maximum surrogate likelihood estimator, respectively. 
\end{proof}
}

\subsection{Proof of Theorem \ref{theorem:posterior_convergence_tvm}}
\begin{proof}
Similar to the earlier proofs, the large probability bounds below are with regard to $n\geq N_{c,\delta,\lambda}$ for some large constant $N_{c,\delta,\lambda}$ depending on $c,\delta,\lambda$. 
By definition, $\bt=\sqrt{n}\bW\transpose(\bx_i-\widehat\bx_i)$, then $\bx_i=\widehat\bx_i+\bW\bt/\sqrt{n}$. Denote the parameter space of $\bt$ by $\widehat\Theta_{in}=\{\bt\in\mathbb{R}^d:\|\widehat\bx_i+\bW\bt/\sqrt{n}\|_2\leq1\}$. Denote the normalizing constant by
\[
d_{in} = \int_{\mathbb{R}^d}\exp\left\{n\widetilde M_{in}(\widehat\bx_i+\frac{\bW\bt}{\sqrt{n}}) - n\widetilde M_{in}(\widehat\bx_i)\right\}\pi(\widehat\bx_i+\frac{\bW\bt}{\sqrt{n}})\indicator(\bt\in\widehat\Theta_{in})\mathrm{d}\bt.
\]
By definition,
\[
\widetilde\pi_{in}^*(\bt\mid\bA) = \frac{1}{d_{in}}\exp\left\{n\widetilde M_{in}(\widehat\bx_i+\frac{\bW\bt}{\sqrt{n}}) - n\widetilde M_{in}(\widehat\bx_i)\right\}\pi(\widehat\bx_i+\frac{\bW\bt}{\sqrt{n}})\indicator(\bt\in\widehat\Theta_{in}).
\]
It is sufficient to show that
\begin{equation}\label{equation:sufficient_condition_tvm}
\begin{aligned}
&\max_{i\in[n]}\int_{\mathbb{R}^d}\left(1+\|\bt\|_2^\alpha\right){\Bigg|}\exp\left\{n\widetilde M_{in}(\widehat\bx_i+\frac{\bW\bt}{\sqrt{n}}) - n\widetilde M_{in}(\widehat\bx_i)\right\}\pi\left(\widehat\bx_i+\frac{\bW\bt}{\sqrt{n}}\right)\indicator(\bt\in\widehat\Theta_{in})\\
&\quad\quad- e^{-\frac{1}{2}\bt\transpose\bG_{0in}\bt}\pi\left(\rho_n^{\frac{1}{2}}\bW\bx_{0i}\right){\Bigg|}\mathrm{d}\bt = o_{\prob_0}(1).
\end{aligned}
\end{equation}
To see this, note that \eqref{equation:posterior_convergence_tvm} in the manuscript can be rewritten as
\begin{align*}
&\max_{i\in[n]}\frac{1}{d_{in}}\int\left(1+\|\bt\|_2^\alpha\right)\Bigg|\exp\left\{n\widetilde M_{in}(\widehat\bx_i+\frac{\bW\bt}{\sqrt{n}}) - n\widetilde M_{in}(\widehat\bx_i)\right\}\pi(\widehat\bx_i+\frac{\bW\bt}{\sqrt{n}})\indicator(\bt\in\widehat\Theta_{in}) \\
&\qquad\qquad\qquad\qquad\qquad - \frac{d_{in}e^{-\bt\transpose\bG_{0in}\bt/2}}{\det(2\pi\bG_{0in}\inverse)\halfpower}\Bigg|\mathrm{d}\bt\\
&\quad\leq\max_{i\in[n]}\frac{1}{d_{in}}\int\left(1+\|\bt\|_2^\alpha\right){\Bigg|}\exp\left\{n\widetilde M_{in}(\widehat\bx_i+\frac{\bW\bt}{\sqrt{n}}) - n\widetilde M_{in}(\widehat\bx_i)\right\}\pi(\widehat\bx_i+\frac{\bW\bt}{\sqrt{n}})\indicator(\bt\in\widehat\Theta_{in})\\
&\qquad\qquad\qquad\qquad\qquad-e^{-\frac{1}{2}\bt\transpose\bG_{0in}\bt}\pi(\rho_n^{\frac{1}{2}}\bW\bx_{0i}){\Bigg|}\mathrm{d}\bt\\
&\quad\quad+\max_{i\in[n]}\left|\frac{\pi(\bW\rho_n\halfpower\bx_{0i})}{d_{in}}-\det(2\pi\bG_{0in}\inverse)\invhalfpower\right|\int\left(1+\|\bt\|_2^\alpha\right)e^{-\frac{1}{2}\bt\transpose\bG_{0in}\bt}\mathrm{d}\bt.
\end{align*}
Since \eqref{equation:sufficient_condition_tvm} implies that $\max_{i\in[n]}|d_{in}-\det(2\pi\bG_{0in}\inverse)\halfpower\pi(\bW\rho_n\halfpower\bx_{0i})|=o_{\prob_0}(1)$ (by taking $\alpha=0$), it can be seen that \eqref{equation:sufficient_condition_tvm} implies that the two terms on the right hand side of the previous display are $o_{\prob_0}(1)$. Hence, we are left with establishing \eqref{equation:sufficient_condition_tvm}.

\noindent
Let $\{\eta_n\}_{n=1}^\infty$ be a sequence to be determined later with $0<\eta_n\to\infty$ and consider the following partition of $\mathbb{R}^d$:
\[
\calA_1=\{\bt\in\widehat\Theta_{in}:\|\bt\|_2\leq\eta_n\}, \qquad \calA_2=\{\bt\in\widehat\Theta_{in}:\|\bt\|_2>\eta_n\}, \qquad \calA_3=\widehat\Theta_{in}^c.
\]
We first consider the integral of \eqref{equation:sufficient_condition_tvm} over $\calA_3$. 
By definition of $\mathbbm{1}(\bt\in\widehat{\Theta}_{in})$, the integral over $\calA_3$ can be bounded by
\begin{equation}\label{equation:suffcond_tvm_A3}
\begin{split}
&
\max_{i\in[n]}\int_{\calA_3}\left(1+\|\bt\|_2^\alpha\right)e^{-\frac{1}{2}\bt\transpose\bG_{0in}\bt}\pi(\rho_n^{\frac{1}{2}}\bW\bx_{0i})\mathrm{d}\bt\\
&\quad\leq\int_{\calA_3}\left(1+\|\bt\|_2^\alpha\right)e^{-\min_{i\in[n]}\lambda_d(\bG_{0in})\|\bt\|_2^2/2}\pi(\rho_n\halfpower\bW\bx_{0i})\mathrm{d}\bt\\
&\quad\leq\int_{\calA_3}\left(1+\|\bt\|_2^\alpha\right)e^{-\lambda\|\bt\|_2^2/2}\pi(\rho_n\halfpower\bW\bx_{0i})\mathrm{d}\bt \to 0,
\end{split}
\end{equation}
since $\calA_3$ is shrinking to empty set and $\min_{i\in[n]}(\bG_{0in})\geq\lambda$ has been shown in the proof of Theorem \ref{theorem:asymptotic_properties_of_MSLE} (see diaplay \eqref{eqn:finite_pd_fisher_info}).
We next consider the integral of \eqref{equation:sufficient_condition_tvm} over $\calA_2$. Define the event
\[
\calE_{2n} =\left\{\bA:\max_{i\in[n]}\max_{\|\bx_i\|_2\leq1}\bs\transpose\frac{\partial^2\widetilde M_{in}}{\partial\bx_i\partial\bx_i\transpose}(\bx_i)\bs\leq-\frac{\lambda}{2}\|\bs\|_2^2\quad\forall\bs\in\mathbb{R}^d\right\}.
\]
Note that by Lemma \ref{lemma:some_frequently_used_results}, Theorem 5.2 in \cite{lei-rinaldo-2015}, and Weyl's inequality, with probability at least $1-n^{-c}$,
\begin{align*}
&\min_{i\in[n]}\min_{\|\bx_i\|_2\leq1}\bs\transpose\left(-\frac{\partial^2\widetilde M_{in}}{\partial\bx_i\partial\bx_i\transpose}(\bx_i)\right)\bs = \min_{i\in[n]}\min_{\|\bx_i\|_2\leq1}\bs\transpose\left(\frac{1}{n}\sum_{j=1}^n\left\{\frac{1}{\widetilde p_{ij}}+\frac{1-A_{ij}}{(1-\bx_i\transpose\widetilde\bx_j)^2}\right\}\widetilde\bx_j\widetilde\bx_j\transpose\right)\bs\\
&\quad\geq \frac{1}{\max_{i,j\in[n]}\widetilde p_{ij}}\frac{1}{n}\sum_{j=1}^n\bs\transpose\widetilde\bx_j\widetilde\bx_j\transpose\bs \geq \frac{1}{n\rho_n}\bs\transpose\widetilde\bX\transpose\widetilde\bX\bs \geq \frac{1}{n\rho_n}\lambda_d(\bA)\|\bs\|_2^2\geq\frac{\lambda}{2}\|\bs\|_2^2.
\end{align*}
This shows that $\prob_0(\calE_{2n})\geq1-n^{-c}$ for all $n\geq N_{c,\delta,\lambda}$. 
By Taylor's expansion, for any $\bt\in\widehat{\Theta}_{in}$, we have
\begin{equation}
\label{eqn:surrogate_likelihood_Taylor_expansion}
\begin{aligned}
n\widetilde M_{in}(\widehat\bx_i+\frac{\bW\bt}{\sqrt{n}}) - n\widetilde M_{in}(\widehat\bx_i)    
& = \frac{1}{2}\bt\transpose\bW
\frac{\partial^2\widetilde M_{in}}{\partial\bx_i\partial\bx_i\transpose}(\bar\bx_i)
\bW\bt,
\end{aligned}
\end{equation}
where $\bar{\bx}_i = \widehat{\bx}_i + \theta_i\bW\bt/\sqrt{n}$ for some $\theta_i\in [0, 1]$ because the gradient of $\widetilde{M}_{in}$ evaluated at $\bx_i = \widehat{\bx}_i$ is zero by definition of the maximum surrogate likelihood estimator $\widehat{\bx}_i$. Over this event, the integral of \eqref{equation:sufficient_condition_tvm} over $\calA_2$ can be upper bounded by
\begin{align*}
&\max_{i\in[n]}\int_{\calA_2}\left(1+\|\bt\|_2^\alpha\right)\exp\left\{\frac{1}{2}\bt\transpose\bW\transpose\frac{\partial^2\widetilde M_{in}}{\partial\bx_i\partial\bx_i\transpose}(\bar\bx_i)\bW\bt\right\}\pi(\widehat\bx_i+\frac{\bW\bt}{\sqrt{n}})\mathrm{d}\bt\\
&\quad\quad + \max_{i\in[n]}\int_{\calA_2}\left(1+\|\bt\|_2^\alpha\right)e^{-\bt\transpose\bG_{0in}\bt/2}\pi(\rho_n\halfpower\bW\bx_{0i})\mathrm{d}\bt\\
&\quad\leq C\int_{\calA_2}\left(1+\|\bt\|_2^\alpha\right)\exp\left\{\max_{i\in[n]}\max_{\|\bx_i\|_2\leq1}\frac{1}{2}\bt\transpose\bW\transpose\frac{\partial^2\widetilde M_{in}}{\partial\bx_i\partial\bx_i\transpose}(\bar\bx_i)\bW\bt\right\}\mathrm{d}\bt\\
&\quad\quad + \max_{i\in[n]}C\int_{\calA_2}\left(1+\|\bt\|_2^\alpha\right)e^{-\bt\transpose\bG_{0in}\bt/2}\mathrm{d}\bt\\
&\quad\leq 2C\int_{\|\bt\|_2>\eta_n}\left(1+\|\bt\|_2^\alpha\right)e^{-\lambda\|\bt\|_2^2/4}\mathrm{d}\bt.
\end{align*}
Denote the last line of the above display by $\epsilon_{2n}$, then $\epsilon_{2n}\to0$ because $\eta_n\to\infty$. It follows that
\begin{align*}
&\prob_0{\Bigg\{}\max_{i\in[n]}\int_{\calA_2}\left(1+\|\bt\|_2^\alpha\right){\Bigg|}\exp\left\{n\widetilde M_{in}(\widehat\bx_i+\frac{\bW\bt}{\sqrt{n}}) - n\widetilde M_{in}(\widehat\bx_i)\right\}\pi(\widehat\bx_i+\frac{\bW\bt}{\sqrt{n}})\indicator(\bt\in\widehat\Theta_{in})\\
&\quad\quad- e^{-\frac{1}{2}\bt\transpose\bG_{0in}\bt}\pi(\rho_n^{\frac{1}{2}}\bW\bx_{0i}){\Bigg|}\mathrm{d}\bt \geq \epsilon_{2n}{\Bigg\}}\leq n^{-c}
\end{align*}
for all $n\geq N_{c,\delta,\lambda}$.
Hence,
\begin{equation}\label{equation:suffcond_tvm_A2}
\begin{aligned}
&\max_{i\in[n]}\int_{\calA_2}\left(1+\|\bt\|_2^\alpha\right){\Bigg|}\exp\left\{n\widetilde M_{in}(\widehat\bx_i+\frac{\bW\bt}{\sqrt{n}}) - n\widetilde M_{in}(\widehat\bx_i)\right\}\pi(\widehat\bx_i+\frac{\bW\bt}{\sqrt{n}})\indicator(\bt\in\widehat\Theta_{in})\\
&\quad\quad- e^{-\frac{1}{2}\bt\transpose\bG_{0in}\bt}\pi(\rho_n^{\frac{1}{2}}\bW\bx_{0i}){\Bigg|}\mathrm{d}\bt\overset{\prob_0}{\to}0.
\end{aligned}
\end{equation}
We next consider the integral of \eqref{equation:sufficient_condition_tvm} over $\calA_1$. Take $\eta_n=\min\{(n\rho_n/\log n)^{(1/8)},\sqrt{(\log n)/\rho_n}\}$. Recall that $\bt=\sqrt{n}\bW\transpose(\bx_i-\widehat\bx_i)$, and $\max_{i\in[n]}\|\bW\transpose\widehat\bx_i-\rho_n\halfpower\bx_{0i}\|_2\lesssim_{c,\delta,\lambda}\sqrt{\frac{\log n}{n\rho_n}}$ with probability at least $1-n^{-c}$ by Theorem \ref{theorem:asymptotic_properties_of_MSLE}. Then
\[
\max_{i\in[n]}\|\bW\transpose\bx_i-\rho_n\halfpower\bx_{0i}\|_2 \leq \max_{i\in[n]}\|\bW\transpose\widehat\bx_i-\rho_n\halfpower\bx_{0i}\|_2 + \max_{i\in[n]}\frac{\|\bt\|_2}{\sqrt{n}} \lesssim_{c,\delta,\lambda} \sqrt{\frac{\log n}{n\rho_n}}
\]
with probability at least $1-n^{-c}$ because $\eta_n/\sqrt{n}\leq \sqrt{(\log n)/(n\rho_n)}$, which also implies that there exists a constant $C_{c,\delta,\lambda} > 0$ (possibly depending on $c,\delta,\lambda$), such that 
\[
\{\bx_i:\|\bt\|_2\leq\eta_n\}\subset\left\{\bx_i:\|\bW\transpose\bx_i-\rho_n\halfpower\bx_{0i}\|_2\leq C_{c,\delta,\lambda} \sqrt{\frac{\log n}{n\rho_n}}\right\}
\]
with probability at least $1-n^{-c}$. Define the event
\begin{equation*}
\begin{aligned}
&\calE_{1n} = \left\{\bA:\max_{i\in[n]}\sup_{\bx_i:\|\bt\|_2\leq\eta_n}\left\|\bW\transpose\frac{\partial^2\widetilde M_{in}}{\partial\bx_i\partial\bx_i\transpose}(\bx_i)\bW+\bG_{0in}\right\|_2\leq K_{c,\delta,\lambda}\sqrt{\frac{\log n}{n\rho_n}}\right\}\\
&\quad\quad\quad\cap \left\{\bA:\max_{i\in[n]}\|\bW\transpose\bx_i-\rho_n\halfpower\bx_{0i}\|_2\leq K_{c,\delta,\lambda} \sqrt{\frac{\log n}{n\rho_n}}\right\}.
\end{aligned}
\end{equation*}
for an appropriate constant $K_{c,\delta,\lambda}$ depending on $c,\delta,\lambda$. 
By Lemma \ref{lemma:concentration_of_hessian_matrices}, one can select $K_{c,\delta,\lambda}$ such that $\prob_0(\calE_{1n})\geq1-n^{-c}$ for all $n\geq N_{c,\delta,\lambda}$. Then over the event $\calE_{1n}$, by Taylor's expansion \eqref{eqn:surrogate_likelihood_Taylor_expansion} and the mean-value theorem applied to the exponential function, we have
\begin{align*}
&\max_{i\in[n]}\int_{\calA_1}\left(1+\|\bt\|_2^\alpha\right)\Bigg|\exp\left\{n\widetilde M_{in}(\widehat\bx_i+\frac{\bW\bt}{\sqrt{n}}) - n\widetilde M_{in}(\widehat\bx_i)\right\}\pi(\widehat\bx_i+\frac{\bW\bt}{\sqrt{n}})\indicator(\bt\in\widehat\Theta_{in}) \\
&\qquad\qquad\qquad\qquad\qquad\qquad\qquad - e^{-\frac{1}{2}\bt\transpose\bG_{0in}\bt}\pi(\rho_n^{\frac{1}{2}}\bW\bx_{0i})\Bigg|\mathrm{d}\bt\\
&\quad=\max_{i\in[n]}\int_{\calA_1}\left(1+\|\bt\|_2^\alpha\right)\Bigg|\exp\left\{\frac{1}{2}\bt\transpose\bW\transpose\frac{\partial^2\widetilde M_{in}}{\partial\bx_i\partial\bx_i\transpose}(\bar\bx_i)\bW\bt\right\}\pi(\widehat\bx_i+\frac{\bW\bt}{\sqrt{n}}) \\
&\qquad\qquad\qquad\qquad\qquad\qquad\qquad - e^{-\frac{1}{2}\bt\transpose\bG_{0in}\bt}\pi(\rho_n^{\frac{1}{2}}\bW\bx_{0i})\Bigg|\mathrm{d}\bt\\
&\quad=\max_{i\in[n]}\int_{\calA_1}\left(1+\|\bt\|_2^\alpha\right)\left|\exp\left\{\frac{1}{2}\bt\transpose\left(\bW\transpose\frac{\partial^2\widetilde M_{in}}{\partial\bx_i\partial\bx_i\transpose}(\bar\bx_i)\bW+\bG_{0in}\right)\bt\right\} -\frac{\pi(\rho_n\halfpower\bW\bx_{0i})}{\pi(\widehat\bx_i+\frac{\bW\bt}{\sqrt{n}})} \right|\\
&\qquad\qquad\qquad\qquad\qquad\qquad\qquad\times e^{-\frac{1}{2}\bt\transpose\bG_{0in}\bt}\pi\left(\widehat\bx_i+\frac{\bW\bt}{\sqrt{n}}\right)\mathrm{d}\bt\\
&\quad\leq\max_{i\in[n]}\int_{\calA_1}\left(1+\|\bt\|_2^\alpha\right){\Bigg\{}\left|\exp\left\{\frac{1}{2}\bt\transpose\left(\bW\transpose\frac{\partial^2\widetilde M_{in}}{\partial\bx_i\partial\bx_i\transpose}(\bar\bx_i)\bW+\bG_{0in}\right)\bt\right\}-1\right| \\
&\qquad\qquad\qquad\qquad\qquad\qquad\qquad+\left|1-\frac{\pi(\rho_n\halfpower\bW\bx_{0i})}{\pi(\widehat\bx_i+\frac{\bW\bt}{\sqrt{n}})} \right|{\Bigg\}} e^{-\frac{1}{2}\bt\transpose\bG_{0in}\bt}\pi\left(\widehat\bx_i+\frac{\bW\bt}{\sqrt{n}}\right)\mathrm{d}\bt\\
&\quad\leq\Bigg(\exp\left\{\frac{1}{2}K_{c,\delta,\lambda}\sqrt{\frac{\log n}{n\rho_n}}\eta_n^2\right\}\frac{1}{2}K_{c,\delta,\lambda}\sqrt{\frac{\log n}{n\rho_n}}\eta_n^2 \\
&\qquad\qquad\qquad+ \max_{i\in[n]}\sup_{\bx_i:\|\bW\transpose\bx_i-\rho_n\halfpower\bx_{0i}\|_2\lesssim_{c,\delta,\lambda}\sqrt{\frac{\log n}{n\rho_n}}}\left|1-\frac{\pi(\rho_n\halfpower\bW\bx_{0i})}{\pi(\bx_i)}\right|\Bigg) \times C\int e^{-\lambda\|\bt\|_2^2/2}\mathrm{d}\bt.
\end{align*}
Denote the last form of the above display by $\epsilon_{1n}$. It is obvious that $\exp\left\{\frac{1}{2}K_{c,\delta,\lambda}\sqrt{\frac{\log n}{n\rho_n}}\eta_n^2\right\}\to1$ (since $\eta_n=(n\rho_n/\log n)^\frac{1}{8}$). By the assumptions on $\pi(\bx_i)$, 
\[
\max_{i\in[n]}\sup_{\bx_i:\|\bW\transpose\bx_i-\rho_n\halfpower\bx_{0i}\|_2\lesssim_{c,\delta,\lambda}\sqrt{\frac{\log n}{n\rho_n}}}\left|1-\frac{\pi(\rho_n\halfpower\bW\bx_{0i})}{\pi(\bx_i)}\right|\to0.
\] 
It follows that $\epsilon_{1n}\to0$ as $n\to\infty$, and
\begin{align*}
&\prob_0{\Bigg\{}\max_{i\in[n]}\int_{\calA_1}\left(1+\|\bt\|_2^\alpha\right){\Bigg|}\exp\left\{n\widetilde M_{in}(\widehat\bx_i+\frac{\bW\bt}{\sqrt{n}}) - n\widetilde M_{in}(\widehat\bx_i)\right\}\pi(\widehat\bx_i+\frac{\bW\bt}{\sqrt{n}})\indicator(\bt\in\widehat\Theta_{in})\\
&\quad\quad\quad\quad- e^{-\frac{1}{2}\bt\transpose\bG_{0in}\bt}\pi(\rho_n^{\frac{1}{2}}\bW\bx_{0i}){\Bigg|}\mathrm{d}\bt\geq\epsilon_{1n}{\Bigg\}}\leq n^{-c},
\end{align*}
for all $n\geq N_{c,\delta,\lambda}$.
Hence,
\begin{equation}\label{equation:suffcond_tvm_A1}
\begin{aligned}
&\max_{i\in[n]}\int_{\calA_1}\left(1+\|\bt\|_2^\alpha\right){\Bigg|}\exp\left\{n\widetilde M_{in}(\widehat\bx_i+\frac{\bW\bt}{\sqrt{n}}) - n\widetilde M_{in}(\widehat\bx_i)\right\}\pi(\widehat\bx_i+\frac{\bW\bt}{\sqrt{n}})\indicator(\bt\in\widehat\Theta_{in})\\
&\quad\quad\quad\quad- e^{-\frac{1}{2}\bt\transpose\bG_{0in}\bt}\pi(\rho_n^{\frac{1}{2}}\bW\bx_{0i}){\Bigg|}\mathrm{d}\bt\overset{\prob_0}{\to}0.
\end{aligned}
\end{equation}
The proof of \eqref{equation:sufficient_condition_tvm} is completed by combining \eqref{equation:suffcond_tvm_A3}, \eqref{equation:suffcond_tvm_A2}, and \eqref{equation:suffcond_tvm_A1}.
\end{proof}

\subsection{Proof of Corollary \ref{corollary:posterior_inference}}
\begin{proof}
We first show the convergence of the mean and covariance of $\widetilde\pi_{in}^*(\bt\mid\bA)$, which is a direct consequence of Theorem \ref{theorem:posterior_convergence_tvm}:
\begin{align*}
\max_{i\in[n]}\left\|\int\bt\widetilde\pi_{in}^*(\bt\mid\bA)\mathrm{d}\bt\right\|_2 &=\max_{i\in[n]}\left\|\int\bt\widetilde\pi_{in}^*(\bt\mid\bA)\mathrm{d}\bt - \int\bt\frac{e^{-\bt\transpose\bG_{0in}\bt/2}}{\det(2\pi\bG_{0in}\inverse)\halfpower}\mathrm{d}\bt\right\|_2\\
&\leq\max_{i\in[n]}\int\|\bt\|_2\left|\widetilde\pi_{in}^*(\bt\mid\bA) - \frac{e^{-\bt\transpose\bG_{0in}\bt/2}}{\det(2\pi\bG_{0in}\inverse)\halfpower}\right|\mathrm{d}\bt\overset{\prob_0}{\to}0,
\end{align*}
and
\begin{align*}
\max_{i\in[n]}\left\|\int\bt\bt\transpose\widetilde\pi_{in}^*(\bt\mid\bA)\mathrm{d}\bt - \bG_{0in}\inverse\right\|_2 &=\max_{i\in[n]}\left\|\int\bt\bt\transpose\widetilde\pi_{in}^*(\bt\mid\bA)\mathrm{d}\bt - \int\bt\bt\transpose\frac{e^{-\bt\transpose\bG_{0in}\bt/2}}{\det(2\pi\bG_{0in}\inverse)\halfpower}\mathrm{d}\bt\right\|_2\\
&\leq\max_{i\in[n]}\int\|\bt\|_2^2\left|\widetilde\pi_{in}^*(\bt\mid\bA) - \frac{e^{-\bt\transpose\bG_{0in}\bt/2}}{\det(2\pi\bG_{0in}\inverse)\halfpower}\right|\mathrm{d}\bt\overset{\prob_0}{\to}0.
\end{align*}
Now
\begin{align*}
\max_{i\in[n]}\|\sqrt{n}(\bx_i^*-\widehat\bx_i)\|_2 &= \max_{i\in[n]}\left\|\int\sqrt{n}(\bx_i-\widehat\bx_i)\widetilde\pi_{in}(\bx_i\mid\bA)\mathrm{d}\bx_i\right\|_2\\
&=\max_{i\in[n]}\left\|\int\bt\widetilde\pi_{in}^*(\bt\mid\bA)\mathrm{d}\bt\right\|_2\\
&=o_{\prob_0}(1),
\end{align*}
then by Theorem \ref{theorem:asymptotic_properties_of_MSLE} and Slutsky's Theorem, $\sqrt{n}\bG_{0in}\halfpower(\bW\transpose\bx_i^*-\rho_n\halfpower\bx_{0i})\overset{\calL}{\to}\mathrm{N}_d(\mathbf{0}_d,\bI_d)$.
Also,
\begin{align*}
&\max_{i\in[n]}\left\| n\bW\transpose\bSigma_{in}^*\bW - \bG_{0in}\inverse\right\|_2\\
&\quad= \max_{i\in[n]}\left\| \int n\bW\transpose(\bx_i-\bx_i^*)(\bx_i-\bx_i^*)\transpose\bW\widetilde\pi_{in}(\bx_i\mid\bA)\mathrm{d}\bx_i - \bG_{0in}\inverse\right\|_2\\
&\quad= \max_{i\in[n]}\left\| \int n\bW\transpose(\bx_i-\widehat\bx_i + \widehat\bx_i-\bx_i^*)(\bx_i-\widehat\bx_i + \widehat\bx_i-\bx_i^*)\transpose\bW\widetilde\pi_{in}(\bx_i\mid\bA)\mathrm{d}\bx_i - \bG_{0in}\inverse\right\|_2\\
&\quad\leq \max_{i\in[n]}\left\|\int\bt\bt\transpose\widetilde\pi_{in}^*(\bt\mid\bA)\mathrm{d}\bt - \bG_{0in}\inverse\right\|_2 + o_{\prob_0}(1)\\
&\quad=o_{\prob_0}(1).
\end{align*}
Note that $\bG_{0in}$ is finite and positive definite. By continuous mapping theorem,
\[
(\rho_n\halfpower\bW\bx_i-\bx_i^*)\transpose(\bSigma_{in}^*)\inverse(\rho_n\halfpower\bW\bx_i-\bx_i^*) \overset{\calL}{\to} \chi^2_d,
\]
so $\prob_0\{(\rho_n\halfpower\bW\bx_i-\bx_i^*)\transpose(\bSigma_{in}^*)\inverse(\rho_n\halfpower\bW\bx_i-\bx_i^*)\leq q_{1-\alpha}\} \to 1-\alpha$. 

We now focus on the last assertion. By the previous proof, we know that $\max_{i\in [n]}\|\bx^* - \widehat{\bx}_i\|_2^2 = o_{\prob_0}(1/n)$. It follows directly that 
\[
\|\bX^* - \widehat{\bX}\|_{\mathrm{F}}^2 = \sum_{i = 1}^n\|\bx^* - \widehat{\bx}_i\|_2^2\leq n\max_{i\in [n]}\|\bx^* - \widehat{\bx}_i\|_2^2 = o_{\prob_0}(1).
\]
Therefore, by Theorem \ref{theorem:asymptotic_properties_of_MSLE} and Cauchy--Schwarz inequality, we have
\begin{align*}
\|\bX^*\bW - \rho_n^{1/2}\bX_0\|_{\mathrm{F}}^2
& = \|\widehat{\bX}\bW - \rho_n^{1/2}\bX_0\|_{\mathrm{F}}^2 + \|\bX^*\bW - \widehat{\bX}\bW\|_{\mathrm{F}}^2\\
&\quad + 2\left\langle
\widehat{\bX}\bW - \rho_n^{1/2}\bX_0, \bX^*\bW - \widehat{\bX}\bW
\right\rangle_{\mathrm{F}}\\
& = \frac{1}{n}\sum_{i = 1}^n\mathrm{tr}(\bG_{0in}^{-1}) + o_{\prob_0}(1) + O\left(\|\widehat{\bX}\bW - \rho_n^{1/2}\bX_0\|_{\mathrm{F}}\|\bX^*\bW - \widehat{\bX}\bW\|_{\mathrm{F}}\right)\\
& = \frac{1}{n}\sum_{i = 1}^n\mathrm{tr}(\bG_{0in}^{-1}) + o_{\prob_0}(1),
\end{align*}
where $\langle\cdot,\cdot\rangle_{\mathrm{F}}$ denotes the Frobenius inner product between matrices. The proof is thus completed.
\end{proof}

\section{Proof of the Convergence of the Stochastic Gradient Descent}

\begin{lemma}[Lemma A.5 in \citealp{mairal-2013-sgd}]
\label{lemma:Lemma_A5_NeurIPS2013}
Let $(a_t)_{t\geq1},(b_t)_{t\geq1}$ be two non-negative real sequences. Assume that $\sum_{t = 1}^\infty a_tb_t$ converges and $\sum_{t = 1}^\infty a_t$ diverges, and $|b_{t + 1} - b_t| \leq K a_t$ for some constant $K \geq 0$. Then $b_t$ converges to 0. 
\end{lemma}

\begin{lemma}[Lemma 2 in \citealp{li-orabona-2019}]
\label{lemma:Lemma_2_LiOrabona2019}
Let $a_0>0$, $a_i\geq0$, $i=1,\ldots,T$ and $\beta>1$. Then $\sum_{t=1}^T\frac{a_t}{(a_0+\sum_{i=1}^ta_i)^\beta} < \frac{1}{(\beta-1)a_0^{\beta-1}}$.
\end{lemma}

\begin{lemma}[Lemma 3 in \citealp{li-orabona-2019}]
\label{lemma:Lemma_3_LiOrabona2019}
Let $f:\calX\subset\mathbb{R}^d\to\mathbb{R}$ be twice continuously differentiable whose minimum is attained at $\bx = \bx^*$ and suppose there exists a constant $L > 0$, such that for all $\bx,\by\in\calX$,
\[
\left\|\frac{\partial f}{\partial\bx}(\bx) - \frac{\partial f}{\partial\bx}(\by)\right\|_2\leq L\|\bx - \by\|_2. 
\]
Suppose $\bg(\bx, \bz)$ is a function of a random vector $\bz$, such that $\expect_z \bg(\bx, \bz) = \partial f(\bx)/\partial\bx$. Let $(\bz_t)_{t\geq 1}$ be a sequence of independent and identically distributed (i.i.d.) copies of $\bz$. 
Consider a sequence of iterates $\bx^{(t)}$ generated by
\[
\bx^{(t + 1)} = \bx^{(t)} - \bH_t \bg(\bx^{(t)}, \bz_t),
\]
where $\bH_t\in\mathbb{R}^{d\times d}$ is a step-size matrix for the $t$th iteration. Then the sequence $(\bx^{(t)})_{t\geq 1}$ satisfies the following inequality:
\begin{align*}
&\expect_{\bz_1,\ldots,\bz_N}\left[\sum_{t = 1}^N\left\langle \frac{\partial f}{\partial\bx}(\bx^{(t)}), 
\bH_t\frac{\partial f}{\partial\bx}(\bx^{(t)})
\right\rangle\right]\\
&\quad\leq f(\bx^{(1)}) - f(\bx^*) + \frac{L}{2}
\expect_{\bz_1,\ldots,\bz_N}\left\{\sum_{t = 1}^N\|\bH_t\bg(\bx^{(t)}, \bz_t)\|^2\right\}.
\end{align*}
\end{lemma}

\begin{proof}\textbf{of Theorem \ref{theorem:SGD_MSLE_convergence}.}
The proof is similar to Theorem 1 in \cite{li-orabona-2019}, with some slight modifications. In the setting here, the expectation is taken with respect to the randomness of the stochastic gradient descent conditioned on the adjacency matrix, that is, the data and the ASE are viewed as deterministic. Here, we suppress the subscript $i\in[n]$ and use $\bx^{(t)}$ to denote the $t$th iterate in the optimization, and $\widehat\bx$ the maximizer of the average surrogate log-likelihood function $\widetilde{M}_{in}(\bx):=(1/n)\widetilde{\ell}_{in}(\bx)$.

For the surrogate log-likelihood function, by the computation of the gradient and Hessian of $\widetilde{M}_{in}$ in the proof of Theorem \ref{theorem:asymptotic_properties_of_MSLE}, they are bounded over $\{\bx_i:\|\bx_i\|_2\leq 1\}$ when $\max_j\|\widetilde{\bx}_j\|_2<1$.
So both $\widetilde M_{in}(\bx)$ and its gradient are Lipschitz in $\{\bx\in\mathrm{R}^d:\|\bx\|_2\leq1\}$ by the mean value theorem. Let $C_1$ and $C_2$ be the Lipschitz constants for $\widetilde M_{in}(\bx)$ and its gradient, respectively. In the context of Section \ref{sub:MSLE_computation}, the random vector $\bz_t$ corresponds to the randomly generated indices $(j_1^{(t)},\ldots,j_s^{(t)})$ in a single iteration of the mini-batch SGD algorithm, and $\bg(\bx^{(t)}, \bz_t)$ takes the form
\[
\bg(\bx^{(t)}, \bz_t) = \frac{1}{s}\sum_{k = 1}^s\frac{\partial m_{i}}{\partial\bx}(\bx^{(t)}, j_k^{(t)}).
\]
It is clear that $\expect_{\bz_t}\bg(\bx^{(t)}, \bz_t)$ coincides with the gradient of $\widetilde{M}_{in}(\bx^{(t)})$. 
Also, for the stochastic gradient $\bg(\bx^{(t)},z_t)$, it is easy to see that $\|\bg(\bx^{(t)},\bz_t)-\nabla\widetilde M_{in}(\bx^{(t)})\|_2\leq C_3$ for all $\bx^{(t)}\in B(\mathbf{0}_d,1)$.

Observe that
\begin{align*}
\sum_{t=1}^\infty \|\alpha_t\bg(\bx^{(t)},z_t)\|_2^2 &=  \sum_{t=1}^\infty\alpha_{t+1}^2\|\bg(\bx^{(t)},z_t)\|_2^2 +  \sum_{t=1}^\infty(\alpha_t^2-\alpha_{t+1}^2)\|\bg(\bx^{(t)},z_t)\|_2^2 \\
&\leq \frac{a_0^2}{2\epsilon b_0^{2\epsilon}} + \max_{t\geq1}\|\bg(\bx^{(t)},z_t)\|_2^2  \sum_{t=1}^\infty(\alpha_t^2-\alpha_{t+1}^2)\\
&\leq \frac{a_0^2}{2\epsilon b_0^{2\epsilon}} + \max_{t\geq1}\|\bg(\bx^{(t)},z_t)\|_2^2  \alpha_1^2\\
&\leq \frac{a_0^2}{2\epsilon b_0^{2\epsilon}} + 2\alpha_1^2\max_{t\geq1}\left(\|\|\nabla\widetilde M_{in}(\bx^{(t)})\|_2^2 + \|\nabla\widetilde M_{in}(\bx^{(t)}) - \bg(\bx^{(t)},z_t)\|_2^2\right)\\
&\leq \frac{a_0^2}{2\epsilon b_0^{2\epsilon}} + \frac{2a_0^2}{ b_0^{1+2\epsilon}}(C_1^2+C_3^2) < \infty,
\end{align*}
where in the first inequality we have used Lemma \ref{lemma:Lemma_2_LiOrabona2019}, in the third one the elementary inequality $\|\bx+\by\|_2^2 \leq 2\|\bx\|_2^2+2\|\by\|_2^2$.
Therefore, for any $m\in\mathbb{N}_+$, by Cauchy--Schwarz inequality, we have
\begin{align*}
\left\|\bx^{(N+m)}-\bx^{(N)}\right\|_2^2 &= \left\|\sum_{t=N}^{N+m-1}\bx^{(t+1)}-\bx^{(t)}\right\|_2^2 \leq \sum_{t=N}^{N+m-1}\left\|\bx^{(t+1)}-\bx^{(t)}\right\|_2^2 \\
&\leq\sum_{t=N}^{N+m-1}\left\|\alpha_t\bg(\bx^{(t)},z_t)\right\|_2^2,
\end{align*}
and the previous infinite sum being finite implies that $\lim_{N\to\infty}\left\|\bx^{(N+m)}-\bx^{(N)}\right\|_2=0$ a.s., that is, $\{\bx^{(t)}\}_t$ forms a Cauchy sequence, and thus converges to some point $\bx^*\in B(0,1)$ a.s.\@. Note that $\bx^*$ is a random variable with respect to the randomness of $z_t$. Next we need to show that $\bx^*$ is indeed the maximizer of the surrogate log-likelihood function.

By Lemma \ref{lemma:Lemma_3_LiOrabona2019}, taking the limit $T\to\infty$ and exchanging the expectation and the limits due to non-negative terms, we have
\[
\expect\left[ \sum_{t=1}^\infty  \alpha_t\left\|\nabla\widetilde M_{in}(\bx^{(t)})\right\|_2^2 \right] \leq \widetilde M_{in}(\bx^*) - \widetilde M_{in}(\bx_1) + \frac{C_2}{2}\expect\left[ \sum_{t=1}^\infty \|\alpha_t\bg(\bx^{(t)},z_t)\|_2^2 \right].
\]
With the right hand side being finite, we have
\[
\sum_{t=1}^\infty \alpha_t\left\|\nabla\widetilde M_{in}(\bx^{(t)})\right\|_2^2<\infty.
\]
Observe that by definition,
\[
\sup_{z_t,\bx^{(t)}}\left\|\alpha_t\bg(\bx^{(t)},z_t)\right\|_2 \leq \frac{a_0}{(b_0)^{1/2+\epsilon}}\sup_{z_t,\bx^{(t)}}\left\|\bg(\bx^{(t)},z_t)\right\|_2 < \infty,
\]
that is, the updating of the iterate is bounded. By assumption, the MSLE $\widehat\bx$ is in the interior of the feasible region. So there exists an integer $m^*$ such that for all $t\in\mathbb{N}_+$, the number of times that step-halving in the algorithm is called is no greater than $m^*$.
This implies that
\[
\frac{1}{m^*}a_0\left[b_0+\sum_{i=1}^{t-1}\|\bg(\bx^{(t)},z_t)\|_2^2\right]^{-(1/2+\epsilon)} \leq \alpha_t \leq a_0\left[b_0+\sum_{i=1}^{t-1}\|\bg(\bx^{(t)},z_t)\|_2^2\right]^{-(1/2+\epsilon)}
\]
for all $t\in\mathbb{N}_+$, which further implies that
\begin{align*}
\sum_{t=1}^\infty\alpha_t &\geq \frac{1}{m^*}\sum_{t=1}^\infty a_0\left[b_0+\sum_{i=1}^{t-1}\|\bg(\bx^{(t)},z_t)\|_2^2\right]^{-(1/2+\epsilon)} \\
&\geq \frac{1}{m^*}\sum_{t=1}^\infty a_0\left[b_0+2(t-1)(C_1^2+C_3^2)\right]^{-(1/2+\epsilon)} = \infty.
\end{align*}
Using the fact that both $\widetilde M_{in}(\bx)$ and $\nabla\widetilde M_{in}(\bx)$ are Lipschitz, we also have
\begin{align*}
&\left| \|\nabla\widetilde M_{in}(\bx_{t+1})\|_2^2 - \|\nabla\widetilde M_{in}(\bx^{(t)})\|_2^2 \right|\\
&\quad= \left( \|\nabla\widetilde M_{in}(\bx_{t+1})\|_2+\|\nabla\widetilde M_{in}(\bx^{(t)})\|_2 \right) \cdot \left| \|\nabla\widetilde M_{in}(\bx_{t+1})\|_2-\|\nabla\widetilde M_{in}(\bx^{(t)})\|_2 \right| \\
&\quad\leq 2C_1C_2\|\bx_{t+1}-\bx_{t}\|_2 
= 2C_1C_2\|\alpha_{t}\bg(\bx^{(t)},z_t)\|_2 
\leq 2C_1C_2(C_1+C_3) \alpha_t.
\end{align*}
Hence, we can use Lemma \ref{lemma:Lemma_A5_NeurIPS2013} to obtain that $\lim_{t\to\infty}\|\nabla\widetilde M_{in}(\bx^{(t)})\|_2 = 0$ a.s.. The continuity of $\nabla M_{in}(\bx)$ implies that $\bx^{(t)}\to\widehat\bx$ a.s..
\end{proof}

\section{Additional implementation details}
\subsection{Additional details of the algorithms}
This subsection provides
the detailed Metropolis--Hastings sampler for computing the joint posterior distribution $\pi_{n}(\bX\mid\bA)$ using the surrogate likelihood function. 
For each $i\in[n]$, we use the normal random walk truncated in the unit ball as the proposal distribution, with the covariance matrix being the inverse of
\[
n\widetilde\bG_{in} = \sum\limits_{j=1}^n \frac{\widetilde{\bx}_j \widetilde{\bx}_j\transpose}{\widetilde{\bx}_i\transpose \widetilde{\bx}_j (1 - \widetilde{\bx}_i\transpose\widetilde{\bx}_j)}.
\]
The above covariance matrix is the plug-in estimator of the asymptotic covariance matrix of the Bernstein--von Mises limit distribution. Below, we provide the detailed Metropolis--Hastings sampler in 
the algorithm below. The computation of the posterior distribution of the entire latent position matrix $\bX$ can be done by a parallelization over $i\in [n]$. 

\begin{algorithm}[h]
\caption{Metropolis--Hastings sampler for computing the posterior distribution of $\bX$.}
\label{alg:MCMC}
\begin{algorithmic}[1]
\State \textbf{Input:} The adjacency matrix $\bA = [A_{ij}]_{n\times n}$;\\
   \quad\ \  The embedding dimension $d$;\\
   \quad\ \  The tuning parameter $\sigma$;\\
   \quad\ \  Number of burn-in iterations $B$;\\
   \quad\ \  Number of post-burn-in samples $n_{\mathrm{mc}}$;\\
   \quad\ \  Thinning size $b$.
\State Compute the spectral decomposition of the adjacency matrix
\[
\bA = \sum_{i=1}^n \widehat\lambda_i \widehat\bu_i \widehat\bu_j\transpose,
\]
where $|\widehat{\lambda}_1| \geq |\widehat{\lambda}_2| \geq \ldots \geq |\widehat{\lambda}_n|$, and $\widehat{\textbf{u}}_i\transpose \widehat{\textbf{u}}_j = \mathbbm{1}(i=j)$ for all $i, j \in[n]$.

\State Compute the adjacency spectral embedding:
\[
  \widetilde{\mathbf{X}} = \widehat{\mathbf{X}}^{\mathrm{ASE}} = [\widehat{\mathbf{u}}_1, \ldots, \widehat{\mathbf{u}}_d] \cdot \mathrm{diag}(|\widehat{\lambda}_1|^{1/2}, \ldots, |\widehat{\lambda}_d|^{1/2}),
\]
and write $\widetilde{\mathbf{X}} = [\widetilde{\mathbf{x}}_1, \ldots, \widetilde{\mathbf{x}}_n]\transpose \in \mathbb{R}^{n\times d}$. Let = $\widetilde{p}_{ij} = \widetilde{\mathbf{x}}_i\transpose\widetilde{\mathbf{x}}_j$ for all $i,j\in[n]$.

\State For $i = 1,2,\ldots,n$

\State \quad Initialize $\bx_i^{(1)} = \widetilde{\bx}_i$.

\State \quad For $t = 2$ to $B + n_{\mathrm{mc}}\times b$

\State \quad  \quad Generate $\bx_i'\sim N\left(\bx_i^{(t)}, \sigma^2\widetilde{\bG}_{in}^{-1}/n\right)\cdot\mathbbm{1}(||\bx_i||_2<1)$. 

\State \quad  \quad Generate $\alpha_t\sim\mathrm{Unif}(0,1)$.

\State \quad\quad If $\log\alpha_t < \widetilde{\ell}_{in}(\bx_i') - \widetilde{\ell}_{in}(\bx_i^{(t)}) + \log\pi(\bx_i') - \log\pi(\bx_i^{(t)})$

\State \quad\quad\quad Set $\bx_i^{(t+1)} \gets \bx_i'$;

\State \quad\quad Else

\State \quad\quad\quad Set $\bx_i^{(t+1)} \gets \bx_i^{(t)}$.

\State \quad\quad End If

\State \quad End For

\State End For

\State \textbf{Output: } $\bX^{(B + 1 + b\times N)}$ for $N = 1, 2, \ldots, \lceil (n_{\mathrm{mc}} - 1)/b\rceil$, where $\bX^{(t)} = [\bx_1^{(t)}, \ldots, \bx_n^{(t)}]\transpose$. 

\end{algorithmic}
\end{algorithm}

\subsection{Convergence diagnostics of the Metropolis--Hastings sampler}
In this subsection, we provide some convergence diagnostics of Metropolis--Hastings sampler.
Specifically, we choose one realization of the simulated data in the case of the stochastic block model with $d=2$ and $n=2000$ (Section \ref{sub:SBM_example} of the manuscript). The parameters of this random dot product graph are the entries of a $2000\times2$ matrix, so we get $2000\times2=4000$ Markov chains as the output of Metropolis--Hastings sampler. The total number of iterations in one Markov chain is $2000$, where we discard the first 2000 as burn-in and apply a thinning of 5 to the rest, resulting in a chain of length 200. To diagnose convergence, we use \texttt{coda::heidel.diag()} in R, which uses the Cramer--von Mises statistic to test the null hypothesis that the sampled values come from a stationary distribution.

Below, Fig.\ \ref{fig:simulation_mcmc_diag} presents the numerical diagnostics results. From the histogram of the $4000$ $p$-values from the output of \texttt{coda::heidel.diag()} applied to the $4000$ Markov chains, we see that there are very few $p$-values that are less than 0.05 (only 36 among the 4000 $p$-values in this trial). Furthermore, with different trials of Metropolis--Hastings sampler, the specific parameters which give the small $p$-values are different. So we can say that the occurrence of some small $p$-values is very likely due to the randomness in the data and in the Metropolis--Hastings sampler. A histogram of the accept rates from the Metropolis--Hastings algorithm of the 2000 vertices is provided as well.
To investigate more closely, the trace plot and auto-correlation function (ACF) plot of the second coordinate of the $808$th vertex which gives a $p$-value smaller than $0.05$ in this trial are provided. We can see that although it gives a small $p$-value, the trace plot and the ACF plot of the Metropolis--Hastings sample are not too abnormal.

\begin{figure}[htbp]
    \begin{subfigure}{0.5\linewidth}
    \includegraphics[width=\linewidth]{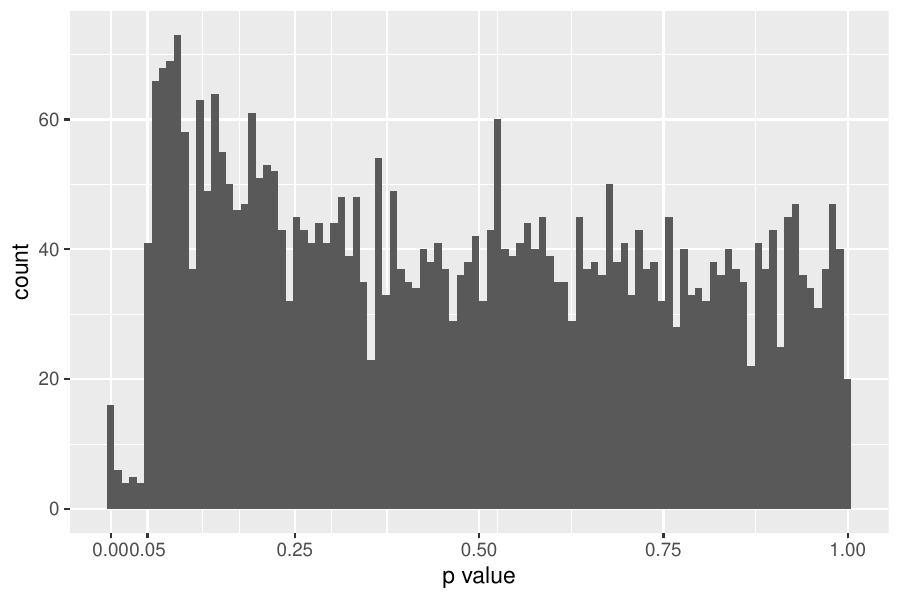}
    \end{subfigure}
    \begin{subfigure}{0.5\linewidth}
    \includegraphics[width=\linewidth]{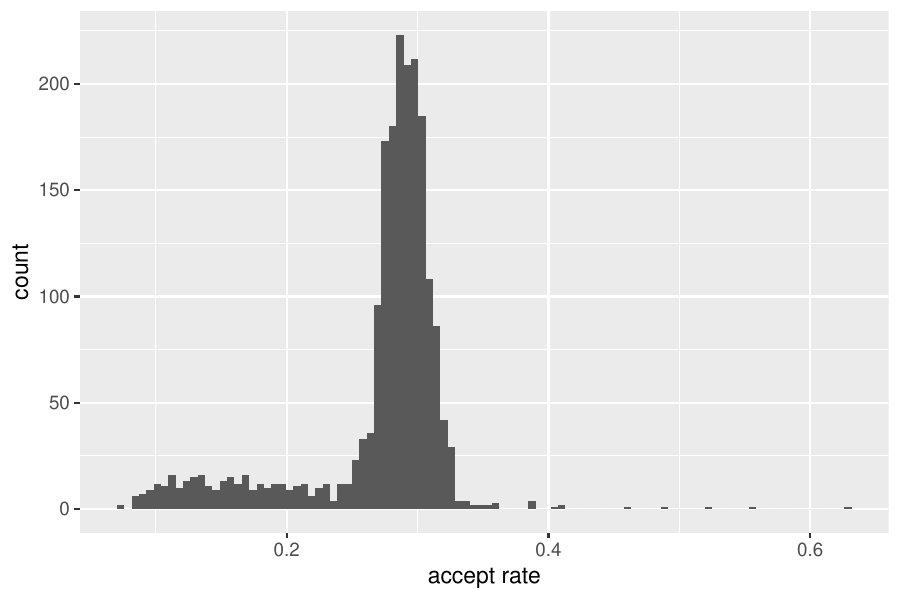}
    \end{subfigure}
    
    \begin{subfigure}{0.5\linewidth}
    \includegraphics[width=\linewidth]{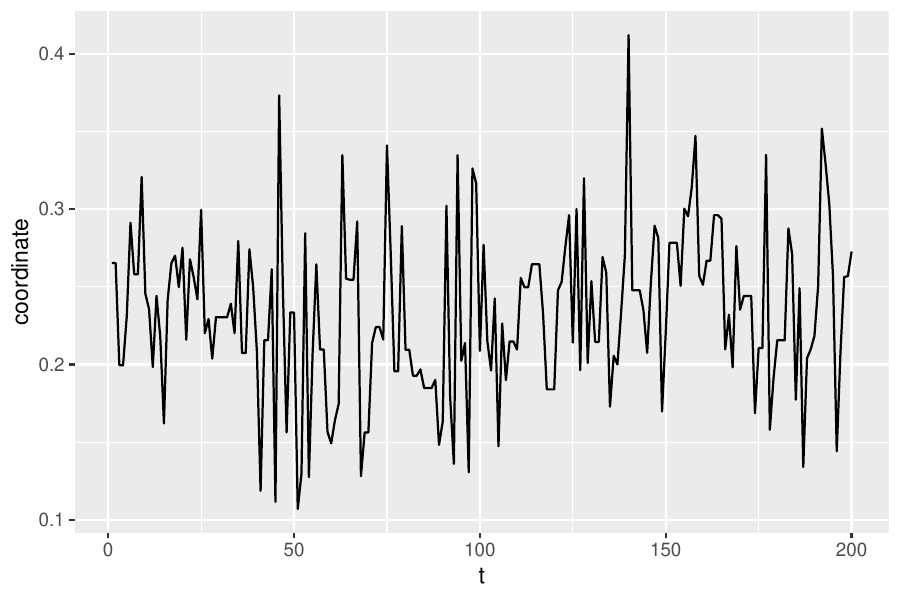}
    \end{subfigure}
    \begin{subfigure}{0.5\linewidth}
    \includegraphics[width=\linewidth]{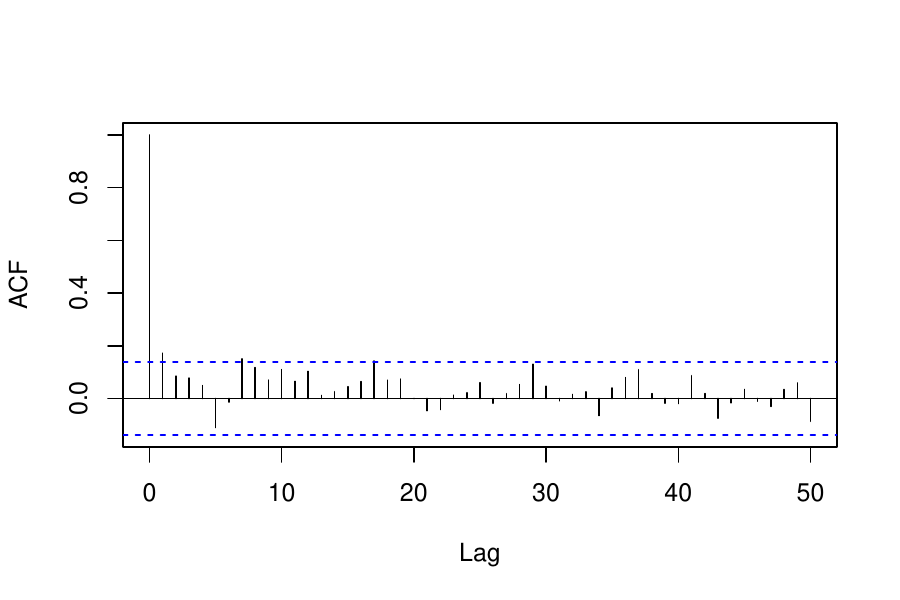}
    \end{subfigure}
    \caption{Convergence diagnostics for the simulation example in Section \ref{sub:SBM_example} of the manuscript. Top left panel: histogram of $4000$ p-values. Top right panel: histogram of $2000$ accept rates. Bottom left panel: Trace plot of a parameter whose Metropolis--Hastings sample gives a p-value less than 0.05. Bottom right panel: ACF plot of a parameter whose Metropolis--Hastings sample gives a p-value less than 0.05.}
    \label{fig:simulation_mcmc_diag}
\end{figure}

Next, we invectigate the convergence of the Metropolis--Hastings sampler in the Wikipedia graph dataset (Section \ref{subsection:wikipedia-graph} of the manuscript). For each $d$, there are $1382\times d$ parameters to estimate, so we get $1382\times d$ markov chains as the output of Metropolis--Hastings sampler. The total number of iterations in one Metropolis--Hastings sampler is $4000(2d+1)$, where we discard the first half as burn-in and apply a thinning of $4d$, resulting in a chain of length slightly more than 1000.

For $d = 1,\ldots,15$, the histograms of $1382$ accept rates and of $1382\times d$ p-values are provided in the upper and lower panel of Fig.\ \ref{fig:wikipedia_mcmc_diag_1}, respectively.
\begin{figure}[htbp]
    \centering{
    \includegraphics[width=\linewidth]{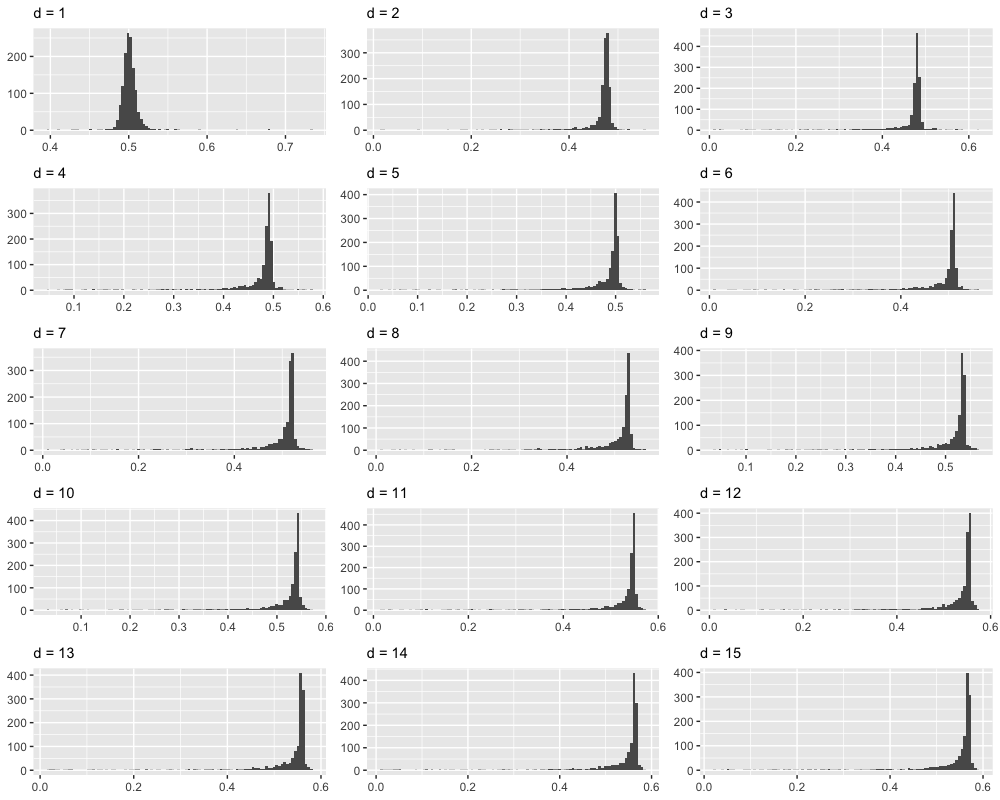}
    }
    \caption{Convergence diagnostics for the Wikipedia graph data example in Section \ref{subsection:wikipedia-graph} of the manuscript: Histograms of accept rates, where the horizontal axis represents accept rates and the vertical axis represents counts.}
\end{figure}
\begin{figure}[htbp]    
    \centering{
    \includegraphics[width=\linewidth]{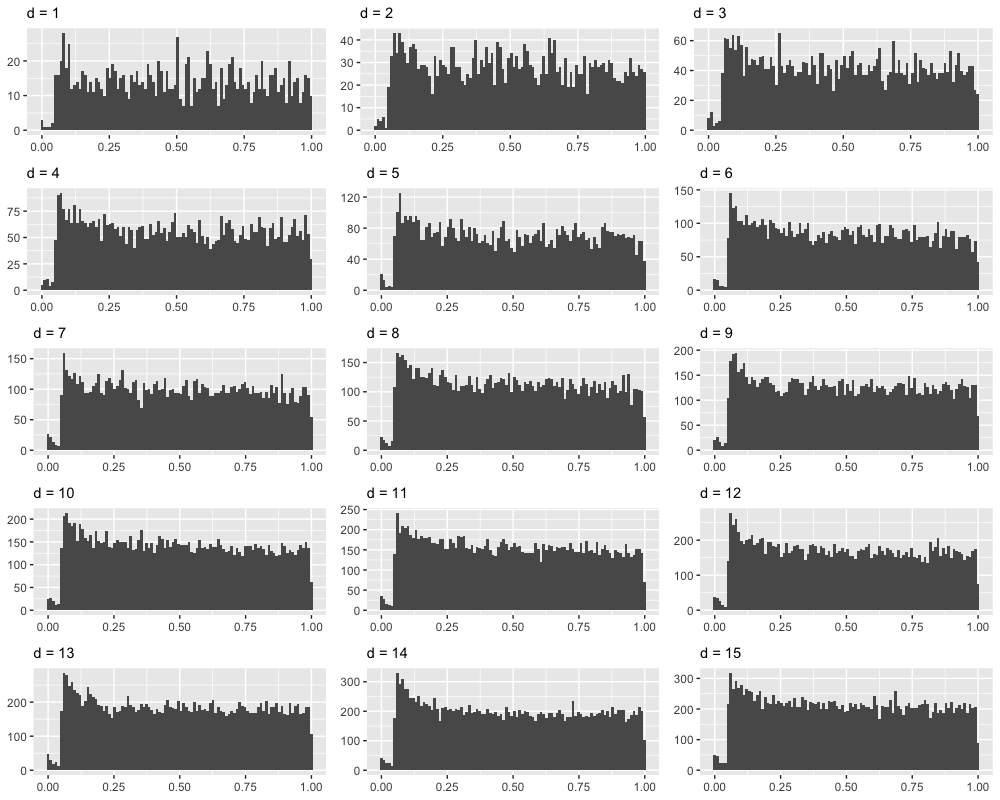}
    }
    \caption{Convergence diagnostics for the Wikipedia graph data example in Section \ref{subsection:wikipedia-graph} of the manuscript. Top panel: histograms of accept rates, where the horizontal axis represents accept rates and the vertical axis represents counts. Bottom panel: Histograms of p-values, where the horizontal axis represents p-values and the vertical axis represents counts.}
    \label{fig:wikipedia_mcmc_diag_1}
\end{figure}

To investigate more closely, the trace plots and autocorrelation function (ACF) plots of two chains which give p-values smaller than $0.05$ are provided, as in Fig.\ \ref{fig:wikipedia_mcmc_diag_2}.
\begin{figure}[htbp]
    \begin{subfigure}{0.5\linewidth}
    \includegraphics[width=\linewidth]{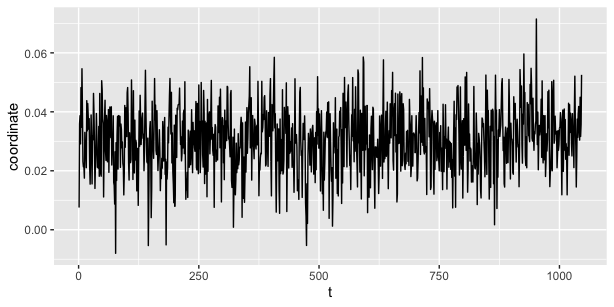}
    \end{subfigure}
    \begin{subfigure}{0.5\linewidth}
    \includegraphics[width=\linewidth]{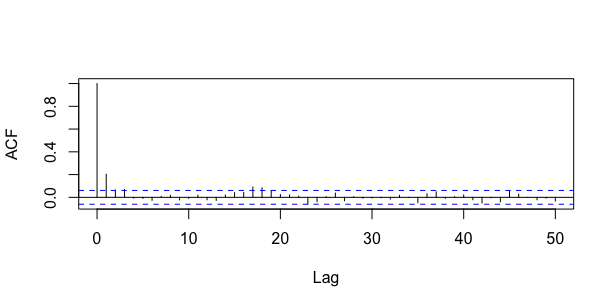}
    \end{subfigure}
    
    \begin{subfigure}{0.5\linewidth}
    \includegraphics[width=\linewidth]{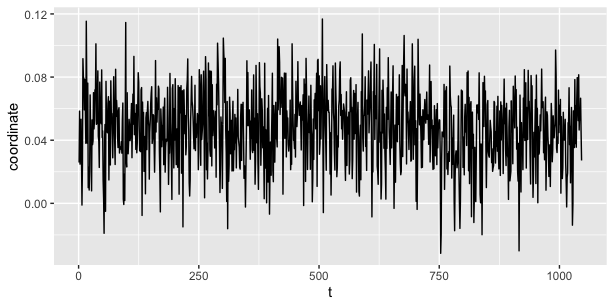}
    \end{subfigure}
    \begin{subfigure}{0.5\linewidth}
    \includegraphics[width=\linewidth]{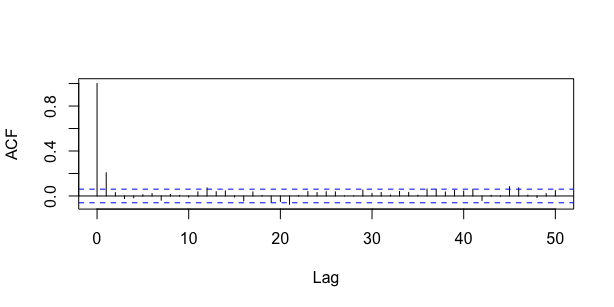}
    \end{subfigure}
    \caption{Convergence diagnostics for the Wikipedia graph data example in Section \ref{subsection:wikipedia-graph} of the manuscript. Top left panel: Trace plot of the Markov chain of the first coordinate of the $354th$ vertex with p-value = 0.0019, $d=11$. Top right panel: ACF plot of the Markov chain of the first coordinate of the $354th$ vertex, $d=11$. Bottom left panel: Trace plot of the Markov chain of the tenth coordinate of the $14th$ vertex with p-value = 0.0004, $d=11$. Bottom right panel: ACF plot of the Markov chain of the tenth coordinate of the $14th$ vertex, $d=11$.}
    \label{fig:wikipedia_mcmc_diag_2}
\end{figure}

\vskip 0.2in
\bibliography{reference}

\end{document}